\documentclass[11pt,letterpaper]{article}

\usepackage[square,sort,comma,numbers]{natbib}
\usepackage{amsmath}
\usepackage{amsthm}
\usepackage{amssymb}

\usepackage{algorithm}

\usepackage{algpseudocode}

\usepackage{color}
\usepackage{xcolor}

\usepackage{babel}

\usepackage{graphicx}
\usepackage{caption}
\usepackage{subcaption}
\usepackage{tikz}
\usepackage{wrapfig,epsfig}
\usepackage{psfrag}
\usepackage{epstopdf}

\usepackage{comment}

\usepackage[bookmarksnumbered=true]{hyperref}

\usepackage[margin=1in, letterpaper]{geometry}
\usepackage{setspace}

\usepackage{tabularx}
\usepackage{longtable}
\usepackage{multirow}
\usepackage{array}
\newcommand{\norm}[1]{\left\|#1\right\|}

\newcommand{\ip}[2]{\left\langle #1, #2 \right\rangle}

\newcommand{\wh}[1]{\widehat{#1}}

\newcommand{\op}{\mathrm{op}}
\newcommand{\HS}{\mathrm{HS}}

\newcommand{\dist}{\operatorname{dist}}

\newcommand{\Tr}{\mathsf{Tr}}

\DeclareMathOperator*{\E}{\mathbb{E}}

\newcommand{\supp}{\mathsf{supp}}

\newcommand{\rect}{\mathsf{rect}}
\newcommand{\sinc}{\mathsf{sinc}}
\newcommand{\comb}{\mathsf{Comb}}

\newcommand{\bott}{\mathsf{bottom}}
\newcommand{\Err}{\mathsf{Err}}

\DeclareMathOperator*{\argmin}{arg\,min}

\newcommand{\R}{\mathbb{R}}

\newcommand{\Z}{\mathbb{Z}}
\newcommand{\C}{\mathbb{C}}
\newcommand{\cC}{\mathcal{C}}

\newcommand{\bi}{\mathbf{i}}

\newtheorem{theorem}{Theorem}[section]
\newtheorem{definition}[theorem]{Definition}
\newtheorem{fact}[theorem]{Fact}
\newtheorem{lemma}[theorem]{Lemma}
\newtheorem{claim}[theorem]{Claim}

\newtheorem{question}[theorem]{Question}
\newtheorem{remark}[theorem]{Remark}

\newtheorem{corollary}[theorem]{Corollary}

\newenvironment{proofof}[1]{\bigskip \noindent {\it Proof of #1.}\quad }
{\qed\par\vskip 4mm\par}

\title{Learning Fourier-Sparse Signals and Multiband Signals}
\author{Dongrun Cai\thanks{\tt{cdr@mail.ustc.edu.cn}, University of Science and Technology of China, Hefei 230026, China.} \and Xue Chen\thanks{\tt{xuechen1989@ustc.edu.cn}, University of Science and Technology of China, Hefei 230026, China and Hefei National Laboratory, Hefei 230088, China. Supported by NSFC 62372424 and Quantum Science and Technology-National Science and Technology Major Project 2021ZD0302901.}
\and Xiaowei Shao \thanks{\tt{shaoxiaowei@mail.ustc.edu.cn}, University of Science and Technology of China, Hefei 230026, China }
\and Yile Wang\thanks{\tt{xortrue@mail.ustc.edu.cn}, University of Science and Technology of China, Hefei 230026, China.}
} \date{}

\begin{document}

\maketitle

\begin{abstract}

A classical problem in continuous sparse Fourier transforms, which dates back to Prony's work  in 1795, is to learn a $k$-Fourier-sparse signal $x(t):=\sum_{j=1}^k \alpha_j e^{2 \pi \bi f_j t}$ with arbitrary $k$ frequencies $f_1,\ldots,f_k \in \mathbb{R}$. We consider this learning problem within a fixed time window $t \in [-1,1]$ in the presence of bounded $L_2$ adversarial noise. Because the frequencies $f_1,\ldots,f_k$ can be arbitrarily close, accurately recovering individual frequencies is information-theoretically impossible even under exponentially small noise from Moitra (STOC 2015). This work therefore examines two fundamental questions: If an algorithm uses $\tilde{O}(k)$ frequencies to approximate $f_1,\ldots,f_k$ of $x$, what is the optimal covering error? And how can we efficiently learn the continuous signal $x(t)$ as a whole over the time window?

Our main result is an algorithm running in time $\tilde{O}(k)$ that outputs a list $L$ of $\tilde{O}(k)$ frequencies with covering error $O(k)$ --- meaning that for most frequencies $f_j$ of $x$, $L$ guarantees $\min_{\tilde{f} \in L} |f_j-\tilde{f}| = O(k)$ --- and an reconstruction $\tilde{x}(t) \approx x(t)$ over $[-1,1]$. When constrained to outputting $\tilde{O}(k)$ frequencies, this $O(k)$ covering error is nearly optimal: Chen and Price (ICALP 2019) demonstrated two indistinguishable $k$-Fourier-sparse signals separated by distance $\tilde{\Omega}(k^2)$ in the frequency domain, which implies an amortized lower bound of $\tilde{\Omega}(k)$ per frequency. Furthermore, our sample and time complexities are nearly optimal, dramatically improving upon the previous best algorithm by Song, Sun, Weinstein, and Zhang (FOCS 2023), which required $\tilde{O}(k^4)$ samples and $\tilde{O}(k^{4\omega})$ running time.

Conceptually, we introduce two key tools: (1) an explicit Haar-Fourier basis that efficiently yields sparse approximations of Fourier-sparse signals with arbitrary frequencies, and (2) a reduction from Fourier-sparse signals to multiband signals. Our algorithm leverages the Haar-Fourier basis in conjunction with the classical discrete sparse Fourier transform to reconstruct $x(t)$ over the time window in $\tilde{O}(k)$ time. On the other hand, this work analyzes multiband signals $z(t)$ whose Fourier transform $\widehat{z}(f)$ is supported on a finite number of intervals $I_1 \cup I_2 \cup \cdots \cup I_n$. We prove that any $k$-Fourier-sparse signal has a multiband approximation with $n \le k$ bands and total  support size $|I_1 \cup \cdots \cup I_n|=\tilde{O}(k^2)$, which is the cornerstone to prove our frequency covering bounds. Based on this connection, we also study how to learn a general multiband signal. Our main result here is an algorithm in time $O(n + |I_1| + \cdots + |I_n|)$ that recovers most bands among  $I_1,\ldots,I_n$. Unlike previous algorithms that assume $I_1,\ldots,I_n$ are given as prior knowledge to reconstruct $z(t)$ in the time window, our learning algorithm recovers the positions of these bands.

        \end{abstract}

\thispagestyle{empty}

\clearpage
\tableofcontents
\clearpage
\pagenumbering{arabic}

\section{Introduction}\label{sec:intro}
We consider classical problems of learning Fourier-sparse signals under noise in the continuous setting. Let $x(t):=\sum_{j=1}^k \alpha_j \cdot e^{2 \pi \bi f_j t}$ denote a $k$-Fourier-sparse signal with frequencies $f_1,\ldots,f_k$ and amplitudes $\alpha_1,\ldots,\alpha_k \in \C$. In this work, these $k$ frequencies $f_1,\ldots,f_k$ can be arbitrary real numbers in a given bandlimit $[-F,F]$. Many problems and applications of the continuous Fourier transform in engineering and computer science consider $x(t)$ in a fixed time window rather than in $(-\infty,+\infty)$. Thus, the basic problem is to learn $x(t)$ from a noisy observation $y(t):=x(t)+\eta(t)$ in the time window given the bandlimit $F$ and the sparsity $k$.

This work studies the recovery of $x(t)$ under \emph{adversarial} noise $\eta(t)$ with a bounded $L_2$ norm in the time window. Throughout this work, let $[-1,1]$ be the fixed time window and $\|z\|_{[-1,1]}:=(\int_{-1}^1 |z(t)|^2 \mathrm{d} t)^{1/2}$ denote the $L_2$ norm  in the window. This work assumes that the adversarial noise satisfies  $\|\eta\|^2_{[-1,1]} \le \epsilon \cdot \|x\|_{[-1,1]}^2$ for a small fixed constant $\epsilon$ and calls an algorithm efficient only if its time complexity is $(k \log F)^{O(1)}$. 

If the $k$ frequencies $f_1,\ldots,f_k$ in $x(t)$ are located on the discrete grid $\frac{1}{2} \cdot \mathbb{Z}$, $x(t)$ is periodic of length $2$. A long line of research on sparse Fourier transforms has studied efficient algorithms to learn $x(t)$ in this setting, including \cite{Man92,GGIMS,AGS,GMS,HIKP,IK}. 

For arbitrary real frequencies that are not multiples of $1/2$ (are ``off-grid''), this problem becomes much more challenging. There are information-theoretic lower bounds \cite{Moitra} indicating that recovering every frequency accurately is impossible even under exponentially small noise, i.e., $\|\eta\|_{[-1,1]} = 2^{-\Omega(k)} \cdot \|x\|_{[-1,1]}$. On the other hand, algorithms for restricted settings, such as the noiseless case and signals with well-separated frequencies $\underset{i \neq j}{\min} |f_i-f_j| = \Omega(1)$, have been extensively studied. In the noiseless setting, several classical methods (including Prony's method \cite{Prony}, Reed-Solomon decoding \cite{m69}, the matrix pencil algorithm \cite{BM86}, and the ESPRIT algorithm \cite{roy86}) can identify the frequencies. When these frequencies are separated, efficient and robust algorithms have been studied in sparse Fourier transforms \cite{BCGLS,PS15,SSWZ22}, compressed sensing \cite{HK15,CF14}, and super-resolution \cite{Moitra,DELZ24}.

For signals with arbitrary $k$ frequencies without any separation, Chen, Kane, Price, and Song \cite{CKPS17} studied the learning problem from a different perspective. Because it is impossible to learn each frequency accurately, they considered how to recover the $k$ frequencies within a bounded error and learn $x(t)$ as a whole signal in the time window. Specifically, their algorithm finds a list $L$ of $\tilde{O}(k)$\footnote{This work uses $\tilde{}$ to hide $(\log \frac{kF}{\epsilon})^{O(1)}$ factors.} frequencies efficiently such that most frequencies in $x$ are covered by some $\tilde{f} \in L$ within error $k^{O(1)}$. Based on this covering bound, they showed how to interpolate $x(t)$ by a sparse representation $\tilde{x} \approx x$ in the time window within $(k \log F)^{O(1)}$ time. Subsequent works \cite{CP19_colt, CP19_ICALP,  SSWZ23} improved the frequency error to $\tilde{O}(k^3)$ and found a reconstruction ion within $m=\tilde{O}(k^4)$ samples and $\tilde{O}(m^{\omega})$ time (for the matrix multiplication exponent $\omega$). 
There are many open questions left in the study of Fourier-sparse signals with arbitrary frequencies. While information-theoretic lower bounds \cite{Moitra,CP19_ICALP} rule out algorithms to learn each frequency accurately, a natural subsequent question is about the optimal error of frequency recovery. 
\begin{question}\label{ques1}
    What is the optimal error of learning $k$ arbitrary frequencies under noise? Are there algorithms that cover frequencies within the optimal error in time $\tilde{O}(k)$?
\end{question}
We remark Question~\ref{ques1} would constrain the output to be $\tilde{O}(k)$ frequencies; otherwise one can output more frequencies to reduce the error. Given the wide applications of Fourier-sparse signals in image processing like super-resolution and astronomy, identifying the optimal error of frequency recovery may improve other applications in engineering and science. However, for this fundamental question about the error of frequency recovery, there is still a gap between the best known upper bound $\tilde{O}(k^3)$ \cite{SSWZ23} and lower bounds \cite{Moitra,CP19_ICALP}. One lower bound in \cite{CP19_ICALP} shows two $k$-Fourier sparse signals $x$ and $z$ that are indistinguishable in the time window but within distance $\tilde{\Omega}(k^2)$ in the frequency domain. While this seems to imply the error is $\tilde{\Omega}(k^2)$, one might amortize this distance $\tilde{\Omega}(k^2)$ over $k$ frequencies for an error $\tilde{O}(k)$. Moreover, the algorithm of upper bound $\tilde{O}(k^3)$ \cite{SSWZ23} takes $\tilde{O}(k^2)$ samples and $\tilde{O}(k^2)$ time, not linear time $\tilde{O}(k)$.

Another intriguing question is to learn $x(t)$ in the time window efficiently. This is equivalent to reconstructing $x(t)$ in a fixed interval, which might find applications in audio compression and signal processing. The state-of-the-art is $m=\tilde{O}(k^4)$ samples and $\tilde{O}(m^{\omega})$ time \cite{CP19_colt, SSWZ23}. This leaves a large gap between the current bound and the information-theoretic lower bound $\Omega(k \log F)$ \cite{CKPS17}. However, in the discrete setting, the restricted isometry property (RIP) \cite{CRT06, RV08, HR16} bounded the sample complexity by $k \cdot (\log N)^{O(1)}$ for any discrete domain of size $N$. Later on, sparse discrete Fourier transforms \cite{HIKP, IK, K16} provided algorithms in time $k \cdot (\log N)^{O(1)}$. 
\begin{question}\label{ques2}
Back to the continuous setting, are there algorithms in time $\tilde{O}(k)$ to reconstruct $x(t)$ in the time window? 
\end{question}

\subsection{Main Results}\label{sec:intro_results}

Our main result is efficient algorithms in time $\tilde{O}(k)$ to learn frequencies within error $\tilde{O}(k)$ and recover the signal $x$ in the time window. This gives affirmative answers to Question~\ref{ques1} and Question~\ref{ques2}. We state our result and then discuss our tools and their connections to multiband signals.

\paragraph{Frequency recovery and reconstruction.}
Our first algorithm runs in time $\tilde{O}(k)$ and outputs a list of frequencies within error $\tilde{O}(k)$ and an reconstruction of the signal. For ease of exposition, this work assumes error $\epsilon$ is a small constant.
\begin{theorem}[Informal version of Theorem~\ref{thm:improper_learning} and Theorem~\ref{thm:efficient-frequency-cover}]\label{inform:improper_learning}
    Given any $F$, sparsity $k$, and a small constant $\epsilon$, there exists an efficient algorithm in $\tilde{O}(k)$ samples and $\tilde{O}(k)$ time such that
    given any noisy observation $y(t):=x(t)+\eta(t)$ with $x(t):=\sum_{j=1}^k \alpha_j e^{2 \pi \bi f_j t}$ and noise $\eta$, it outputs 
    \begin{enumerate}
        \item a list $L$ of frequencies with size $|L|=\tilde{O}(k)$ and error $O(k)$: for most of the frequencies $f_j$ of $x$, $\exists \tilde f \in L$ s.t. $|f_j-\tilde{f}| \le O(k)$;
        
        \item and an reconstruction $\tilde{x}$ such that $\|\tilde{x}-x\|^2_{[-1,1]} = O\big( \epsilon \cdot \|x\|^2_{[-1,1]} + \|\eta\|_{[-1,1]}^2 \big)$.
    \end{enumerate}
\end{theorem}
Here are a few remarks. The first remark is that the error of frequency recovery is almost optimal if the output is constrained to $\tilde{O}(k)$ frequencies. Otherwise an algorithm can always output more frequencies to reduce the error. It improves the previous upper bound $\tilde{O}(k^3)$ \cite{SSWZ23} significantly. Although one lower bound mentioned earlier is based on two indistinguishable signals with distance $\tilde{\Omega}(k^2)$ in the frequency domain, this distance is reduced to $\tilde{\Omega}(k)$ after an amortization over $k$ frequencies. This demonstrates that the error in Theorem~\ref{inform:improper_learning} is almost tight. The second remark is about the guarantee for most frequencies. Because $x(t)$ may have frequencies canceled with each other like $e^{2\pi \bi \delta t} - e^{-2\pi \bi \delta t}$ for a tiny $\delta$, it is impossible to recover all frequencies. We discuss this in Section~\ref{sec:overview} (see the guarantee in \eqref{eq:negligible_freq}) and refer to Theorem~\ref{thm:efficient-frequency-cover} for the formal statement about the frequency recovery. 

The third remark is that the time complexity of this algorithm is almost optimal because of an information-theoretical lower bound $\Omega(k \log F)$ on the sample complexity \cite{CKPS17}. The key of this algorithm is a fixed basis of size $F \cdot k^{O(1)}$. This work calls it \emph{Haar-Fourier basis} because it applies a Haar transformation to the Fourier series. This basis not only guarantees any Fourier-sparse signal is (analytically) sparse in this basis, but also provides a recovery algorithm based on discrete sparse Fourier transform \cite{GMS,HIKP}. The recovery algorithm is different from previous algorithms in continuous sparse FFT \cite{BCGLS,PS15,CKPS17, CP19_ICALP, SSWZ22, SSWZ23} that are in the classical hash-to-bins framework. While the hash-to-bins framework could learn each frequency accurately when frequencies are well-separated, it needs $\tilde{O}(k^2)$ samples and $\tilde{O}(k^2)$ time when frequencies are clustered. Hence Theorem~\ref{inform:improper_learning} provides an improper learning algorithm whose reconstruction $\tilde{x}$ is sparse in a Haar-Fourier basis. 
Next we develop a proper learning algorithm that outputs a $k$-Fourier-sparse reconstruction $\tilde{x}$ in time $\tilde{O}(k)^{O(k)}$. This relies on a statistical upper bound of learning Fourier-sparse signals. As mentioned earlier, the best upper bound on the sample complexity was $\tilde{O}(k^4)$ \cite{CP19_colt,SSWZ23} despite the running time. To optimize the statistical upper bound, we use the weighted sampling $D_k$ from \cite{CP19_colt}:
\begin{equation}
 D_k(t):=\frac{1}{(2+4\ln k)\max\{1-|t|,k^{-2}\}}. \label{eq:weighted-dist}
\end{equation}
We show that $\tilde{O}(k)$ samples from $D_k$ could guarantee empirical $L_2$ norm estimations for all $k$-Fourier-sparse signals, analogous to the restricted isometry property for discrete vectors. Combining with Theorem~\ref{inform:improper_learning} and a frequency net argument, this provides a proper learning algorithm.
\begin{theorem}\label{inform:stat_bound}
    Given any bandlimit $F$, sparsity $k$, and a small constant $\epsilon$, $m:=\frac{k \cdot (\log Fk/\epsilon)^{O(1)}}{\epsilon^2}$ independent random samples $i_1,\ldots,i_m \sim D_k$ (defined in \eqref{eq:weighted-dist}) guarantee that with probability $0.99$,
    \[
    \forall x(t):=\sum_{j=1}^k \alpha_j e^{2 \pi \bi f_j t} \text{ with } |f_j| \le F, \qquad \frac{1}{m}\left( \frac{|x(i_1)|^2}{D_k(i_1)}+\cdots+\frac{|x(i_m)|^2}{D_k(i_m)} \right) = (1\pm\epsilon) \cdot \|x\|_{[-1,1]}^2.
    \]
    Moreover, there exists an algorithm in $\tilde{O}(k)^{O(k)}$ time and  $m:=\frac{k \cdot (\log Fk)^{O(1)}}{\epsilon^2}$ samples such that
    given any noisy observation $y(t):=x(t)+\eta(t)$ of a $k$-Fourier-sparse signal $x(t)$ and noise $\eta$,  it returns a $k$-Fourier-sparse signal $x'$ with $\|x'-x\|^2_{[-1,1]} = O\left(\epsilon \cdot \|x\|^2_{[-1,1]} + \|\eta\|_{[-1,1]}^2 \right)$.
\end{theorem}

\paragraph{Techniques and tools.} We develop a variety of tools to prove our main results. We list a few of them that may be of independent interest. The first contribution is that Fourier-sparse signals with arbitrary real frequencies have sparse approximations in a fixed Haar-Fourier basis mentioned earlier. Let $\langle w,z \rangle_{[-1,1]}:=\int_{-1}^1 w(t) \overline{z(t)} \mathrm{d} t$ be the inner product corresponding to $\|\cdot\|_{[-1,1]}$. 
\begin{lemma}[Informal version of Theorem~\ref{thm:approximation}]\label{lem:sparse_in_Haar_dictionary}
    Given any error $\delta$, $k$, and bandlimit $F$, there exists an orthonormal Haar-Fourier basis $\mathcal{B}$ of size $2F + (\frac{k}{\delta})^{O(1)}$ such that any $k$-Fourier-sparse signal $x(t)$
        admits an analytically sparse approximation under $\mathcal{B}$:
    let $\alpha_\phi:=\langle x, \phi \rangle_{[-1,1]}$ for every $\phi \in \mathcal{B}$ such that $x':=\sum_{\phi \in \mathcal{B}} \alpha_\phi \cdot \phi$ is the orthogonal projection of $x$ onto $span\{\mathcal{B}\}$.
    Then
    \begin{align}
    \text{ the approximation error }
    \|x-x'\|_{[-1,1]} & \le \delta \cdot \|x\|_{[-1,1]}
    \label{eq:approximation_error}    \\
    \text{ and the coefficient vector $(\alpha_\phi)_{\phi \in \mathcal{B}}$ is analytically sparse } \sum_{\phi \in \mathcal{B}} | \alpha_\phi| & \le \tilde{O}(\sqrt{k}) \cdot (\sum_{\phi \in \mathcal{B}} |\alpha_\phi|^2)^{1/2}.    \label{eq:analytic_sparse}
    \end{align}
\end{lemma}
We call a vector $v$ $s$-analytic-sparse only if $\|v\|_1 \le s \cdot \|v\|_2$. It is well known that any vector of analytic sparsity $s$ has $(1-\delta)$ fraction of its squared $\ell_2$ norm concentrated on its largest $\frac{s^2}{\delta}$ elements (for any $\delta$). So Lemma~\ref{lem:sparse_in_Haar_dictionary} implies that $x$ always has $\tilde{O}(k/\delta)$-sparse approximations in  $\mathcal{B}$. To the best of our knowledge, this is the first fixed basis such that Fourier-sparse signals (with arbitrary frequencies) are $\tilde{O}(k)$-sparse in it. We believe this would find applications in other problems such as super-resolution.

Another conceptual contribution of this work is a connection between Fourier-sparse signals and multiband signals. This work calls $z(t)$ a multiband signal if its Fourier transform $\wh{z}$ is supported on a finite number of intervals, say $I_1 \cup I_2 \cup \cdots \cup I_n$ in the bandlimit $[-F,F]$. When $n=1$, we call $z$ a single-band signal. For convenience, let $\supp(\wh{z})$ denote the union of these $n$ intervals $I_1,\ldots,I_n$ as the support of its Fourier spectrum, and $|\supp(\wh{z})|$ denote the support size $|I_1| + \cdots + |I_n|$. The frequency recovery in Theorem~\ref{inform:improper_learning} relies on a multiband approximation of Fourier-sparse signals.

\begin{theorem}[Informal version of Theorem~\ref{thm:multiband_approx}]\label{thm:learn_Fourier_sparse}
    Given any $\epsilon>0$ and $k$, for any $x(t):=\sum_{j=1}^k \alpha_j e^{2 \pi \bi f_j t}$ with $k$ arbitrary frequencies $f_1,\ldots,f_k$, there exists a signal $z$ whose Fourier support has $\le k$ bands such that (a) $|\supp(\wh{z})|=\tilde{O}(k^2/\epsilon)$; (b) $\|z-x\|_{[-1,1]}^2 \le \epsilon \cdot \|x\|_{[-1,1]}^2$; (c) most of the energy of $z(t)$ is concentrated in the time window: $\|z\|_{[-1,1]}^2 \ge (1-\epsilon) \cdot \|z\|_2^2$ for $\|z\|_2^2 := \int_{-\infty}^{+\infty} |z(t)|^2 \mathrm{d} t$.
\end{theorem}
We remark that Property (c) is the main constraint of $z$; otherwise setting $z=x$ gives a multiband signal. The main step of the frequency recovery in Theorem~\ref{inform:improper_learning} is to find the Fourier support of this multiband approximation. After recovering the multiband, because it has $\le k$ bands and a support size $\tilde{O}(k^2/\epsilon)$, one could use $\tilde{O}(k/\epsilon)$ frequencies to cover these bands within distance $O(k)$. This is described in Section~\ref{sec:freq_recovery}. 

Compared to previous works \cite{CKPS17,CP19_ICALP,SSWZ23,zhang2026optimalextrapolationboundssparse}, Theorem~\ref{thm:learn_Fourier_sparse} improves the support size $|\supp(\wh{z})|$ from $\tilde{O}(k^3)$ to $\tilde{O}(k^2)$. The key idea of the improvement is to (1) partition frequencies in $\wh{x}$ into clusters $\cC_1,\ldots,\cC_n$ that are as small as possible while keeping them almost orthogonal to each other (in the time window) and (2) replace each cluster $\cC_i$ with a band of size proportional to the number of frequencies in $\cC_i$. To partition frequencies into almost orthogonal parts, an important ingredient is an almost-tight characterization of the orthogonality between two Fourier-sparse signals.

\begin{lemma}[Informal version of Lemma~\ref{lem:almost_orthogonal_clusters}]
\label{inform:almost_orthogonal_clusters}
Let $w(t)$ be an $\ell$-Fourier-sparse signal with frequencies $f_1,\ldots,f_{\ell}$ and $z(t)$ be an $r$-Fourier-sparse signal with frequencies $f'_1,\ldots,f'_r$. Given any correlation $\delta$, if the \emph{frequency separation} $\underset{i \in \{1,\ldots,\ell\},j \in \{1,\ldots,r\}}{\min} |f_i-f'_{j}|=\Omega(\frac{\ell r}{\delta} \cdot \log \frac{\ell r}{\delta})$, then

\begin{equation}
\label{eq:orthogonal_windows_intro}
|\langle w,z\rangle_{[-1,1]}|
    \le \delta \cdot \|w\|_{[-1,1]} \cdot \|z\|_{[-1,1]}.
\end{equation}
\end{lemma}

A separation $\tilde{\Omega}(\frac{\max\{\ell,r\}^2}{\delta})$ was indicated by Chen and Price's analysis \cite{CP19_ICALP} and properties of Fourier-sparse signals from Erd{\'e}lyi \cite{erdelyi2016inequalities}. Based on an extrapolation bound about Fourier-sparse signals proved by Zhang \cite{zhang2026optimalextrapolationboundssparse}, we improve the analysis in \cite{CP19_ICALP} to obtain a clean  characterization of the frequency separation. On the other hand, a separation $\Omega(\frac{\ell r}{\delta})$ is necessary to guarantee \eqref{eq:orthogonal_windows_intro}. We defer this discussion and the proof of Lemma~\ref{inform:almost_orthogonal_clusters} to Section~\ref{sec:orthogonal}.

The proper learning algorithm in Theorem~\ref{inform:stat_bound} relies on a frequency net $\mathcal{N}$ of Fourier-sparse signals. This finite frequency net shows that any $k$-Fourier-sparse signal $x$ with arbitrary frequencies $f_1,\ldots,f_k$ could be converted to a $k$-Fourier-sparse approximation $\tilde{x}$ whose frequencies are in $\mathcal{N}$. Previous work \cite{CKPS17} showed a net $\mathcal{N}$ of size $F \cdot k^{O(k^2)}$. This work improves the size of $\mathcal{N}$ to $F \cdot k^{O(1)}$. The key is a frequency-shifting lemma that bounds the error of shifting one frequency.
\begin{lemma}\label{lem:shift_one_freq}
    For any $x(t)=\sum_{j=1}^k \alpha_j e^{2 \pi \bi f_j t}$ with $k$ arbitrary frequencies $f_1,\ldots,f_k$ and any $f_{k+1} \notin \{f_1,\ldots,f_{k-1}\}$, there exists $x'(t)=\sum_{j=1}^{k-1} \alpha_j' e^{2 \pi \bi f_j t} + \alpha_{k+1} e^{2 \pi \bi f_{k+1} t}$ such that
    \[
    \|x(t)-x'(t)\|_{[-1,1]} \le O(1) \cdot |f_k-f_{k+1}| \cdot \|x\|_{[-1,1]}.
    \]
\end{lemma}

\paragraph{Learning multiband signals.} 
Given the connection between Fourier-sparse signals and multiband signals in Theorem~\ref{thm:learn_Fourier_sparse}, we studied how to learn a general multiband signal $z(t)$ finally. For multiband signals, previous works \cite{slepian1961prolate,landau1961prolate,landau1962prolate,xiao2001prolate,shkolnisky2006approximation,AKMMVZ18} have focused on the \emph{reconstruction} question --- how to reconstruct $z(t)$ in the time window efficiently given $\supp(\wh{z})=I_1 \cup \ldots \cup I_n$? In this work, we call such an algorithm a \emph{reconstructing} algorithm in order to distinguish it from our \emph{learning} problem where $\supp(\wh{z})$ is unknown. Unlike in the discrete setting, given $\supp(\wh{z})$, one cannot learn all Fourier coefficients $\wh{z}(f)$ in the continuous setting. So a natural goal is therefore to output $\tilde{z}$ such that $\tilde{z} \approx z$ in the time window. 

In particular, the authors of~\cite{AKMMVZ18} provided an efficient reconstructing algorithm whose sample complexity $m:=\tilde{O}(n+|I_1|+\cdots+|I_n|)$ is almost optimal and time complexity is $m^{O(1)}$. However, all previous works assume the locations of bands, equivalently the Fourier support $\supp(\wh{z})$, are known or given to the algorithm.  An exception is \cite{CP19_ICALP} that studies how to recover the Fourier support of a single-band signal efficiently. While multiband signals have a wide variety of applications in communication, imaging, and statistics \cite{Eldar_2015}, an important open question \cite{AKMMVZ18} is how to recover the locations of the bands in $\supp(\wh{z})$ efficiently?

\begin{theorem}[Informal version of Theorem~\ref{thm:locations_multiband}]\label{inform:learn_multiband}
    Given a small constant $\epsilon$, let $z$ be a multiband signal whose Fourier spectrum $\supp(\wh{z})=I_1 \cup I_2 \cup \cdots \cup I_n$ in $[-F,F]$ has $n$ \emph{unknown} intervals. In the time window $[-1,1]$, suppose that $z$ satisfies
    \begin{equation}\label{cond:multiband_recovery_intro} 
        \int_{-1}^1 |z(t)|^2 \mathrm{d} t \ge (1-\epsilon) \int |z(t)|^2 \mathrm{d} t.
    \end{equation} 
    
     Given any noisy observation $y(t)=z(t)+\eta(t)$ in $[-1,1]$ with $\|\eta\|_{[-1,1]}^2 \le \epsilon \cdot \|z\|_{[-1,1]}^2$ and an estimate of $n +\sum_i |I_i|$,       there exists an efficient algorithm that takes $\tilde{O}(\frac{n+\sum_i |I_i|}{\epsilon})$ samples and $\tilde{O}\big(n +\sum_i |I_i|\big)^3$ time to recover most intervals $I_i$ in $\supp(\wh{z})$ and return a reconstruction $\tilde{z}$ with $\|z-\tilde{z}\|_{[-1,1]}^2 = O(\epsilon) \cdot \|z\|_{[-1,1]}^2$. 
\end{theorem}

While this algorithm will recover $\supp(\wh{z})$, it will take an upper bound on $n  + |\supp(\wh{z})|$ as an input parameter. Condition~\eqref{cond:multiband_recovery_intro} is necessary to guarantee that the algorithms can recover $\supp(\wh{z})$ accurately. We discuss this in Remark~\ref{thm:locations_multiband} after the formal statement of recovering $\supp(\wh{z})$ in Theorem~\ref{thm:locations_multiband}. An intriguing question is to relax Condition~\ref{cond:multiband_recovery_intro} and output an recovery $\tilde{z} \approx z$ without recovering $\supp(\wh{z})$ accurately (like the frequency recovery in Theorem~\ref{inform:improper_learning}). We show two quantities about $z$ identified in \cite{CP19_ICALP} that could be used to relax this condition in Corollary~\ref{cor:interpolating_multiband}. 

Finally, we compare our results in Theorems~\ref{inform:improper_learning}
and~\ref{inform:learn_multiband} with previous works in Table~\ref{tab:results}. For Fourier-sparse signals, the third column about frequency recovery denotes error $\min_{\tilde{f} \in L} |\tilde{f} - f_j|$ for frequencies $f_j \in \wh{x}$,
when the algorithm is required to output a list $L$ of $\tilde{O}(k)$ frequencies. All results of \cite{CKPS17} have large polynomials in $k$. For example, its frequency recovery has an error bound $\tilde{O}(k^{7.5})$.

{\renewcommand{\arraystretch}{1.3}\begin{table}[h]
\centering
\begin{tabular}{|m{4em}|m{6em}|m{12em}|m{8em}|m{8em}|}
     \hline
     & results & frequency recovery & sample complexity & time complexity\\  
     \hline
          \multirow{3}{4em}{Fourier-sparse signals} & \cite{CKPS17} & error $\tilde{O}(k^{O(1)})$ & $\tilde{O}(k^{O(1)})$ & $\tilde{O}(k^{O(1)})$ \\
          \cline{2-5} &\cite{CP19_colt} & \centering - & $\tilde{O}(k^4)$ & $ (k^{O(k^2)} \cdot FT)^{O(k)}$ \\
     \cline{2-5}
     & \cite{SSWZ23} & error $\tilde{O}(k^3)$ & $\tilde{O}(k^4)$ & $\tilde{O}(k^{4\omega})$  \\
     \cline{2-5}
     & Theorem~\ref{inform:improper_learning} & error $\tilde{O}(k)$ & $\tilde{O}(k)$ & $\tilde{O}(k)$ \\
     \hline
     
     \multirow{3}{4em}{multiband signals} & \cite{AKMMVZ18} & assume $\supp(\wh{z})=I_1 \cup \cdots \cup I_n$ is given & $\tilde{O}(n+ \sum_i |I_i|)$  & $O(\sum_i |I_i| + n)^{\omega} + n \cdot O(\sum_i |I_i|)^{2}$\\
     \cline{2-5}
     & Theorem~\ref{inform:learn_multiband} & a superset covers $\supp(\wh{z})$ with size $O(|\supp(\wh{z})| + n)$ & $\tilde{O}(n+ \sum_i |I_i|)$ & $\tilde{O}(\sum_i |I_i| + n)^{3}$\\
     \hline

\end{tabular}
\caption{Summary on frequency recovery, sample complexity, and time complexity of learning $k$-Fourier-sparse signals and multiband signals.}\label{tab:results}
\end{table}}

\paragraph{Organization.} In Section~\ref{sec:prelim}, we present preliminaries about Fourier transforms. Section~\ref{sec:overview} provides an overview of this paper. Section~\ref{sec:characterizations} analyzes the Haar-Fourier basis for Fourier-sparse signals and proves the analytically sparse result of Lemma~\ref{lem:sparse_in_Haar_dictionary} with the statistical upper bound in Theorem~\ref{inform:stat_bound}. In Section~\ref{sec:improper_learning}, we prove the reconstruction part of Theorem~\ref{inform:improper_learning}. Section~\ref{sec:orthogonal} and Section~\ref{sec:multiband_approx} prove Theorem~\ref{thm:learn_Fourier_sparse}  and Lemma~\ref{inform:almost_orthogonal_clusters} respectively. In Section~\ref{sec:freq_recovery}, we prove the frequency recovery part of Theorem~\ref{inform:improper_learning} and the proper learning part of Theorem~\ref{inform:stat_bound}. Finally, Section~\ref{sec:multiband_recovery} proves Theorem~\ref{thm:locations_multiband} about multiband recovery. In Appendix~\ref{sec:related_work}, we discuss related works due to the space constraint.

\section{Preliminaries}\label{sec:prelim}
For ease of exposition, this work fixes the time window to be $[-1,1]$ and the bandlimit to be $[-F,F]$. We always treat the error $\epsilon$ as a fixed small constant, while we use $C$ to denote various constants in the proof. Also, we use $a=b \pm c$ to indicate $a \in [b-c,b+c]$. For any $n > 0$, let $[n]:=\{0,1,\ldots,n-1\}$ in this work. 

This work also uses $\lesssim$ and $\gtrsim$ to hide various constants in analyses and calculations.

For a vector $v \in \mathbb{C}^\mathbb{Z}$, let $\|v\|_1:=\sum_n |v[n]|$ and $\|v\|_2:=(\sum_n |v[n]|^2)^{1/2}$. For any $b$, let \emph{shift} of $v$ at $b$ denote the vector $u[n]:=v[n-b]$ and $\bott_s(v)$ be the vector after erasing the largest $s$ absolute entries. This fact is useful in finding large Haar-Fourier coefficients, whose proof is in Appendix~\ref{sec:proof_bottom}.
\begin{fact}\label{fact:bottom_vectors}
    For any sparsity $s$ and two vectors $u$ and $v$,
    \[
    \|\bott_s(u)\|^2_2 = O(\|\bott_s(v)\|^2_2 + \|u-v\|^2_2).
    \]
\end{fact}
Moreover, we say $v$'s analytic sparsity is $s$ if $\|v\|_1 \le s \cdot \|v\|_2$. It is well known that $v$ of analytic sparsity $s$ has sparse approximations: $\|\bott_{s^2/\epsilon}(v)\|_2^2 \le \epsilon \cdot \|v\|_2^2$ for any $\epsilon<1$.

For an interval $I \subset (-\infty,+\infty)$, let $\mathbf{1}_I$ denote the indicator function of interval $I$. So for a signal $x:\mathbb{R}\rightarrow \mathbb{C}$, $\mathbf{1}_I \cdot x$ denotes the truncation of $x$ on interval $I$.  For an integrable signal $z$, let $\|z\|_2:=(\int |z(t)|^2 \mathrm{d} t)^{1/2}$. For any signal $z$ and any interval $I$, $\|z\|_I:=(\int_I |z(t)|^2 \mathrm{d} t)^{1/2}$. For convenience, we call $\|z\|_2^2$ the energy of $z$ and $\|z\|_{I}^2$ the energy of $z$ in $I$.

 For two integrable functions $x$ and $y$, we define the inner product over $\mathbb{R}$ as $\langle x,y \rangle := \int x(t) \bar{y}(t) \mathrm{d} t$. For any two signals $x$ and $y$, the inner product over an interval $I$ is $\langle x , y \rangle_I := \int_I x(t) \bar{y}(t) \mathrm{d} t$. 

\paragraph{Fourier transforms.} Given the time window $[-1,1]$, for any signal $x:\R \to \C$, our learning algorithms would consider the truncation of  $\mathbf{1}_{[-1,1]} \cdot x$ as a periodic signal. So $\wh{x}[n]$ for $n \in \mathbb{Z}$ would denote its Fourier coefficients on the grid:
\begin{equation}\label{eq:def_periodic_FFT}
\wh{x}[n]:=\frac{1}{\sqrt{2}} \int_{-1}^1 x(t) \cdot e^{- \bi \pi n t} \mathrm{d} t.   
\end{equation}
Then we have the reverse Fourier transform $
x(t)=\sum_{n=-\infty}^{+\infty} \wh{x}[n] e^{i \pi n t}/ \sqrt{2}$. The Plancherel-Parseval identity shows
\begin{equation}\label{eq:plancherel_parseval}
    \|x\|_{[-1,1]}^2 = \sum_{n=-\infty}^{+\infty} |\wh{x}[n]|^2 \text{ and } \langle x,y\rangle_{[-1,1]}=\sum_{n=-\infty}^{+\infty} \wh{x}[n] \overline{\wh{y}[n]}.
\end{equation}

Our work also uses the Fourier transform over $\mathbb{R}$. To distinguish it from the transform of periodic signals, we use $\wh{x}[n]$ to denote the Fourier coefficients as a periodic signal and $\wh{x}(f)$ to denote the coefficient as a signal over $\R$. Then the Fourier transform $\wh{g}(f)$ of an integrable function $g$ is 
\[
\wh{g}(f)=\int_{-\infty}^{+\infty} g(t) e^{-2\pi \bi ft} \mathrm{d} t \text{ and } g(t)=\int_{-\infty}^{+\infty} \wh{g}(f) e^{2\pi \bi ft} \mathrm{d} f.
\]
Thus, $\wh{\mathbf{1}_{[-1,1]} x}[n]=\frac{1}{\sqrt{2}}\wh{\mathbf{1}_{[-1,1]} x}(n/2)$ by their definition. Next is the Plancherel-Parseval identity  over $\R$.
\begin{theorem}\label{thm:Plancherel_Parseval}
    For any square-integrable function $x$, $\|x\|_2=\|\wh{x}\|_2$. For any two square-integrable functions $x$ and $y$, $\langle x,y \rangle=\langle \wh{x}, \wh{y} \rangle$.
\end{theorem}

The discrete Fourier transform of a vector $u \in \mathbb{C}^B$ is \[
\wh{u}[f]=\sum_{t \in [B]} u[t] \cdot e^{-2 \pi \bi ft/B} \text{ and } u[t]=\frac{1}{B} \sum_{f \in [B]} \wh{u}[f] \cdot e^{2 \pi \bi ft/B}
\]

\paragraph{Fourier-sparse signals.} For a function $h$ whose Fourier transform $\wh{h}$ has a compact support, we define the Fourier support of $h$ as the set $\{f:\wh{h}(f) \neq 0\}$. 
Formally, a \emph{Fourier-sparse} signal $x(t)=\sum_{j=1}^k \alpha_j e^{2 \pi \bi f_j t}$ has its Fourier transform $\wh{x}(f)=\sum_{j=1}^k \alpha_j \delta_{f_j}(f)$ with the Dirac-delta function $\delta$. We use the following bounds of Erd\'elyi \cite{erdelyi2016inequalities} and Zhang \cite{zhang2026optimalextrapolationboundssparse} on Fourier-sparse signals.
\begin{lemma}
\label{lemma:bounds_Fourier_sparse_signals}
Let $x(t)=\sum_{j=1}^{s}\alpha_j e^{2\pi\bi f_j t}$, with arbitrary real
frequencies $f_j$. Then
\begin{enumerate}
    \item \label{item:uniform_bound_on_interval}
    For $|t|\le 1$, $|x(t)|\le \frac{\pi s}{2 \sqrt{2}} \|x\|_{[-1,1]}$ \cite[Theorem~2.3]{erdelyi2016inequalities} --- this is tight: there exists an $s$-Fourier-sparse signal $x$ such that $\max_{t \in [-1,1]} |x(t)| = \Omega(s) \cdot \|x\|_{[-1,1]}$ \cite{CKPS17}.

    \item \label{item:bound_on_t} For any $t \in [-1,1]$, $|x(t)|^2 \le \frac{s}{1-|t|} \cdot \|x\|_{[-1,1]}^2$ \cite{erdelyi2016inequalities}. 
    
    \item \label{item:polynomial_bound_outside_interval}
    For $|t|>1$, $|x(t)|\le \bigl(e(|t|+1)\bigr)^s \max_{u\in[-1,1]}|x(u)|$ 
    \cite{erdelyi2016inequalities}. 
    \item \label{item:square_root_bound_outside_interval}
    For $1\le |t|\le 2$, $|x(t)|\le 192 \cdot s \exp \bigl(3 \sqrt{2} \cdot s\sqrt{|t|-1}\bigr) \|x\|_{[-1,1]}$ 
    \cite[Theorem~1.1]{zhang2026optimalextrapolationboundssparse}.
\end{enumerate}
\end{lemma}

\section{Overview}\label{sec:overview}
We provide an overview of algorithms and analyses discussed in Section~\ref{sec:intro_results}.

\paragraph{Haar-Fourier basis and analytic sparsity.} We first discuss Lemma~\ref{lem:sparse_in_Haar_dictionary}, which lays the foundation of Theorem~\ref{inform:improper_learning}. In Lemma~\ref{lem:sparse_in_Haar_dictionary} and Theorem~\ref{inform:improper_learning}, we consider $x(t)$ in the time window $[-1,1]$ as a periodic signal over $[-1,1]$. The first choice of a fixed basis to guarantee the approximation of a Fourier-sparse signal $x(t)$ in Lemma~\ref{lem:sparse_in_Haar_dictionary} would be low-frequency Fourier series (on the grid) for periodic signals: $e^{\bi \pi nt} \text{ for } |n| \le N$ for $N:=O(F) + (k/\epsilon)^{O(1)}$. One issue of this basis is that Property 1 of Lemma~\ref{lemma:bounds_Fourier_sparse_signals} implies $\frac{\max_{t \in [-1,1]}|x(t)|^2}{\|x\|_{[-1,1]}^2}$ could be $\Omega(k^2)$ for Fourier-sparse signals. This indicates that the sparsity of approximations in the Fourier series becomes $\Omega(k^2)$. 

To reduce the sparsity, we apply a Haar transformation (equivalently a Haar basis) to this basis. First, it is well known that a Haar transform can compress piecewise simple signals \cite{Efr99} because Haar vectors are piecewise constant functions. We use Table~\ref{tab:haar-h8} to illustrate a Haar basis on a block of size 8. While we defer the formal definition of our Haar basis and  Haar-Fourier basis to Definition~\ref{def:basis_B}, we provide a short introduction here. In general, for a 2-power $K=O(k)$ and block $[K]=\{0,1,\ldots,K-1\}$, there are $\log K + 1$ types of vectors in the Haar basis $\mathcal{H}_K$ of this block: a normalized constant-$1$ vector $\frac{1}{\sqrt{K}}H^{sc}$ and a normalized Haar vector $H^{wav}_{[\ell]}$ for every 2-power $\ell$ in $[2,K]$. In particular, each $H^{wav}_{[\ell]}$ is $1/\sqrt{\ell}$ in its first half of support $[\ell/2]$ and $-1/\sqrt{\ell}$ in its second half of support $\ell/2 + [\ell/2]$; and a shift of $H^{wav}_{[\ell]}$ at $b \ell$ (for $b=0,1,\ldots,K/\ell-1$) is defined as $H^{wav}_{b \ell+[\ell]}[n]:=H^{wav}_{[\ell]}[n-b \ell]$. Then $\mathcal{H}_K$, formed with $H^{sc}$ and all these shifts of $H^{wav}_{b\ell + [\ell]}$, is an orthonormal basis of $\C^{[K]}$. We pick $K=O(k)$, partition low frequencies $\{-N,-N+1,\ldots,0,1,\ldots,N-1\}$ into $\frac{2N}{K}$ blocks, and take the union of $\mathcal{H}_K$ over all $\frac{2N}{K}$ blocks as the Haar basis $\mathcal{H}_{2N}$. 

\begin{table}[htbp]
\centering \footnotesize
\[
\begin{array}{c|cccccccc|c}
\text{vector} & 0&1&2&3&4&5&6&7& \text{scale}\\
\hline
H^{\mathrm{sc}}_{[8]} 
& +&+&+&+&+&+&+&+ & \frac{1}{\sqrt{8}}\\[2mm]
H^{\mathrm{wav}}_{[8]} 
& +&+&+&+&-&-&-&- & \frac{1}{\sqrt{8}}\\[2mm]
H^{\mathrm{wav}}_{0+[4]} 
& +&+&-&-&\cdot&\cdot&\cdot&\cdot & \frac{1}{2}\\[2mm]
H^{\mathrm{wav}}_{4+[4]} 
& \cdot&\cdot&\cdot&\cdot&+&+&-&- & \frac{1}{2}\\[2mm]
H^{\mathrm{wav}}_{0+[2]} 
& +&-&\cdot&\cdot&\cdot&\cdot&\cdot&\cdot & \frac{1}{\sqrt{2}}\\[2mm]
H^{\mathrm{wav}}_{2+[2]} 
& \cdot&\cdot&+&-&\cdot&\cdot&\cdot&\cdot & \frac{1}{\sqrt{2}}\\[2mm]
H^{\mathrm{wav}}_{4+[2]} 
& \cdot&\cdot&\cdot&\cdot&+&-&\cdot&\cdot & \frac{1}{\sqrt{2}}\\[2mm]
H^{\mathrm{wav}}_{6+[2]} 
& \cdot&\cdot&\cdot&\cdot&\cdot&\cdot&+&- & \frac{1}{\sqrt{2}}
\end{array}
\]
\caption{Haar basis for $\mathcal{H}_8$ with block size $K=8$}
\label{tab:haar-h8}
\end{table}
The Haar-Fourier transform applies this Haar basis to the Fourier series $\{e^{\bi \pi n t}:n \in \{-N,-N+1,\ldots,0,1,\ldots,N-1\}\}$. In other words, a vector $\phi_H$ in the Haar-Fourier basis $\mathcal{B}$ is associated with a Haar vector $H \in \mathcal{H}_{2N}$, which is defined as $\phi_H(t):=\{\frac{1}{\sqrt{2}} \cdot \sum_{n} (-1)^n \cdot H(n) \cdot e^{\bi \pi n t}\}$. 
We remark that the $(-1)^n$ term in $\phi_H$ is to cancel the counterpart of the Fourier transform (see \eqref{eq:simplification_b_n} below) and $K=O(k)$ is chosen to minimize the error of frequency recovery later. By the definition of the Haar basis $\mathcal{H}_{2N}$, the approximation error of $x$ onto $\mathcal{B}$ is equivalent to its error onto the low-frequency Fourier series. So we only discuss the analytic sparsity of $x$ under $\mathcal{B}$.

On the other hand, for a $k$-Fourier-sparse signal $x(t):=\sum_{j=1}^k \alpha_j e^{2 \pi i f_j t}$, although its Fourier transform $\wh{x}[n]=\frac{1}{\sqrt{2}} \int_{-1}^1 x(t) e^{- \bi \pi n t} \mathrm{d} t$ is not piecewise constant, it is piecewise monotone (on both the real and imaginary parts respectively) after removing a $(-1)^n$ factor. Consider $a_n:=(-1)^n \cdot \wh{x}[n]$ for $n \in \{-N,\ldots,N-1\}$ and simplify it by the definition of $\wh{x}[n]$ in \eqref{eq:def_periodic_FFT}:
\begin{align}
    a_n & = \frac{(-1)^n}{\sqrt{2}} \int_{-1}^1 x(t) e^{- \bi \pi n t} \mathrm{d} t \notag \\
    & = \frac{(-1)^n}{\sqrt{2}} \int_{-1}^1 (\sum_{j=1}^k \alpha_j e^{2 \pi i f_j t}) \cdot e^{- \bi \pi n t} \mathrm{d} t \notag \\
    & = \frac{(-1)^n}{\sqrt{2}} \sum_{j=1}^k \alpha_j \cdot \frac{e^{(2 \pi \bi f_j - \bi \pi n)} - e^{-(2 \pi \bi f_j - \bi \pi n)}}{2 \pi \bi f_j - \bi \pi n} \notag \\
    & = \frac{(-1)^n}{\sqrt{2}} \sum_{j=1}^k \alpha_j \cdot \frac{2\sin(2 \pi f_j - \pi n)}{2 \pi f_j - \pi n}  = \sqrt{2} \sum_{j=1}^k \alpha_j \cdot \frac{\sin(2 \pi f_j)}{2 \pi f_j - \pi n}. \label{eq:simplification_b_n}
\end{align}

Since $\sin(2 \pi f_j)$ and $\alpha_j$ are fixed, $a_n=\frac{A(n)}{Q(n)}$ for $Q(n):=\prod_{j=1}^k (2 \pi f_j - \pi n)$ of degree $k$ and $A(n):=\sum_{j=1}^k \sqrt{2}\alpha_j \cdot \sin(2 \pi f_j) \cdot \prod_{i \neq j} (2 \pi f_i - \pi n)$ of degree at most $k-1$. So $a_n$ can be partitioned into $O(k)$ monotone subsequences because $\frac{A(x)}{Q(x)}$ has at most $2k-1$ zeros. 

Now the crucial point is that a monotone sequence has an analytically sparse representation in any Haar basis \cite{Boc69}. For example, consider a monotone sequence $(a_0,\ldots,a_{L-1}) \in \mathbb{R}^L$ over $L/K$ blocks under the Haar transform $\mathcal{H}_{2N}$. For any fixed 2-power $\ell \in [2,K]$, the sum of  coefficients over shifts of $H^{wav}_{[\ell]}$ of the form $\frac{1}{\sqrt{\ell}} (1_{s+[\ell/2]}-1_{s+\ell/2+[\ell/2]})$ is at most $\|a\|_2$:
\begin{align*}
\sum_{H=H^{wav}_{b+[\ell]}} \left|\sum_{n=b}^{n=b+\ell-1} a_n \cdot H(n) \right| & = \sum_{b=0,\ell,\cdots,L-\ell} \frac{1}{\sqrt{\ell}} \left(\sum_{n=b}^{b+\ell/2-1} a_n - \sum_{n=b+\ell/2}^{b+\ell-1} a_n \right) \\
& \le \frac{1}{\sqrt{\ell}} (a_0 + a_1 + \cdots + a_{\ell/2-1} - a_{L-\ell/2} - \cdots - a_{L-1}) \le \|a\|_2.    
\end{align*}
This bounds the sum over $H^{wav}_{[\ell]}$. For the constant-1 vector $H^{sc}$, we bound $\sum_n |\wh{x}[n]|$ via properties from Lemma~\ref{lemma:bounds_Fourier_sparse_signals}. This part appears in Section~\ref{sec:characterizations}.

\paragraph{Reconstruction in $\tilde{O}(k)$-time.} Next, we discuss the reconstruction algorithm in Theorem~\ref{inform:improper_learning}. Its main property is an almost optimal running time $\tilde{O}(k)$. First of all, the analytic sparsity in Lemma~\ref{lem:sparse_in_Haar_dictionary} implies that it is sufficient to recover the largest $\tilde{O}(k)$ coefficients under the Haar-Fourier basis. Hence for a fixed type $H^{wav}_{[\ell]}$ in $\mathcal{H}_{2N}$, we consider coefficients $\langle x, \phi_H \rangle$ of a Haar-Fourier vector $\phi_H$ where $H \in \mathcal{H}_{2N}$ are shifts of $H^{wav}_{[\ell]}$ like $H^{wav}_{b\ell + [\ell]}$.

The key observation is that $\phi_H(t)$ rotates $\phi_{H^{wav}_{[\ell]}}(t)$ by $e^{ \bi \pi b \ell t}$ such that coefficients of $H^{wav}_{b\ell +[\ell]}$ look like Fourier-coefficients of a periodic signal at $b$. Specifically, $\phi_H$ of $H^{wav}_{b\ell + [\ell]}$ has 
\[
\phi_H(t):=\frac1{\sqrt2}\sum_n(-1)^nH^{wav}_{[\ell]}(n-b\ell) \cdot e^{\bi\pi nt}=(-1)^{b\ell}e^{\bi\pi b\ell t}\phi_{H^{wav}_{[\ell]}}(t)=e^{\bi\pi b\ell t} \cdot \phi_{H^{wav}_{[\ell]}}(t).
\] 
So the coefficient $\langle x, \phi_H \rangle = \int_{-1}^1 x(t) \overline{\phi_{H^{wav}_{[\ell]}}(t)} \cdot e^{-\bi \pi b \ell  t} \mathrm{d} t$. We round it $\ell$ times to transform this inner product to a Fourier transform over $[-1,1)$: $\zeta(u) 
= \frac{\sqrt{2}}{\ell} \underset{t \in [-1,1): \ell t \equiv u \mod 2}{\sum} x(t) \overline{\phi_{H^{wav}_{[\ell]}}(t)}$ satisfies  that $\wh{\zeta}[b]$ equals the above inner product.

Now we can apply sparse Fourier transforms \cite{GMS,HIKP} to recover the largest $\tilde{O}(k)$ coefficients of $\zeta(u)$ (and the largest coefficients of $x$ in the Haar-Fourier basis). The only issue is that accessing $\zeta(u)$ takes $\ell=O(k)$ samples to $x$. The last ingredient here is a random sampling in $\{t \in [-1,1): t \ell \equiv u\}$ to approximate $\zeta(u)$. The first distribution $q_\ell(t)$ might be defined to be proportional to $|\phi_{H^{wav}_{[\ell]}}(t)|$. While this reduces the sample complexity to $\tilde{O}(1)$, sampling from $q_{\ell}(t)$ still takes time $O(\ell)$. Our algorithm will turn $q_{\ell}(t)$ into a piecewise constant distribution $p_{\ell}(t)$ to improve the sampling time to $\tilde{O}(1)$. 
For ease of exposition, we illustrate the difference between $q_{\ell}$ and $p_{\ell}$ in Figure~\ref{fig:sampling} and defer the formal description of this reconstruction algorithm to Section~\ref{sec:improper_learning}.

\begin{figure}
    \centering
    \includegraphics[width=0.6\linewidth]{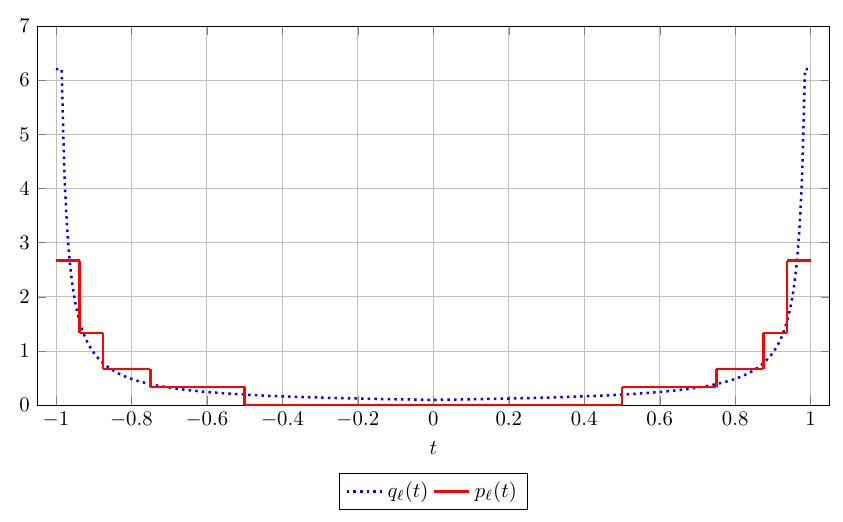}
    \caption{Comparison of $q_\ell(t)$ and $p_\ell(t)$  for approximating $\zeta(u)$ by sampling}
    \label{fig:sampling}
\end{figure}

\paragraph{Frequency recovery and multiband approximation.} Now we discuss the frequency recovery part in Theorem~\ref{inform:improper_learning} and the multiband approximation in Theorem~\ref{thm:learn_Fourier_sparse}. First of all, the large coefficients in the Haar-Fourier basis provide a multiband $J \subset \R$ of $\tilde{O}(k)$ intervals and total length $\tilde{O}(k^2)$ because each Haar-Fourier coefficient is a combination of a block of $K=O(k)$ Fourier coefficients (which gives an interval in $\R$ of length $O(k)$) and there are at most $\tilde{O}(k)$ large Haar-Fourier coefficients. However, the technical question is how to prove that frequencies in $x$ fall into $J$? 

We will show this for an expansion $J'$ of $J$ with total length $\tilde{O}(k^2)$. More importantly, we would prove that frequencies not in $J'$ are negligible: for $x=\sum_{j=1}^k \alpha_j e^{2 \pi \bi f_j t}$,
\begin{equation}\label{eq:negligible_freq}
    \left\|\sum_{f_j \notin J'} \alpha_j e^{2 \pi \bi f_j t}\right\|_{[-1,1]}^2 = O(\epsilon) \cdot \|x\|_{[-1,1]}^2.
\end{equation}
This is the formal guarantee about covering most frequencies in Theorem~\ref{inform:improper_learning}. Because frequencies may get canceled together, this shows non-negligible ones got covered by $I$. The analysis of \eqref{eq:negligible_freq} relies on the multiband approximation in Theorem~\ref{thm:learn_Fourier_sparse}. Because $J'$ has at most $\tilde{O}(k)$ bands and $|J'|=\tilde{O}(k^2/\epsilon)$, one can cover it by $\tilde{O}(k/\epsilon)$ frequencies of radius $O(k)$.

Let us focus on the multiband approximation of Fourier-sparse signals right now. The high level idea is to partition $f_1,\ldots,f_k$ of $x$ into as many clusters $\cC_1,\ldots,\cC_n$ as possible while ensuring that they are almost orthogonal to each other. We define this part formally in Section~\ref{sec:multiband_approx}. Basically, each cluster $\cC_{\ell}$ is a set of frequencies with coefficients $(f_j,\alpha_j)$ in $x$ and induces a signal $x_{\cC_\ell}:=\sum_{(f,\alpha) \in \cC_\ell} \alpha \cdot e^{2 \pi \bi f t}$ such that $x=\sum_{\ell=1}^n x_{\cC_\ell}$. Since these clusters are almost orthogonal, we want to show that any non-negligible cluster $\cC_j$ will be covered by $I$.

First of all, the orthogonality between two clusters is exactly controlled by Lemma~\ref{inform:almost_orthogonal_clusters}. Based on Lemma~\ref{inform:almost_orthogonal_clusters}, we propose a partition procedure of clusters (described formally in Algorithm~\ref{alg:cluster}): let $|\cC|$  denote the number of frequencies in a cluster and $\dist(\cC,\cC'):=\underset{f \in \cC,f' \in \cC'}{\min} |f-f'|$ denote the distance between two clusters in the frequency domain. Then it starts with $k$ clusters $\cC_\ell=\{(\alpha_\ell,f_\ell)\}$ and merges two clusters $\cC_i$ and $\cC_j$ whenever $\dist(\cC_i,\cC_j) \le \tilde{O}(k) \cdot \min\{|\cC_i|,|\cC_j|\}$ until there is no pair satisfying this condition. The distance $\tilde{O}(k) \cdot \min\{|\cC_i|,|\cC_j|\}$ is chosen to guarantee that the contribution of every $x_{\cC_{\ell}}$ in the total energy $\|x\|_{[-1,1]}^2$ is about $\|x_{\cC_{\ell}}\|_{[-1,1]}^2$ --- see Lemma~\ref{lem:total_energy} for a formal statement. 

Then we set a band $I_\ell$ for each cluster $\cC_\ell$ in the multiband approximation: 
\begin{equation}\label{eq:interval_I}
I_\ell=\left[ \min_{f \in \cC_{\ell}} f - \tilde{O}(k) \cdot |\cC_\ell|, \max_{f \in \cC_{\ell}} f + \tilde{O}(k) \cdot |\cC_\ell| \right].    
\end{equation}
In fact, the length $\tilde{O}(k) \cdot |\cC_\ell|$ is the improvement of this multiband approximation compared to previous works \cite{CKPS17,SSWZ23} whose length is always $\tilde{O}(k^2)$. The construction of $z$ is obtained by carefully designed filter functions $M_{\ell,\delta}$ where $\ell$ is the number of frequencies in the cluster and $\delta$ is the approximation error (see Lemma~\ref{lem:filter_sparse} in Section~\ref{sec:orthogonal}) and prove the truncation of $x$ in $[-1,1]$
\begin{equation}\label{eq:multiband_app}
\mathbf{1}_{[-1,1]} \cdot x  \approx M_{k,\epsilon} \cdot x \approx \sum_{\ell} M_{|\cC_\ell|,\delta_{\ell}} \cdot x_{\cC_\ell}.
\end{equation}
The crucial thing is to choose $\delta_{\ell}$ guaranteeing both \eqref{eq:multiband_app} and $M_{|\cC_\ell|,\delta_{\ell}}$ has a compact Fourier support of size $|\supp(\wh{M_{|\cC_\ell|,\delta_{\ell}}})| = \tilde{O}(k) \cdot |\cC_\ell|$ such that $\supp(\wh{M_{|\cC_\ell|,\delta_{\ell}} \cdot x_{\cC_\ell}}) \subset I_{\ell}$. This gives the multiband approximation guarantee in Theorem~\ref{thm:learn_Fourier_sparse}, whose proof is in Section~\ref{sec:multiband_approx}.

Back to the frequency recovery, \eqref{eq:multiband_app} implies that it suffices to recover intervals $I_{\ell}$ with large Fourier mass in $\sum_{\ell} M_{|\cC_\ell|,\delta_{\ell}} \cdot x_{\cC_\ell}$. Consider a cluster $\cC_{\ell}$ with  non-negligible $\|x_{\cC_\ell}\|_{[-1,1]}^2$, from the approximation of the Haar-Fourier basis and $M_{|\cC_\ell|,\delta_{\ell}} \cdot x_{\cC_\ell}$, there is at least one large Haar-Fourier coefficient falling into $I_\ell$. While this implies large Haar-Fourier coefficients are within distance  $\tilde{O}(k) \cdot |\cC_\ell|$ to $\cC_\ell$, one cannot guarantee intervals of length $O(K)=O(k)$ from large Haar-Fourier coefficients will cover $\cC_{\ell}$. Hence the frequency recovery algorithm (in Section~\ref{sec:freq_recovery}) has a postprocessing step: for each interval $J_\ell$ whose Haar-Fourier coefficients have energy $\mathcal{E}$, it expands $J_\ell$ to $J_\ell \pm \tilde{O}(\frac{\mathcal{E}}{\epsilon^2 \cdot \|x\|_{[-1,1]}^2/k^2})$. The reason is that it is enough to consider $x_{\cC_\ell}$ with $\|x_{\cC_\ell}\|_{[-1,1]}^2 \ge \frac{\epsilon \cdot |\cC_{\ell}|}{k} \cdot \|x\|_{[-1,1]}^2$ (otherwise $\cC_\ell$ is negligible). Since the Haar Fourier transform will approximate $\|x_{\cC_\ell}\|_{[-1,1]}^2$ and $I_\ell$ defined in \eqref{eq:interval_I} is of length $\tilde{O}(k) \cdot |\cC_\ell|$, stretching $\tilde{O}(\frac{\mathcal{E}}{\epsilon^2 \cdot \|x\|_{[-1,1]}^2/k^2})$ would cover $I_\ell$ and all frequencies in $\cC_\ell$. This part is shown in Section~\ref{sec:freq_recovery}.

\paragraph{Recovering Multibands.} The algorithm of Theorem~\ref{inform:learn_multiband} is based on the classical hash-to-bins framework. The key point is that while it can not guarantee different bands got mapped to the same bucket, it does not hurt us if two nearby bands are mapped to the same bucket. This is because we will return a superset of the multiband whose size is enlarged by a constant factor.

\paragraph{Proofs of orthogonality and frequency shifting.} The proofs of Lemma~\ref{inform:almost_orthogonal_clusters} and Lemma~\ref{lem:shift_one_freq} are based on the same idea. For Lemma~\ref{inform:almost_orthogonal_clusters}, instead of computing $\langle w,z\rangle$ directly, we consider two approximations $w' \approx w \cdot \mathbf{1}_{[-1,1]}$ and $z' \approx z \cdot \mathbf{1}_{[-1,1]}$ where $\mathbf{1}_{[-1,1]}$ is the indicator function on interval $[-1,1]$. This bounds $\langle w,z\rangle_{[-1,1]} \approx \langle w',z' \rangle = \langle \wh{w'}, \wh{z'} \rangle$. So our goal becomes to guarantee $\langle \wh{w'}, \wh{z'} \rangle=0$. This could be achieved by making their Fourier support disjoint. Given the separation between their frequencies, our task is to construct a filter function such that (1) $M w \approx w \cdot \mathbf{1}_{[-1,1]}$ and (2) its Fourier support is compact in $\tilde{O}(\frac{\ell r}{\delta})$. This is exactly the filter function constructed in \cite{CP19_ICALP}. A new ingredient is to use the extrapolation bound by Zhang \cite{zhang2026optimalextrapolationboundssparse} to reduce $|\supp(M)|$. The error in Lemma~\ref{lem:shift_one_freq} is bounded by a similar idea. One intuition of the factor $|f_{k+1}-f_{k}|$ is from the Paley-Wiener space. Roughly speaking, the error of replacing $f_k$  by $f_{k+1}$ is equivalent to the relative distance of moving $e^{2 \pi \bi f_{k+1} t}$ to get in parallel with $e^{2 \pi \bi f_{k} t}$. This moving is equivalent to a multiplication with $(f-f_k)$ in the Paley-Wiener space (such that the Fourier coefficient of $f_k$ becomes 0 to guarantee orthogonality), whose relative distance is at most $O(|f_{k+1}-f_k|)$.
We show this filter function $M$ and the proof of  Lemma~\ref{inform:almost_orthogonal_clusters} in Section~\ref{sec:orthogonal} and defer the proof of Lemma~\ref{lem:shift_one_freq} to Section~\ref{sec:proof_shifting} separately.

\section{Haar-Fourier Basis for Fourier-sparse signals}\label{sec:characterizations}
 
We prove Lemma~\ref{lem:sparse_in_Haar_dictionary} and discuss its extensions in this section. First, we define the Haar-Fourier basis used in Lemma~\ref{lem:sparse_in_Haar_dictionary}.

In this section, we always consider $x(t)$ in the time window $[-1,1]$ as a periodic signal over $[-1,1]$. Recall the Fourier transform of periodic signals are 
\begin{equation}
    \wh{x}[n]:= \frac{1}{\sqrt{2}} \cdot \int_{-1}^1 x(t) \cdot e^{-\bi \pi n t} \mathrm{d} t \quad \text{ and } \quad x(t)=\frac{1}{\sqrt{2}} \cdot\sum_{n=-\infty}^{+\infty} \wh{x}[n] e^{\bi \pi n t}.
\end{equation}

 Let $N := (F \cdot k/\delta)^{O(1)}$ be the frequency threshold determined later and the low-frequency Fourier series of size $2N$ be
\[
\mathcal{B}_{Fourier}:=\{e^{\bi \pi n t} / \sqrt{2}: \forall n \in \{-N,-N+1,\ldots,-1,0,1,\ldots,N-1\}\}.
\]
One could prove that the projection onto $\mathcal{B}_{Fourier}$ provides an $L_2$ approximation on $[-1, 1]$ for \eqref{eq:approximation_error} with a proper $N$ \cite{stein2011fourier} (based on properties about Fourier-sparse signals in Lemma~\ref{lemma:bounds_Fourier_sparse_signals}). However, some $k$-Fourier-sparse signals require the sparsity of approximations in $span(\mathcal{B}_{Fourier})$ to be $\tilde{\Omega}(k^2)$ from Property 1 of Lemma~\ref{lemma:bounds_Fourier_sparse_signals}.

To obtain a sparse approximation, the key is to consider a Haar transformation on $\mathcal{B}_{Fourier}$. Now we define the Haar transformation and the Haar-Fourier series $\mathcal{B}$ of Lemma~\ref{lem:sparse_in_Haar_dictionary} formally.

\begin{definition}\label{def:basis_B}
    Let $K \ge k$ be a 2-power. Define the basis $\mathcal{H}_K$ of the Haar transform over $[K]:=\{0,1,\ldots,K-1\}$ as follows: for every 2-power $\ell \in [2,K]$ and for every $s \in \{0,\ell,2\ell,\cdots,K-\ell\}$, let $\mathbf{1}_{s+[\ell/2]}$ be the indicator vector that is $1$ in $\{s,s+1,\ldots,s+\ell/2-1\}$ and $H_{s+[\ell]}^{wav}:=\frac{1}{\sqrt{\ell}} (\mathbf{1}_{s+[\ell/2]} - \mathbf{1}_{s+\ell/2+[\ell/2]})$ be a vector in $\mathcal{H}_K$. The last vector in $\mathcal{H}_K$ is $H^{sc}_{[K]}:=\frac{1}{\sqrt{K}} \mathbf{1}_{[K]}$. So there are $\log K + 1$ types of vectors in our Haar basis: shifts of $H^{wav}_{[\ell]}$ and shifts of $H^{sc}_{[K]}$. Let $W$ denote these types.

Let $N \ge 2F$ be a 2-power and $\{-N,-N+1,\ldots,-1,0,1,\ldots,N-1\}$ be the set of low frequencies. We partition these frequencies into $\frac{2N}{K}$ consecutive blocks of size $K$: $b_0 + [K], b_1 + [K], \ldots, b_{\frac{2N}{K}-1} + [K]$ for indices $b_i := -N+K \cdot i$. Then for each block $Q:=b_i + [K]$, let $\mathcal{H}_Q$ denote the basis that shifts the support set of $\mathcal{H}_K$ from $[K]$ to $Q=b_i + [K]$. 

Finally, we define the Haar basis $\mathcal{H}_{2N}:=\cup_{Q=b_i + [K]} \mathcal{H}_{Q}$, the Haar-Fourier vector of $H \in \mathcal{H}_Q$ as
\begin{equation}\label{eq:basis_B}
    \phi_{H}(t):= \tfrac{1}{\sqrt{2}}\sum_{n \in Q} H(n) \cdot (-1)^n \cdot e^{\bi \pi n t}, \text{ and Haar-Fourier basis }
    \mathcal{B}:= \left\{ \phi_H  \text{ for all } H \in \mathcal{H}_{2N} \right\}.
\end{equation}
\end{definition}

Here are a few basic facts about $\mathcal{H}_K$, $\mathcal{H}_{2N}$, and $\mathcal{B}$, whose proof is deferred to Section~\ref{sec:haar_facts}.
\begin{fact}\label{fact:Haar}
    $\mathcal{H}_K$ is an orthonormal basis of linear subspace of $\mathbb{C}^{[K]}$. Hence for any block $Q:=b+\{0,1,\ldots,K-1\}$, the linear subspace $span\{e^{\bi \pi n t}:n \in Q\}=span\{\phi_H: H \in \mathcal{H}_{Q} \}$ and for any signal $v$ over $[-1,1]$, $\sum_{n \in Q} |\hat{v}[n]|^2=\sum_{H \in \mathcal{H}_{Q}} |\langle \phi_H, v \rangle|^2_{[-1,1]}$. 
    
    Similarly, $\mathcal{H}_{2N}$ is an orthonormal basis of the linear subspace $\mathbb{C}^{\{-N,\ldots,N-1\}}$ such that $span\{\mathcal{B}_{Fourier}\}=span\{\mathcal{B}\}$
        and $\sum_{n \in \{-N,\ldots,N-1\}} |\hat{v}[n]|^2=\sum_{\phi_H \in \mathcal{B}} |\langle \phi_H, v \rangle|^2_{[-1,1]}$ for any $v$.
    Moreover,
    \begin{equation}\label{eq:haar-geometric-formulas}
    \begin{aligned}
        \phi_{H^{sc}_{b+[K]}}(t)
        &=\frac{1-e^{\bi\pi K(t+1)}}{\sqrt{2 K}(1-e^{\bi\pi(t+1)})}e^{\bi\pi b(t+1)},\\
        \phi_{H^{wav}_{s+[\ell]}}(t)
        &=\frac{(1-e^{\bi\pi\ell(t+1)/2})^2}{\sqrt{2 \ell}(1-e^{\bi\pi(t+1)})} e^{\bi\pi s(t+1)}.
    \end{aligned}
    \end{equation}
\end{fact}

For convenience, let $P_N$ denote the projection onto $\mathcal{B}$ in the rest of this section, which is equivalent to the projection onto $\mathcal{B}_{Fourier}$ from the above fact.

\begin{equation}\label{def:Project_PN}
    P_N v = \sum_{\phi_H \in \mathcal{B}} \langle v, \phi_H \rangle_{[-1,1]} \cdot \phi_H =\sum_{n \in \{-N,\ldots,N-1\}} \wh{v}[n] e^{\bi \pi n t} / \sqrt{2} .
\end{equation}

While we will prove that $P_N x$ is an $L_2$-approximation of any Fourier-sparse signal $x$, the statistical upper bound needs a stronger $\ell_{\infty}$-approximation. One caveat here is a $k$-Fourier-sparse signal $x(t)$ might not satisfy $x(-1) = x(1)$ such that $\mathcal{B}$ (nor $\mathcal{B}_{Fourier}$) does not guarantee an $\ell_{\infty}$-approximation. Hence we add an extra $\iota(t):=t$ for the $\ell_{\infty}$-approximation. The intuition of $\iota(t)$ is that $x_p(t):=x(t) - \frac{x(1) - x(-1)}{2} \cdot \iota(t)$ satisfies $x_p(1)=x_p(-1)$ as a periodic signal. One more ingredient is that the $\ell_{\infty}$-approximation is with respect to the weighted sampling $D_k(t)=\frac1{(2+4\ln k)\max\{1-|t|,k^{-2}\}}$.

Here is the full statement of Lemma~\ref{lem:sparse_in_Haar_dictionary}.
\begin{theorem}\label{thm:approximation}
    $\mathcal{B}$ satisfies the following properties for a universal constant $C$ (with different $N$):     \begin{enumerate}
        \item\label{item:approx-sparsity} For any $k$-Fourier-sparse signal $x(t)$, the coefficient vector $\alpha_{H}:=\langle x, \phi_H \rangle_{[-1,1]}$ over all $\phi_H \in \mathcal{B}$ is analytically sparse: 
        \begin{equation}\label{eq:sum_coeff_Haar}
            \sum |\alpha_H| = O(\sqrt{k} \cdot \log N) \cdot (\sum |\alpha_H|^2)^{1/2}.
        \end{equation}
        
        \item\label{item:approx-l2} When $N \ge 2F+ C \cdot \frac{k^2}{\delta^2}$, for any $k$-Fourier-sparse signal $x(t)$, $\|P_N x - x\|_{[-1,1]} \le \delta \cdot \|x\|_{[-1,1]}$.

        \item\label{item:approx-weighted} For any $k$-Fourier-sparse signal $x(t)$, let signal $\iota^{\bot}:=\iota - P_N \iota$ be the orthogonal component of $\iota(t)$ to $span(\mathcal{B})$ and $g_N:=P_N x + \frac{x(1)-x(-1)}{2} \cdot \iota^{\bot}$ be the approximation of $x$ in $span\{\mathcal{B} \cup \{\iota\}
        \}$. When $N \ge  C (\frac{(F+k^2)\sqrt{\log k}}{\delta})$, $g_N$ is an $\ell_{\infty}$-approximation with respect to $D_k$ defined in \eqref{eq:weighted-dist}
        \[
        \max_{t \in [-1,1]} \frac{|x(t)-g_N(t)|}{\sqrt{D_k(t)}} \le \delta \cdot \|x\|_{[-1,1]}.
        \] 
        
        \item\label{item:approx-conditioning} All $\phi_H$ in $\mathcal{B}$ and $\iota^{\bot}$ are well-conditioned with respect to $D_k$: 
        \[
           \max_{t \in [-1,1]} \frac{|\phi_H(t)|}{\sqrt{D_k(t)}} = O( \sqrt{\log 2k}) \text{ and } \max_{t \in [-1,1]} \frac{|\iota^\bot(t)|}{\sqrt{D_k(t)}} = O(\frac{\sqrt{\log 2k}}{k}),
        \]
        where $|H|$ denotes the size of the support of $H$.
        In particular, $|\phi_H(t)| = O(1) \cdot \min\left\{\sqrt{|H|},\frac1{\sqrt{|H|}(1-|t|)}\right\}$.
    \end{enumerate}
\end{theorem}

Given the $\ell_{\infty}$ approximation in Theorem~\ref{thm:approximation}, the statistical upper bound follows from the classical restricted isometry property \cite{RV08}. We provide the proof of Theorem~\ref{inform:stat_bound} about statistical upper bounds in Section~\ref{sec:stat_upper_bound}.

The key to proving the first property, the coefficient vector $\alpha_H$ is analytically sparse, is the monotonicity of $(-1)^n \wh{x}(n)$ \cite{Boc69, Efr99}. Formally, we call a subsequence $b_i,b_{i+1},\ldots,b_{j-1}, b_{j}$ is a \emph{monotone run} iff both their imaginary parts and their real parts are  monotone increasing or decreasing separately. The key property of the Haar basis is that it transforms a monotone sequence into an analytically sparse sequence \cite{Boc69}. We show the formal statement for $b_n:=(-1)^n \cdot \wh{x}[n]$.
\begin{lemma}\label{lem:monotone}
    For any $k$-Fourier-sparse signal, sequence $b_n:=(-1)^n \cdot \wh{x}[n]$ for $\wh{x}[n]=\frac{1}{\sqrt{2}} \int_{-1}^1 x(t) e^{-\bi \pi n t}\mathrm{d} t$ and $n \in \{-N,\ldots,N-1\}$ has $O(k)$ monotonicity runs. 
    
    Then we fix a 2-power $\ell \in [2,K]$ and a type of Haar vector $H^{wav}_{[\ell]}=\frac{1}{\sqrt{\ell}} (1_{[\ell/2]}-1_{\ell/2+[\ell/2]})$. The sum of the coefficients of $b_n$ over $H \in \mathcal{H}_{2N}$ of shifts $H^{wav}_{s + [\ell]}=\frac{1}{\sqrt{\ell}} (1_{s+[\ell/2]}-1_{s+\ell/2+[\ell/2]})$ is

    \begin{equation}\label{eq:sum_fixed_ell_abs_coef}
        \sum_{H: \text{ shifts of } H^{wav}_{[\ell]}} \left|\sum_n b_n \cdot H[n]\right| = O(\sqrt{k}) \cdot \left(\sum_{n=-N}^{N-1} |b_n|^2\right)^{1/2}.
    \end{equation}
\end{lemma}
Because $\sum_n b_n \cdot H(n)=\langle x,\phi_H\rangle$ by the definition of $b_n$ and $\phi_H$, Lemma~\ref{lem:monotone} bounds the summation of $|\langle x,\phi_H\rangle|$ in \eqref{eq:sum_coeff_Haar} of Theorem~\ref{thm:approximation} except coefficients related to $H^{sc}_{[K]}$. Because these Haar vectors are of the form $\frac{1}{\sqrt{K}} \mathbf{1}_{b_i + [K]}$, we use $\sum_n |\hat{x}[n]|$ to bound it.
\begin{lemma}\label{lem:sum_coef}
    For any $k$-Fourier-sparse signal $x$ and any $N > 1$,
    \begin{equation}\label{eq:sum_abs_coef}
        \sum_{n:|n| \le N} |\wh{x}[n]|=O(k \cdot \log N) \cdot \|x\|_{[-1,1]}.
    \end{equation}
\end{lemma}
We provide the full proof of Theorem~\ref{thm:approximation} in Section~\ref{sec:proof_thm_approx}. Proofs of Lemma~\ref{lem:monotone} and Lemma~\ref{lem:sum_coef} are deferred to Section~\ref{sec:proofs_ell_1_Haar_Fourier}.

\subsection{Statistical Upper Bound}\label{sec:stat_upper_bound}
We prove the statistical upper bound in Theorem~\ref{inform:stat_bound} that $\tilde{O}(k)$ samples are sufficient to estimate $\|x\|_{[-1,1]}^2$ for all $k$-Fourier-sparse signals $x$. This proof relies on the restricted isometry property from Rudelson and Vershynin \cite{RV08}.
\begin{theorem}\label{thm:RIP}
    Given a distribution $D$ over $[-1,1]$, let $v_1,\ldots,v_n$ be an orthonormal series with $\langle v_i, v_j\rangle_{[-1,1]}=\mathrm{I}(i=j)$ and $\max_{t \in [-1,1]} \frac{|v_i(t)|}{\sqrt{D(t)}} \le \beta$. Given a sparsity parameter $\tau$, let $\mathcal{F}:=\{\alpha_1 v_1 + \cdots+\alpha_n v_n: \sum_{i=1}^n |\alpha_i| \le \tau \cdot (\sum_{i=1}^n |\alpha_i|^2)^{1/2} \}$ be the sparse family. 
    
    For any error $\delta$, let $m=C \cdot \frac{\tau^2 \beta^2 \log^4 n}{\delta^2}$ and $i_1,\ldots,i_m$ be $m$ random points sampled from $D$. Then with probability 0.98,
    \[
    \forall u \in \mathcal{F}, \frac{1}{m} \left( \frac{|u(i_1)|^2}{D(i_1)} + \cdots + \frac{|u(i_m)|^2}{D(i_m)}
    \right) = (1\pm \delta) \cdot \|u\|_{[-1,1]}^2.
    \]
\end{theorem}
\begin{remark}
    There are two remarks. First, the sampling distribution $D$ does not affect the statement after normalizing. Second, while Theorem 3.3 in \cite{RV08} states RIP for $r$-sparse vectors, its proof holds for analytic sparse vectors. Specifically, because its covering bound in Lemma 3.7 holds for $\ell_1$-bounded coefficient vectors $(\alpha_1,\ldots,\alpha_n)$, the whole proof holds for any family whose coefficient vectors are $\ell_1$-bounded.
\end{remark}

\begin{proof}[Proof of the statistical upper bound in Theorem~\ref{inform:stat_bound}]
For convenience, this proof uses 
\[\|x\|_{D_k, m}^2:=\frac{1}{m} \sum_{j=1}^m \frac{|x(i_j)|^2}{D_k(i_j)}.\]
Let $N$ be the smallest 2-power satisfying $ N\ge \frac{C F^2k^4 \log k}{\epsilon^2}$.
Recall that
\begin{align*}
P_N x =\sum_{n \in \{-N,\ldots, N-1\}} \wh{x}[n] e^{\bi \pi n t} / \sqrt{2} = \sum \alpha_H \phi_H,
\quad
\iota^{\bot} = \iota - P_N \iota, \\
~\text{and}~ \quad
g_N = P_N x + \frac{x(1)-x(-1)}{2} \cdot \iota^{\bot}
\end{align*}
where $\{\phi_H\}$ is an orthonormal basis defined in Section~\ref{sec:characterizations}. Moreover, Theorem~\ref{thm:approximation} shows
\begin{align*}
\sum |\alpha_H| = O(\sqrt{k} \cdot \log N) \cdot (\sum |\alpha_H|^2)^{1/2} \\
\max_{t \in [-1,1]} \frac{|\phi_H(t)|}{\sqrt{D_k(t)}} = O( \sqrt{\log 2k}) \\
\|P_N x - x\|_{[-1,1]}^2 \lesssim \epsilon \cdot \|x\|_{[-1,1]}^2.
\end{align*}
Applying Theorem~\ref{thm:RIP} with $D=D_k$ and $\delta=O(\epsilon)$, we obtain, with probability at least $0.99$,
\begin{align}
    \Big| \norm{P_Nx}_{D_k, m} - \|x\|_{[-1,1]} \Big| & \lesssim \Big| \norm{P_Nx}_{D_k, m} - \|P_Nx\|_{[-1,1]} \Big| +  \Big| \|x\|_{[-1,1]} - \|P_Nx\|_{[-1,1]} \Big| \notag \\
    & \lesssim  \epsilon \cdot \|P_Nx\|_{[-1,1]} + \epsilon \cdot \|x\|_{[-1,1]} \notag \\
    & \lesssim \epsilon \cdot \|x\|_{[-1,1]}. \label{eq:stat-projection-RIP}
\end{align}
simultaneously for every signal $P_Nx \in \mathcal{F}$.
And $m$ is at most $O\Big(\frac{k}{\epsilon^2}\log^7\frac{kF}{\epsilon}\Big)$.

It remains to control $\Big| \norm{P_Nx}_{D_k, m} - \norm{x}_{D_k, m} \Big|$. 
The triangle inequality and Cauchy--Schwarz inequality yield
\begin{align}
    \|x\|_{D_k, m} & \le \|g_N\|_{D_k, m} + \|x - g_N\|_{D_k, m}. \label{eq:stat-projection-RIP2}
        \end{align}
And Property~\ref{item:approx-weighted} of Theorem~\ref{thm:approximation} shows
\[
    \max_{t\in[-1,1]}\frac{|x(t)-g_N(t)|}{\sqrt{D_k(t)}} \lesssim \epsilon\|x\|_{[-1,1]},
\]
which implies $\|x - g_N\|_{D_k, m} \lesssim \epsilon\|x\|_{[-1,1]}$.

Again, we have
\begin{align}
    \|g_N\|_{D_k, m} & \le \|P_N x\|_{D_k, m} + \|g_N - P_N x\|_{D_k, m}, \label{eq:stat-projection-RIP3}
\end{align}
where $g_N-P_Nx=\frac{x(1)-x(-1)}2\iota^\bot$ is a scalar multiple of the fixed vector $\iota^\bot$.
Since the Fourier coefficients of $\iota(t)=t$ have magnitude $\sqrt{2}/(\pi|n|)$ for $n \ne 0$, Parseval's identity implies
\[
\|\iota^\bot\|_{[-1,1]}^2 \le \frac{2}{\pi^2} \sum_{|n|\ge N} \frac1{n^2} \lesssim \frac{1}{N}.
\]
And $\E \norm{\iota^\bot}_{D_k, m}^2 =\|\iota^\bot\|_{[-1,1]}^2$.
Markov's inequality therefore shows that, with probability at least $0.999$, $\norm{\iota^\bot}_{D_k, m}^2 \lesssim \frac{1}{N}$.
Also, Property~\ref{item:uniform_bound_on_interval} of Lemma~\ref{lemma:bounds_Fourier_sparse_signals} gives $|x(1)-x(-1)|/2 \lesssim k\|x\|_{[-1,1]}$.
So $\|g_N - P_N x\|_{D_k, m}^2 \lesssim \frac{k^2}{N} \|x\|_{[-1,1]}^2$.

Consequently, combining the above estimates with \eqref{eq:stat-projection-RIP}, \eqref{eq:stat-projection-RIP2} and \eqref{eq:stat-projection-RIP3}, we obtain
\[
\Big| \norm{x}_{D_k, m}^2 - \|x\|_{[-1,1]}^2 \Big| \lesssim \epsilon \cdot \|x\|_{[-1,1]}^2
\]
with probability 0.98.
\end{proof}

\subsection{Proofs of Lemma~\ref{lem:monotone} and Lemma~\ref{lem:sum_coef}} \label{sec:proofs_ell_1_Haar_Fourier} We first prove Lemma~\ref{lem:monotone}.

\begin{proofof}{Lemma~\ref{lem:monotone}}   
Recall that \eqref{eq:simplification_b_n} shows $b(n):=(-1)^n \, \widehat x(n)$ is a rational function:
\[
 b(n) = \sum_{j=1}^k \alpha_j \cdot \frac{\sqrt{2} \sin(2 \pi f_j)}{2 \pi f_j - \pi n} =\frac{A(n)}{Q(n)}
\]
where $Q(n)=\prod_{j=1}^k (2 \pi f_j- n \pi)$ is real and $\deg A < \deg Q \le k$. $b_n$ can be partitioned into $O(k)$ monotone subsequence: First of all, we split the sequence at $f_j$ if $2f_j \in \mathbb{Z}$. Then we consider the real part $\frac{Re A(n)}{Q(n)}$ and the imaginary part separately. Observe that $\partial \frac{Re A(x)}{Q(x)}$ has at most $2k-1$ roots. Because it is monotone between 2 consecutive roots, $\frac{Re A(n)}{Q(n)}$ has $O(k)$ monotone runs. Similarly, the imaginary part has $O(k)$ monotone runs such that the total number of monotone runs is $O(k)$.

Without loss of generality, suppose the sequence is real. Otherwise, we apply the following argument to its real and imaginary parts separately.

Because $\ell$ is given, we partition $-N,\ldots,N-1$   into $2N/\ell$ blocks such that each shift of $H^{wav}$ is on one block. We call a block $B$ crossing if $\big( b(n) \big)_{n \in B}$ contains a boundary between two monotone runs of $b(n)$ and clean otherwise.  Since there are at most $O(k)$ crossing blocks, Cauchy--Schwarz's inequality shows
\[
    \sum_{\text{crossing blocks } B}|\ip b{H_B}|
    \leq\sqrt{O(k)} \left(\sum_{\text{crossing blocks } B}|\ip b{H_B}|^2\right)^{1/2}
    \lesssim \sqrt k \,\norm b_2.
\]

For clean blocks of $b(n)$, we consider a maximal consecutive sequence of blocks contained in one nonincreasing run from $b(p)$ to $b(q)$.  
By the definition of $H_B$, we have
\begin{align*}
  \sqrt \ell \sum_{B \subset \{p,\ldots,q\}} |\ip b{H_B}|
  & \leq\sum_{n=p}^{q- \ell/2 }(b(n)-b(n+\ell/2))
  \\
  & =\sum_{n=p}^{p+\ell/2-1}b(n)-\sum_{n=q-\ell/2+1}^q b(n) \lesssim \sqrt{\ell} \cdot (|b(p)|^2 + \cdots + |b(q)|^2)^{1/2},
\end{align*}
Since there are at most $O(k)$ monotone runs,
\begin{align*}
    \sum_{B}|\ip b{H_B}| = \sum_{B\ \mathrm{crossing}}|\ip b{H_B}| + \sum_{r:\text{ monotone run}}\sum_{B\ \mathrm{ in\ } r}|\ip b{H_B}| \lesssim \sqrt k \,\norm b_2.
\end{align*}

\end{proofof}

Then we finish the proof of Lemma~\ref{lem:sum_coef}, which relies on the Riesz representation. 
\begin{lemma}[Riesz representation]
\label{lem:finite-riesz}
If $V$ is a nonzero finite-dimensional complex Hilbert space and $L:V\to\mathbb C$ is linear, there is $e_L\in V$ such that $L(u)=\ip u{e_L}$ for every $u\in V$, and $\|e_L\|=\sup_{\|u\|=1}|L(u)|$.
\end{lemma}

\begin{proofof}{Lemma~\ref{lem:sum_coef}}
Let $x(t)=\sum_{j=1}^k \alpha_je^{2\pi\bi f_jt}$. 
For $z\in\operatorname{span}\{e^{-2\pi\bi f_jt}:1\leq j\leq k\}$,
we consider a kernel $g_+$ supported on $[0,1]$ such that
\begin{equation}\label{eq:reproducing-kernel}
\int_0^1g_+(s)z(s)\,ds=z(0).
\end{equation}
And we extend $g_+$ by zero on $[-1, 0)$.  Lemma~\ref{lem:finite-riesz} and Property~\ref{item:uniform_bound_on_interval} of Lemma~\ref{lemma:bounds_Fourier_sparse_signals} imply $\|g_+\|_{[-1,1]}\lesssim k$. Basically, let $\psi_1,\ldots,\psi_k$ be the orthonormal basis of $\operatorname{span}\{e^{-2\pi\bi f_jt}:1\le j\le k\}$ in $(0,1)$. Then $g_+:=\sum_{j=1}^k \psi_j(0) \cdot \overline{\psi_j}$. Because $\|g_+\|_{[0,1]}^2 = \sum_j |\psi_j(0)|^2$ and $|g_+(0)|=\sum_j |\psi_j(0)|^2$ again, we have $\|g_+\|_{[-1,1]}\lesssim k$.

Similarly, we define another kernel $g_-$ with the following properties:
\begin{equation}\label{eq:reproducing-kernel-2}
\int_{-1}^0 g_-(s)z(s)\,ds=z(0), \quad g_-(t) = 0 ~\text{for}~ t \in (0, 1], \quad \|g_-\|_{[-1,1]}\lesssim k.
\end{equation}

Moreover, we define
\[
 y_+(t):=\int_{-1}^1x(t-s)g_+(s)\,ds ~\text{ and }~ y_-(t):=\int_{-1}^1x(t-s)g_-(s)\,ds.
\]
For $t,s\in[0,1]$, their difference $t-s$ lies in $[-1,1]$. 
Since $x(t-s)$ lies in $\operatorname{span}\{e^{-2\pi\bi f_js}:1\leq j\leq k\}$, $y_+(t)=x(t)$ on $[0,1]$ by \eqref{eq:reproducing-kernel}. 
The same reasoning for $t,s\in[-1,0]$ gives $y_-(t)=x(t)$ on $[-1,0]$. 
We thus obtain
\begin{equation}\label{eq:linear-two-halves}
x=\mathbf1_{[0,1]}y_++\mathbf1_{[-1,0)}y_-.
\end{equation}

For $\widehat y_\pm[n]=\sqrt2\,\widehat x[n]\widehat g_\pm[n]$, Cauchy--Schwarz and Parseval imply both of them satisfy
\begin{align}
\sum_{n\in\Z}|\widehat y_\pm[n]|
\leq \sqrt{2} \left(\sum_n|\widehat x[n]|^2\right)^{1/2} \left(\sum_n|\widehat g_\pm[n]|^2\right)^{1/2}
= \sqrt{2} \|x\|_{[-1,1]}\|g_\pm\|_{[-1,1]}
\lesssim k\|x\|_{[-1,1]}.
\label{eq:convolution-fourier-mass}
\end{align}
Direct integration shows
\[
\widehat{\mathbf1_{[0,1]}}[0]= \frac{1}{\sqrt{2}},\qquad
\widehat{\mathbf1_{[0,1]}}[m]=\frac{1-(-1)^m}{\bi\sqrt{2} \pi m}\quad(m\ne0).
\]
Hence,
\begin{equation}\label{eq:indicator-fourier-mass}
\sup_{\ell\in\Z}\sum_{|n|\leq N}|\widehat{\mathbf1_{[0,1]}}[n-\ell]|
\leq \frac{1}{\sqrt{2}}+\frac{2 \sqrt{2}}{\pi}\sum_{m=1}^N\frac1m
\lesssim\log N.
\end{equation}
By $\widehat{\mathbf1_{[0,1]}y_+}[n] = \frac{1}{\sqrt{2}}\sum_{\ell\in\Z}\widehat y_+[\ell]\widehat{\mathbf1_{[0,1]}}[n-\ell]$,
\eqref{eq:convolution-fourier-mass} and \eqref{eq:indicator-fourier-mass} yield
\[
\sum_{|n|\leq N}|\widehat{\mathbf1_{[0,1]}y_+}[n]|
\lesssim\log N \cdot \sum_{\ell\in\Z}|\widehat y_+[\ell]|
\lesssim k\log N \cdot \|x\|_{[-1,1]}.
\]
And the same argument bounds $\sum_{|n|\leq N}|\widehat{\mathbf1_{[-1,0]}y_-}[n]|$. 
Applying these to \eqref{eq:linear-two-halves} proves
\[
\sum_{n:|n| \le N} |\widehat{x}[n]|\lesssim k\log N \cdot \|x\|_{[-1,1]}.
\]
\end{proofof}

\subsection{Proof of Fact~\ref{fact:Haar}}\label{sec:haar_facts}

By the definition of $\mathcal{H}_K$, $H_{s+[\ell]}^{wav}:=\frac{1}{\sqrt{\ell}} (\mathbf{1}_{s+[\ell/2]} - \mathbf{1}_{s+\ell/2+[\ell/2]})$ and $H^{sc}_{[K]}:=\frac{1}{\sqrt{K}} \mathbf{1}_{[K]}$. So $\norm{H_{s+[\ell]}^{wav}}_2 = \norm{H^{sc}_{[K]}}_2 = 1$. 
Distinct vectors of the same type have disjoint supports. Vectors of different types either have disjoint supports, or the smaller wavelet is supported where the other vector is constant
The Haar transform is therefore an orthonormal transformation on $\R^{[K]}$. 

This implies $\sum_{n \in Q} |\hat{v}[n]|^2=\sum_{H \in \mathcal{H}_{Q}} |\langle \phi_H, v \rangle|^2_{[-1,1]}$. The same coefficient-energy identity holds for $\mathcal H_{2N}$.

With $(-1)^n = e^{\bi \pi n}$, we have
\begin{align*}
\phi_{H^{sc}_{b+[K]}}(t) & = \tfrac{1}{\sqrt{2K}}\sum_{n=b}^{K+b-1} (-1)^n \cdot e^{\bi \pi n t} 
= \tfrac{e^{\bi \pi b (t+1)}}{\sqrt{2K}} \sum_{n=0}^{K-1} e^{\bi \pi n (t+1)} \\
& = \frac{1-e^{\bi\pi K(t+1)}}{\sqrt{2 K}(1-e^{\bi\pi(t+1)})}e^{\bi\pi b(t+1)}, \\
\phi_{H^{wav}_{s+[\ell]}}(t) & = \tfrac{1}{\sqrt{2\ell}} \Big[ \sum_{n=s}^{s+\ell/2-1} (-1)^n \cdot e^{\bi \pi n t} + \sum_{n=s+\ell/2}^{s+\ell-1} (-1)^n \cdot e^{\bi \pi n t} \Big] \\
& = \tfrac{(1-e^{\bi\pi\ell(t+1)/2}) e^{\bi\pi s(t+1)}}{\sqrt{2 \ell}} \sum_{n=0}^{\ell/2-1} e^{\bi \pi n (t+1)} \\
& = \frac{(1-e^{\bi\pi\ell(t+1)/2})^2}{\sqrt{2 \ell}(1-e^{\bi\pi(t+1)})} e^{\bi\pi s(t+1)}.
\end{align*}

\subsection{Proof of Theorem~\ref{thm:approximation}}\label{sec:proof_thm_approx}
First of all, we use Lemma~\ref{lem:monotone} and Lemma~\ref{lem:sum_coef} to prove the first property that $\alpha_H$ is analytically sparse.
Then we prove the approximation guarantee in the second property and the third property.
Finally, we verify that $\phi_H$ and $\iota^{\bot}$ are well-conditioned.

\subsubsection{Proof of Property~\ref{item:approx-sparsity} in Theorem~\ref{thm:approximation}}

Let $b_n=(-1)^n\widehat x[n]$ for $n\in \{-N,\ldots,N-1\}$. 
By Lemma~\ref{lem:monotone}, each type of the Haar-Fourier vectors have a summation of coefficients at most $O(\sqrt k)\|b\|_2$ to $\sum_H|\alpha_H|$. 
Since there are $\log K$ scales, their total contribution is at most
\begin{equation}\label{eq:wavelet-coefficient-mass}
 \sum_{H\text{ wavelet}}|\alpha_H|
 \lesssim\sqrt k\log K\|b\|_2.
\end{equation}
For the scaling vectors $H^{sc}$, Lemma~\ref{lem:sum_coef} and $K\geq k$ imply
\begin{equation}\label{eq:scaling-coefficient-mass}
 \sum_Q|\alpha_{H_Q^{sc}}|
 \leq \frac1{\sqrt{K}}\sum_{n\in \{-N,\ldots,N-1\}}|\widehat x[n]|
 \lesssim \sqrt k \log N\|x\|_{[-1,1]}.
\end{equation}

With orthogonality, we have
\begin{equation}\label{eq:retained-energy}
\sum_H|\alpha_H|^2= \|b\|_2^2 =\|P_Nx\|_{[-1,1]}^2 \gtrsim \|x\|_{[-1,1]}^2.
\end{equation}
Substituting this into \eqref{eq:wavelet-coefficient-mass} and \eqref{eq:scaling-coefficient-mass} proves \eqref{eq:sum_coeff_Haar}.

\subsubsection{Proofs of Property~\ref{item:approx-l2} and Property~\ref{item:approx-weighted} in Theorem~\ref{thm:approximation}}

\begin{lemma} \label{lem:ell-2-approximation}
For any $k$-Fourier-sparse signal $x$ and any $N \ge 2 F + \Omega(k^2)$,
\begin{equation}\label{eq:l2-tail-bound}
\|x-P_Nx\|_{[-1,1]}^2 \lesssim\frac{k^2}{N-2F}\|x\|_{[-1,1]}^2.
\end{equation}
\end{lemma}

For a sufficiently large constant, the cutoff satisfies the hypotheses of Lemma~\ref{lem:ell-2-approximation} and makes the right-hand side of \eqref{eq:l2-tail-bound} at most $\delta^2\|x\|_{[-1,1]}^2$. 
Taking square roots proves the property.

\begin{proof}[Proof of Lemma~\ref{lem:ell-2-approximation}]
Given $x = \sum_{j=1}^k \alpha_j e^{2\pi\bi f_jt}$, let $\mathcal V_x=\operatorname{span}\{e^{2\pi\bi f_jt}:1\leq j\leq k\}$. 
Define operator $B := -\frac{\bi}{2\pi}\frac d{dt}$ from $\mathcal V_x$ to $\mathcal V_x$ with
\begin{equation} \label{eq:derivative-from-B}
  \|x'\|_{[-1,1]}=2\pi\|Bx\|_{[-1,1]}
  \leq2\pi\|B\|_{\op}\|x\|_{[-1,1]}.
\end{equation}

We take a unitary Schur triangularization $U^*BU = D + T$, where $D$ is diagonal and $T$ is strictly upper triangular. 

Since $B(e^{2\pi\bi f_jt})=f_j e^{2\pi\bi f_jt}$ for all $f_j$ and $e^{2\pi\bi f_1 t},\ldots,e^{2\pi\bi f_kt}$ span the whole space $\mathcal V_x$, the eigenvalues of $B$ are $f_1,\ldots,f_k \in [-F,F]$. Hence the diagonal entries of $D$ are real and lie in $[-F,F]$. 
Therefore $D^*=D$ and $\|D\|_{\op}\leq F$.

The matrices $T$ and $T^*$ have disjoint off-diagonal supports, and $U^*(B-B^*)U=T-T^*$. Let $\norm{A}_{\HS}^2 := \Tr [A^* A]$ for any matrix $A$. Unitary invariance therefore shows
\begin{align*}
    \norm{T}_{\HS}^2 = \frac{1}{2} \norm{T - T^*}_{\HS}^2 = \frac{1}{2} \norm{B - B^*}_{\HS}^2.
\end{align*}

To bound the strictly triangular part, we express $B - B^*$ via endpoints evaluations. 
For $u$ and $v\in\mathcal V_x$, integration by parts gives
\begin{align}
  \ip{(B - B^*) u}{v}_{[-1,1]} & = \ip{Bu}{v}_{[-1,1]}-\ip{u}{Bv}_{[-1,1]} \notag \\
  & = -\frac{\bi}{2\pi} \Bigg[ \int_{-1}^{1} \Big[\frac{d}{dt} u(t)\Big] \overline{v(t)} d t + \int_{-1}^{1} u(t) \overline{\Big[\frac{d}{dt} v(t)\Big]} d t \Bigg] \notag \\
    & =  -\frac{\bi}{2\pi} \Bigg[ u(t) \overline{v(t)} \bigg|_{-1}^1 - \int_{-1}^{1} u(t) d \overline{v(t)} + \int_{-1}^{1} u(t) \overline{\Big[\frac{d}{dt} v(t)\Big]} d t \Bigg] \notag \\
    & = -\frac{\bi}{2\pi}\bigl(u(1)\overline{v(1)} - u(-1)\overline{v(-1)}\bigr). \label{eq:B-boundary-identity}
\end{align}

We apply Lemma~\ref{lem:finite-riesz} about the Riesz representation to the measurements of two endpoints. Let $e_+,e_-\in\mathcal V_x$ be the representing vectors of two linear functions $u \mapsto u(1)$ and $u \mapsto u(-1)$, so for $u\in\mathcal V_x$,
\[
  u(1)=\ip u{e_+}, \qquad u(-1)=\ip u{e_-}.
\]
Similar to the proof in Lemma~\ref{lem:sum_coef}, Property~\ref{item:uniform_bound_on_interval} of Lemma~\ref{lemma:bounds_Fourier_sparse_signals} implies
\begin{equation} \label{eq:endpoint-vector-bound}
  \|e_+\|_{[-1,1]} =\sup_{0\ne u\in\mathcal V_x}\frac{|u(1)|}{\|u\|_{[-1,1]}} \lesssim k \quad \text{ and } \quad
  \|e_-\|_{[-1,1]} = \sup_{0\ne u\in\mathcal V_x}\frac{|u(-1)|}{\|u\|_{[-1,1]}} \lesssim k.
\end{equation}

The boundary identity \eqref{eq:B-boundary-identity} gives
\[
 (B-B^*)u=-\frac{\bi}{2\pi}
       \bigl(\langle u,e_+\rangle e_+-\langle u,e_-\rangle e_-\bigr).
\]
It follows that the operator norm of $T$, denoted by $\|T\|_{\op}$, has
\begin{equation}
\|T\|_{\op} \leq \norm{T}_{\HS}
\lesssim \norm{B - B^*}_{\HS}
\lesssim \|e_+\|_{[-1,1]}^2+\|e_-\|_{[-1,1]}^2
\lesssim k^2.
\label{eq:skew-part-bound}
\end{equation}

Combining $\|D\|_{\op}\leq F$ with \eqref{eq:skew-part-bound} shows
\begin{equation}\label{eq:differentiation-operator-bound}
    \norm{B}_{\op}
    \le \norm{D}_{\op}+\norm{T}_{\op}
    \lesssim F + k^2. 
\end{equation}

Next, integration by parts gives, for every $u\in\mathcal V_x$
\begin{align}
\widehat{Bu}[n]
& = -\frac{\bi}{2\sqrt{2}\pi}\int_{-1}^1u'(t)e^{-\bi\pi nt}\,dt\notag\\
& = \frac n2\widehat u[n] - \frac{\bi(-1)^n}{2\sqrt{2}\pi}(u(1)-u(-1)).
\label{eq:fourier-boundary-recurrence}
\end{align}

Applying \eqref{eq:fourier-boundary-recurrence} to $u:=(nI-2B)^{-1}x$ implies
\begin{equation}\label{eq:sharp-fourier-boundary}
\widehat x[n]
= \frac{\bi(-1)^n}{\sqrt{2}\pi}
\left(\bigl((nI-2B)^{-1}x\bigr)(1) -\bigl((nI-2B)^{-1}x\bigr)(-1)\right).
\end{equation}

Let $G_n=nI-2D$. Because $\|D\|_{op} \le F$, $\|G_n^{-1}\|_{\op}\leq(|n|-2F)^{-1}$ when $|n|-2F \gtrsim k^2$. 
And by \eqref{eq:skew-part-bound}, $\|2G_n^{-1}T\|_{\op}\leq1/2$ for sufficiently large $|n|$. 
Because $G_n-2T=G_n(I-2G_n^{-1}T)$ and $U^*BU = D + T$,
\begin{align*}
(nI-2B)^{-1} & = U (nI - 2D - 2T)^{-1} U^*\\
& = U (I-2G_n^{-1}T)^{-1} G_n^{-1} U^* \\
& = U \left(\sum_{r=0}^{\infty}(2G_n^{-1}T)^r\right)G_n^{-1} U^*.   
\end{align*}
Thus $\|(nI-2B)^{-1}\|_{\op} \leq \|G_n^{-1}\|_{\op} \cdot \sum_{r=0}^{\infty} \|2G_n^{-1}T\|_{\op}^r \leq 2/(|n|-2F)$.
The vector $(nI-2B)^{-1}x$ still lies in $\mathcal V_x$, so applying the endpoint bound to this vector in \eqref{eq:sharp-fourier-boundary} proves
\begin{equation}\label{eq:exterior-fourier-bound}
|\widehat x[n]|\lesssim\frac{k}{|n|-2F}\|x\|_{[-1,1]}.
\end{equation}

Parseval and \eqref{eq:exterior-fourier-bound} therefore imply
\begin{align*}
\|x-P_Nx\|_{[-1,1]}^2
=\sum_{n\notin \{-N,\ldots,N-1\}} |\widehat x[n]|^2
\lesssim k^2\|x\|_{[-1,1]}^2 \sum_{m=N}^{\infty}\frac1{(m-2F)^2}
\lesssim\frac{k^2}{N-2F}\|x\|_{[-1,1]}^2.
\end{align*}
\end{proof}

Next we strengthen the $L_2$ approximation to an $L_{\infty}$ approximation for Property~\ref{item:approx-weighted}. 
\begin{lemma} \label{lem:ell-infty-approximation}
For any $k$-Fourier-sparse signal $x$ and any $N = \Omega(F + k^2)$, every $t\in[-1,1]$ satisfies
\begin{equation}\label{eq:corrected-pointwise-error}
 |x(t)-P_Nx(t)-c_x\iota^\bot(t)|
 \lesssim\frac{k(F+k^2)}N\|x\|_{[-1,1]}
                  \min\left\{1,\frac1{N(1-|t|)}\right\},
\end{equation}
where $c_x=(x(1)-x(-1))/2$ and $\iota^\bot=\iota-P_N\iota$.
\end{lemma}

Recall that $D_k(t)=\frac1{(2+4\ln k)\max\{1-|t|,k^{-2}\}}$. For $N\geq k^2$ and $t\in[-1,1]$, 
\begin{equation}\label{eq:weighted-tail-factor}
\sqrt{\max\{1-|t|,k^{-2}\}} \cdot \min\{1,(N(1-|t|))^{-1}\} \le
\left\{
\begin{matrix}
    k^{-1}, & 1-|t|\leq k^{-2} \\
    N^{-1}(1-|t|)^{-1/2} \le k^{-1} & 1-|t|> k^{-2}
\end{matrix}
\right.
\end{equation}
Indeed, for $1-|t|\leq k^{-2}$ it is at most $1/k$. For $1-|t|\geq k^{-2}$ we have $N(1-|t|)\geq1$, and the factor equals $1/(N\sqrt{1-|t|})\leq k/N\leq1/k$.
Thus Lemma~\ref{lem:ell-infty-approximation} yields
\[
 \max_{t\in[-1,1]}\frac{|x(t)-g_N(t)|}{\sqrt{D_k(t)}}
 \lesssim\frac{(F+k^2)\sqrt{\log k}}N\|x\|_{[-1,1]}.
\]
The cutoff makes this bound at most $\delta\|x\|_{[-1,1]}$, and also ensures $N\geq k^2$ and the hypothesis of Lemma~\ref{lem:ell-infty-approximation}.

\begin{proof}[Proof of Lemma~\ref{lem:ell-infty-approximation}]
We use the operator $B$ on $\mathcal V_x$ from the proof of Lemma~\ref{lem:ell-2-approximation}.
By \eqref{eq:differentiation-operator-bound}, $\|B\|_{\op}\lesssim F+k^2$. Iterating integration by parts \eqref{eq:fourier-boundary-recurrence} shows, for $|n| > 2\|B\|_{\op}$,
\begin{equation}\label{eq:finite-boundary-expansion}
\widehat x[n]
=\frac{\bi(-1)^n}{\sqrt{2} \pi}\sum_{r=0}^{\infty} \frac{2^r}{n^{r+1}}\bigl((B^rx)(1)-(B^rx)(-1)\bigr)
\end{equation}
And direct integration implies
\begin{equation}\label{eq:iota-fourier-coefficients}
\widehat\iota[0]=0,\qquad
\widehat\iota[n]
= \left[\frac{te^{-\bi \pi nt}}{-\bi \sqrt{2} \pi n}\right]_{-1}^1
    + \frac{1}{\bi \sqrt{2} \pi n}\int_{-1}^1e^{-\bi\pi nt}\,dt
= \frac{\sqrt{2}\bi(-1)^n}{\pi n}\quad(n\ne0).
\end{equation}
Thus $y:=x-c_x\iota$ has matching endpoint values, and subtracting $c_x \widehat\iota[n]$ cancels the $r=0$ term in \eqref{eq:finite-boundary-expansion}. 
For $|n|>2\|B\|_{\op}$ we have $\widehat y[n]=(-1)^n a(n)$, where
\begin{equation}\label{eq:corrected-coefficient-function}
a(f):=\frac{\bi}{\sqrt{2} \pi}\sum_{r=1}^{\infty} \frac{2^r}{f^{r+1}}\bigl((B^rx)(1)-(B^rx)(-1)\bigr),
\qquad |f|>2\|B\|_{\op}.
\end{equation}
For $|f| \geq 4\|B\|_{\op}$, by Property~\ref{item:uniform_bound_on_interval} of Lemma~\ref{lemma:bounds_Fourier_sparse_signals} and $\|B\|_{\op}\lesssim F+k^2$,
\begin{align}
|a(f)|
&\lesssim\frac{k\|x\|_{[-1,1]}}{|f|}
             \sum_{r=1}^{\infty}\left(\frac{2\|B\|_{\op}}{|f|}\right)^r
\lesssim\frac{k(F+k^2)}{|f|^2}\|x\|_{[-1,1]},
\label{eq:corrected-coefficient-decay}\\
|a'(f)|
&\lesssim\frac{k\|x\|_{[-1,1]}}{|f|^2}
             \sum_{r=1}^{\infty}(r+1)\left(\frac{2\|B\|_{\op}}{|f|}\right)^r
\lesssim\frac{k(F+k^2)}{|f|^3}\|x\|_{[-1,1]}.
\label{eq:corrected-coefficient-variation}
\end{align}

Moreover,
\begin{equation}\label{eq:corrected-tail-identity}
x(t)-P_Nx(t)-c_x\iota^\bot(t)
= y(t)-P_Ny(t)
= \frac{1}{\sqrt{2}}\sum_{n\notin \{-N,\ldots,N-1\}}a(n)e^{\bi\pi n(t+1)}.
\end{equation}
Absolute summation of \eqref{eq:corrected-coefficient-decay} bounds this tail by $O(k(F+k^2)\|x\|_{[-1,1]}/N)$ for every $t$.

We next give the improvement bound in the interior. 
For $p\leq q$, summation rearrangement shows
\begin{equation}\label{eq:tail-summation-by-parts}
 \sum_{n=p}^qa(n)e^{\bi n \pi(t+1)}
 =a(q)\sum_{m=p}^qe^{\bi m \pi(t+1)}
  +\sum_{n=p}^{q-1}(a(n)-a(n+1))\sum_{m=p}^ne^{\bi m \pi(t+1)}.
\end{equation}
And
\begin{equation}\label{eq:geometric-partial-sum}
 \left|\sum_{n=p}^qe^{\bi n \pi(t+1)}\right|
 \leq\frac2{|1-e^{\bi \pi(t+1)}|}
 =\frac1{\sin(\pi(1-|t|)/2)}\leq\frac1{1-|t|},
\end{equation}
where the last inequality follows from concavity of sine on $[0,\pi/2]$. 
For $p=N$, the first term in \eqref{eq:tail-summation-by-parts} tends to zero as $q\to\infty$. And \eqref{eq:corrected-coefficient-variation} implies
\[
\sum_{n=p}^{\infty}|a(n+1)-a(n)|
\leq\int_p^\infty|a'(f)|\,df
\lesssim\frac{k(F+k^2)}{N^2}\|x\|_{[-1,1]}.
\]
For the negative tail holds for the same bound.
Equation~\eqref{eq:tail-summation-by-parts} therefore bounds the tail by $O(k(F+k^2)\|x\|_{[-1,1]}/(N^2(1-|t|)))$.

Combining the two tail bounds proves \eqref{eq:corrected-pointwise-error}.
\end{proof}

\subsubsection{Proof of Property~\ref{item:approx-conditioning} in Theorem~\ref{thm:approximation}}

We prove Property~\ref{item:approx-conditioning} from two pointwise bounds. 

\begin{lemma}\label{lem:haar-pointwise-envelope}
For $H\in\mathcal H_{2N}$ and every $t\in[-1,1]$,
\begin{equation}\label{eq:haar-pointwise-envelope}
 |\phi_H(t)|\lesssim\min\left\{\sqrt{|H|},\frac1{\sqrt{|H|}(1-|t|)}\right\},
\end{equation}
where $|H|$ denotes the size of the support of $H$.
\end{lemma}

\begin{lemma}\label{lem:iota-pointwise-envelope}
For every integer $N\geq1$ and $t\in[-1,1]$,
\begin{equation}\label{eq:iota-pointwise-envelope}
 |\iota^\bot(t)|\lesssim\min\{1,(N(1-|t|))^{-1}\}.
\end{equation}
\end{lemma}

We have $|H|\leq K<2k$. If $1-|t|\geq k^{-2}$, then \eqref{eq:haar-pointwise-envelope} implies
\[
\frac{|\phi_H(t)|^2}{D_k(t)}
\lesssim(1+2\ln k)\min\{|H|(1-|t|),(|H|(1-|t|))^{-1}\}
\lesssim \log k.
\]

For $\iota^\bot$, the choice of the cutoff ensures $N\geq k^2$. Equation~\eqref{eq:iota-pointwise-envelope} shows
\[
 \frac{|\iota^\bot(t)|}{\sqrt{D_k(t)}}
 \lesssim\sqrt{1+2\ln k}\,
       \sqrt{\max\{1-|t|,k^{-2}\}}\min\{1,(N(1-|t|))^{-1}\}
 \leq\frac{\sqrt{1+2\ln k}}k
 \lesssim\frac{\sqrt{\log k}}k.
\]
This completes the proof of Theorem~\ref{thm:approximation}.

\begin{proof}[Proof of Lemma~\ref{lem:haar-pointwise-envelope}]
Recall that $\phi_H(t) = \frac{1}{\sqrt{2}}\sum_{n\in Q}H(n)(-1)^n e^{\bi\pi nt}$, where $H_{s+[\ell]}^{wav}=\frac{1}{\sqrt{\ell}} (\mathbf{1}_{s+[\ell/2]} - \mathbf{1}_{s+\ell/2+[\ell/2]})$ and $H^{sc}_{[K]}=\frac{\mathbf{1}_{[K]}}{\sqrt{K}}$.
With $|H(n)| \le |H|^{-1/2}$ and $|e^{\bi\pi nt}| \le 1$, the triangle inequality implies $|\phi_H(t)| \le \sqrt{|H|/2}$. 

Moreover, \eqref{eq:haar-geometric-formulas} shows
\begin{align*}
\phi_{H^{sc}_{b+[K]}}(t) = \frac{1-e^{\bi\pi K(t+1)}}{\sqrt{2K}(1-e^{\bi\pi(t+1)})} e^{\bi\pi b(t+1)}, \quad
\phi_{H^{wav}_{s+[\ell]}}(t) = \frac{(1-e^{\bi\pi\ell(t+1)/2})^2}{\sqrt{2\ell}(1-e^{\bi\pi(t+1)})} e^{\bi\pi s(t+1)}.
\end{align*}
The numerators have absolute value at most $4$. For denominators,
\[
|1-e^{\bi\pi(t+1)}| = 2\sin\Big(\frac{\pi(1 - |t|)}{2}\Big) \ge 2(1-|t|).
\]
Substitution proves the second bound in \eqref{eq:haar-pointwise-envelope}.
\end{proof}

\begin{proof}[Proof of Lemma~\ref{lem:iota-pointwise-envelope}]
\eqref{eq:iota-fourier-coefficients} implies
\begin{align}
P_N\iota(t) & = \sum_{n = -N}^{-1} \frac{\bi}{\pi n} e^{\bi n\pi(t+1)} + \sum_{n = 1}^{N - 1} \frac{\bi}{\pi n} e^{\bi n\pi(t+1)} \notag\\
& = -\frac2\pi\sum_{n=1}^{N-1}\frac{\sin(n\pi(t+1))}n - \frac{\bi}{\pi N}e^{\bi N\pi(t+1)}. \label{eq:iota-partial-sum}
\end{align}
And $|\iota^{\bot}(t)| \le |\iota(t)| + |P_N \iota(t)|$. Since $|\iota(t)| \le 1$ and $|e^{-\bi N\pi(t+1)}| = 1$, it remains to bound $\sum_{n=1}^{N-1}\frac{\sin(n\pi(t+1))}n$.
Since $|\sin u|\le|u|$ and $|\sin(n\pi(t+1))|=|\sin(n\pi(1-|t|))|$, we have $|\sin(n\pi(t+1))|\le\pi n(1-|t|)$. Thus
\[
\sum_{n=1}^{\lfloor(1-|t|)^{-1}\rfloor} \frac{|\sin(n\pi(t+1))|}n\leq\pi.
\]
If $N - 1 \le \lfloor(1-|t|)^{-1}\rfloor$, the uniform bound for $\iota^{\bot}$ follows. Otherwise, applying \eqref{eq:tail-summation-by-parts} and \eqref{eq:geometric-partial-sum} shows that,
\begin{align*}
\left|\sum_{n=\lfloor(1-|t|)^{-1}\rfloor}^{N - 1}\frac{\sin(n\pi(t+1))}n\right|
& \leq\frac1{1-|t|} \cdot \left(\frac{1}{N - 1}+\sum_{n=\lfloor(1-|t|)^{-1}\rfloor+1}^{N - 2} \left(\frac1n-\frac1{n+1}\right)\right) \\
& = \frac1{1-|t|} \cdot \frac{1}{\lfloor(1-|t|)^{-1}\rfloor} \leq 2,
\end{align*}
which proves $\iota^{\bot}(t) = O(1)$.
Similarly, we have
\[
\left| \sum_{n=N}^{\infty}\frac{e^{\bi n\pi(t+1)}}{n} \right| \le \frac{1}{N(1-|t|)}.
\]
With $\iota^\bot(t)=\frac{\bi}{\pi} \left(\sum_{n=N}^{\infty}\frac{e^{\bi n\pi(t+1)}}n -\sum_{n=N+1}^{\infty}\frac{e^{-\bi n\pi(t+1)}}n\right)$, it proves the second bound in \eqref{eq:iota-pointwise-envelope}.
\end{proof}

\section{Improper Learning in Time $\tilde{O}(k)$}\label{sec:improper_learning}
We show the algorithm of time $\tilde{O}(k)$ in  Theorem~\ref{inform:improper_learning}, which is an improper learning algorithm based on the Haar-Fourier basis constructed in Section~\ref{sec:characterizations}.
\begin{theorem}\label{thm:improper_learning}
    Given any bandlimit $F$, sparsity $k$, and error $\epsilon$, given any noisy observation $y(t):=x(t)+\eta(t)$ of a $k$-Fourier-sparse signal $x(t)$ and noise $\eta$, in $\tilde{O}(k/\epsilon^2)$ samples and $\tilde{O}(k/\epsilon^2)$ time, Algorithm~\ref{alg:improper_learning} reconstructs $\tilde{x}:=\sum_{H} \tilde{a}_H \phi_H$ in the Haar-Fourier basis $\mathcal{B}$ of sparsity $\tilde{O}(k/\epsilon)$ such that with probability $0.9$, $\|\tilde{x}-x\|_{[-1,1]}^2 = O(\epsilon) \cdot \|x\|_{[-1,1]}^2 + O(1) \cdot \|\eta\|_{[-1,1]}^2$. 
    \end{theorem}
 In the rest of this section, we present details of Algorithm~\ref{alg:improper_learning} and prove Theorem~\ref{thm:improper_learning}. For convenience, let $P_N x := \sum_{H \in \mathcal{H}_{2N}} \alpha_H \cdot \phi_H$ for $\alpha_H:=\langle x,\phi_H\rangle_{[-1,1]}$ be the orthogonal projection of $x$ onto $\operatorname{span}(\mathcal B)$. Because Theorem~\ref{thm:approximation} implies $\|P_N x - x\|_{[-1,1]}^2 \le \epsilon \cdot \|x\|_{[-1,1]}^2$, it is enough to learn $\sum_H \alpha_H \cdot \phi_H$. Moreover, the coefficient vector $(\alpha_H)_{H}$ is $O(\sqrt{k} \cdot \log Nk)$ analytic-sparse such that it is enough to learn the largest $\tilde{O}(k)$ coefficients (in absolute value). 
\begin{claim}\label{clm:analytic_sparse_approx}
    For any $k$-Fourier-sparse $x$ with $P_N x := \sum_{H \in \mathcal{H}_{2N}} \alpha_H \cdot \phi_H$ and integer $s$, the top $s$ functions $\sum\limits_{H: \alpha_H \text{is the largest $s$ coefficients}} \alpha_H \phi_H$ is a good approximation of $P_N x$ with error $\|\bott_s(\alpha_H)\|_2^2 < O(\frac{k \log^2 Nk}{s}) \| P_N x\|^2_{[-1,1]}$.
\end{claim}
\begin{proof}
    Let $\theta:=\frac{\sum_H |\alpha_H|}{s+1}$ be the largest possible value of the largest $s+1$ coefficient in $\alpha_H$. Then 
    \[
    \|\bott_s(\alpha_H)\|_2^2 \le \theta \cdot (\sum_H |\alpha_H|) < (\sum_H |\alpha_H|)^2/s = O(\frac{k \log^2 Nk}{s}) \cdot \|\alpha_H\|_2^2.
    \]
\end{proof}

Here are some notations in Algorithm~\ref{alg:improper_learning} and its analysis: 
Because there are $|W|=\log K + 1$ types of vectors in the Haar basis $\mathcal{H}_{2N}$ --- a shift of either $H^{wav}_{[\ell]}=\frac{1}{\sqrt{\ell}} (\mathbf{1}_{[\ell/2]} - \mathbf{1}_{\ell/2 + [\ell/2]})$ for a 2-power $\ell \in [2,K]$ or $H^{sc}=\frac{1}{\sqrt{K}} \mathbf{1}_{[K]}$, we consider how to recover heavy coefficients of each type separately. 
We fix a type $w \in W$ of support size $\ell$ (either $H^{wav}_{[\ell]}$ or $H^{sc}$) such that there are $2N/\ell$ shifts of $w$ in $\mathcal{H}_{2N}$. 
For convenience, let $w_{b}$ for $b \in \mathbb{Z}$ denote the shifted vector of $w$ whose support is $b \ell + [\ell]$ instead of $[\ell]$. 
We still use $\phi_{w_b}:=\sum_{n} w_b[n] \cdot (-1)^n \cdot e^{\bi \pi n t} /\sqrt{2}$ to denote the corresponding vector in the Haar-Fourier basis $\mathcal{B}$ like Definition~\ref{def:basis_B}. 
By the definition of $w_b$, $\phi_{w_b}(t)=\phi_w(t) \cdot (-1)^{b \ell} \cdot e^{\bi \pi b \ell t}=\phi_w(t) \cdot e^{\bi \pi b \ell t}$. We will consider $w_b$ and $\phi_{w_b}$ for all $b \in \mathbb{Z}$, although $\mathcal{H}_{2N}$ only contains shifts with $b \in \{-N/\ell,-N/\ell+1,\ldots,N/\ell-1\}$.  
For convenience, let $G_w:=\{-N/\ell,-N/\ell+1,\ldots,N/\ell-1\}$ be the candidate set of $b$.

\begin{algorithm} [t!]
    \caption{Learning Fourier-sparse signals under the Haar-Fourier basis \label{alg:improper_learning}} 
    \begin{algorithmic}[1]
            \Function{ApproxAccess}{$w,\ell,R,u$}
            \State $z \gets 0$.             \For{$r \in [R]$}
                \State sample $j \in [\log_2 \ell]$ uniformly at random
                \State interval $J \gets (1-2^{-j},1-2^{-j-1}]$ if $j < \log_2 \ell - 1$; otherwise $J \gets (1-2/\ell,1]$                 \State sample $u_r$ from the points in $\frac{u+2\mathbb{Z}}{\ell} \cap (J \cup -J)$ uniformly at random
                \State $z \gets z+ \cdot \frac{y(u_r)\cdot\overline{\phi_w(u_r)}}{\sqrt{2}  p_{\ell}(u_r)}$
            \EndFor
            \State \Return $\tilde{z}(u):=z/R$
            \EndFunction                      \Procedure{ImproperLearning}{$F,K,\epsilon$} 
                                \For{each $w \in W$ of the form $H^{wav}_{[\ell]}$ or $H^{sc}_{[K]}$}
                    \State $G_w \gets \{-N/\ell,\ldots,N/\ell-1\}$ as the potential set of frequencies
                    \State $\wh{z_w} \gets$ sparseFFT\big(noisy access to signal $\zeta_w(u)$: \textsc{ApproxAccess}($w,\ell,R:=O(\frac{\log^2 k}{\epsilon}),u$), frequency set $G_w$, sparsity $s:=C \frac{k \log^2 Nk}{\epsilon}$, error $\delta:=\frac{\epsilon}{\log k}$\big) \Comment{ $\wh{z_w} \in \mathbb{C}^{G_w}$ approximates $\wh{\zeta_w}$ } \label{line:sparseFFT}
                    \State For every $b \in G_w$ with $\wh{z_w}[b] \neq 0$, set $\beta_{H} \gets \wh{z_w}[b]$ for $H:=(w \text{ shifts to } b\ell +[\ell])$ in the Haar basis $\mathcal{H}_{2N}$
                \EndFor
            \State Return $\sum_{H \in \mathcal{H}_{2N}} \beta_H \cdot \phi_H$
            \EndProcedure
    \end{algorithmic}
\end{algorithm}

The key of Algorithm \ref{alg:improper_learning} is to apply classical sparse Fourier transforms to learn large coefficients $\langle y, \phi_{w_b} \rangle_{[-1,1]}$ among $b \in G_w$. To apply sparse FFT, we plan to define a periodic signal $\zeta_w(t)$ such that $\wh{\zeta_w}[b]= \langle y, \phi_{w_b} \rangle_{[-1,1]} = \int_{-1}^1 y(t) \cdot \overline{\phi_w(t)} \cdot e^{- \bi \pi b \cdot \ell t} \mathrm{d} t $. Since we consider $y$ and $\phi_w$ as periodic signals over $[-1,1]$, let $t \mod 2$ be the remainder in $[-1,1)$ in this section. Then 
\begin{equation}\label{eq:def_Zeta}
\zeta_w(u) 
= \frac{\sqrt{2}}{\ell} \sum_{t \in [-1,1): \ell t \equiv u\mod 2} y(t) \overline{\phi_w(t)} 
= \frac{\sqrt{2}}{\ell} \sum_{a \in [\ell]} y\left(\frac{u+2a}{\ell}\right) \overline{\phi_w\left(\frac{u+2a}{\ell}\right)}.    
\end{equation}

\begin{claim}\label{clm:def_zeta}
    $\zeta_w(u)$ defined in \eqref{eq:def_Zeta} satisfies $\wh{\zeta_w}[b]=\langle y, \phi_{w_b} \rangle_{[-1,1]}$, which provides a coefficient of $y$ in the Haar-Fourier basis $\mathcal{B}$. This implies $\sum_{w \in W} \|\zeta_w\|_{[-1,1]}^2  = \|y\|_{[-1,1]}^2$.
\end{claim}

However, \eqref{eq:def_Zeta} takes $\ell$ samples and $O(\ell)$ time to compute $\zeta_w(u)$. Since $\ell$ could be as large as $K = O(k)$, accessing $\zeta_w(u)$ like \eqref{eq:def_Zeta} leads to a learning algorithm of time $\tilde{O}(k^2)$ and $\tilde{O}(k^2)$ samples. Next we use Function \textsc{ApproxAccess} of $\zeta_w(u)$ to optimize the time complexity and sample complexity. 
We first optimize the sample complexity by a weighted distribution. Since $y$ would be affected by the noise, we choose a weighted sampling distribution for $\phi_{w_b}$ from Property 4 of Theorem~\ref{thm:approximation}:
\begin{equation}\label{eq:def_distribution_q_ell}
    q_{\ell}(t)=\frac{1}{2(1+\ln \ell) \cdot \max\{1/\ell,1-|t|\}}.
\end{equation}
The first idea is to replace $\zeta_w(u)$ defined in \eqref{eq:def_Zeta} by an expectation whose probability of $a$ is proportional to $q_{\ell}(\frac{u+2a}{\ell})$:
\begin{align*}
\zeta_w(u) & = \E_{a \propto q_{\ell}(\frac{u+2a}{\ell})}\left[\frac{\sqrt{2} \cdot y(\frac{u+2a}{\ell})\cdot \overline{\phi_w(\frac{u+2a}{\ell})}}{\ell \cdot q_{\ell}(\frac{u+2a}{\ell})/(\sum_{i \in [\ell]} q_{\ell}(\frac{u+2i}{\ell}))}
\right] \\
& = \E_{a \propto q_{\ell}(\frac{u+2a}{\ell})}\left[ \frac{\sqrt{2} \cdot y(\frac{u+2a}{\ell})\cdot \overline{\phi_w(\frac{u+2a}{\ell})}}{q_{\ell}(\frac{u+2a}{\ell})}
\right] \cdot \left (\sum_{i \in [\ell]} q_{\ell}(\frac{u+2i}{\ell}) \right)/\ell.
\end{align*}
One remark is that all $\frac{u+ 2a}{\ell}$ and $\frac{u+2i}{\ell}$ are modulo $2$ into $[-1,1)$ in the above expression. Since $\sum_{i \in [\ell]} q_{\ell}(\frac{u+2i}{\ell})/\ell = O(1)$ and Property 4 of Theorem~\ref{thm:approximation} bounds $\frac{|\phi_w(\frac{u+2a}{\ell})|^2}{q_{\ell}(\frac{u+2a}{\ell})}=O(\ln \ell)$, $R$ samples would guarantee the error over $u \in [-1,1]$ is $\int_{-1}^1 \E[|\tilde{z}(u)-\zeta_w(u)|^2] \mathrm{d} u = O(\frac{\ln \ell}{R}) \cdot \|y\|_{[-1,1]}^2$. Setting $R=(\ln k)^{O(1)}/\epsilon$ would improve the sample complexity by a factor of $\tilde{O}(k)$.

To improve the time complexity of sampling, we approximate $q_{\ell}(t)$ as follows
\begin{equation}\label{eq:def_distribution_p_ell}
    p_{\ell}(t)=\begin{cases}
        \frac{2^{j}}{\log_2 \ell } & \text{ if } |t| \in (1-2^{-j}, 1-2^{-j-1}] \text{ for } j \in [\log_2 \ell - 1] \\
        \frac{\ell/4}{\log_2 \ell } & \text{ for } |t| > 1-2/\ell
    \end{cases}
\end{equation}
By the definition, $p_{\ell}(t)$ is uniform in each interval of $(0,1/2],(1/2,3/4],\ldots,(1-2/\ell,1]$. Moreover, the $\ell$ points $\{\frac{u+2a}{\ell} \mod 2: a \in [\ell]\}$ of $u$ have $\ell/4$ points in $(0,1/2]$, $\ell/8$ points in $(1/2,3/4]$, $\ldots$, $1$ point in $(1-4/\ell,1-2/\ell]$, and $1$ point in $(1-2/\ell,1]$ such that the sum of $p_{\ell}(\frac{u+2a}{\ell})$ over points in each interval is the same. 
This brings the efficient sampling procedure in \textsc{ApproxAccess}.

\begin{claim}\label{clm:approx_zeta}
    Given $w$ of support size $\ell$ and any $u \in [-1,1) \setminus \frac{\mathbb{Z}}{2\ell}$, each $u_r$ in the for-loop of Function~\textsc{ApproxAccess} is a random sample from $p_{\ell}$.
    
    Moreover, the output $\tilde{z}(u)$ satisfies $\E[\tilde{z}(u)]=\zeta_w(u)$ and $\int_{-1}^1 \E[|\tilde{z}(u)-\zeta_w(u)|^2] du \le \frac{O(\log k)}{R} \cdot \|y\|_{[-1,1]}^2$. \end{claim}

We defer the proofs of Claim~\ref{clm:def_zeta} and Claim~\ref{clm:approx_zeta} to Section~\ref{sec:proofs_zeta}. The last piece in the proof of Theorem~\ref{thm:improper_learning} is the following guarantee of sparse FFT algorithms from Theorem 4.9 in \cite{HIKP}. While there are sparse FFT algorithms for periodic signals \cite{GMS,Iw13}, we consider the discrete one in \cite{HIKP} because it has an $\ell_2/\ell_2$ guarantee. One remark is that Lemma~\ref{lem:sparseFFT} reduces the failure probability to $\frac{0.01}{\log k}$ by losing a factor $(\log k)^{O(1)}$ on the time complexity. This is achieved by adjusting parameters in Theorem 4.9 in \cite{HIKP}.
\begin{lemma}\label{lem:sparseFFT}
    Given dimension $m$, sparsity $k$, and any error $\delta$, there are sparse-FFT algorithms in $k \log^{O(1)} \frac{m}{\delta}$ time such that for any input $v \in \mathbb{C}^m$, it outputs a sparse $\wh{v^o}$ s.t. with probability $1-\frac{0.01}{\log k}$, 
    \[    
    \|\wh{v}-\wh{v^o}\|_2 \le 2 \cdot \| \bott_k(\wh{v})\|_2 + \delta \cdot \|\wh{v}\|_2.        
    \]
\end{lemma}

\begin{proofof}{Theorem~\ref{thm:improper_learning}}
We specify the sparse FFT algorithm in Line \ref{line:sparseFFT}. 
    For a fixed $w$ of support size $\ell$, we set $M_w=|G_w|=2N/\ell$ and pick $\theta_w \sim [0,2/M_w)$. Let $v_w \in \mathbb{C}^{M_w}$ with $v_w[j]:=\tilde{z}(\theta_w + 2j/M_w-1)$ be the approximation from \textsc{ApproxAccess}($w,\ell,R,\theta_w + 2j/M_w-1$). Then we apply Lemma~\ref{lem:sparseFFT} to $v_w$ and obtain a sparse vector $\wh{v^{o}_w}$. For each $b \in G_w$, if $\wh{v^o_w}[f_b] \neq 0$ for $f_b := b \mod M_w$, we set $\wh{z_w}[b]:=\wh{v^o_w}[f_b] \cdot e^{-\bi \pi (\theta_w-1) b} \cdot \sqrt{2}/M$. We assume all $\wh{v^{o}_w}$ have the guarantee in Lemma~\ref{lem:sparseFFT} (with probability 0.99). Because $w$ would be fixed in most part of this proof, we omit notations like $v_w$ and $M_w$ as $v$ and $M$.

    \paragraph{Correctness.} Let $\alpha_H:=\langle x,\phi_H\rangle_{[-1,1]}$ and $\alpha_w \in \mathbb{C}^{G_w}$ be $\alpha_w[b]=\langle x, \phi_{w_b} \rangle_{[-1, 1]}$ for $H:=w_b$. As a sub-vector of $(\alpha_H)_{H \in \mathcal{H}_{2N}}$,      Claim~\ref{clm:analytic_sparse_approx} implies $\|\bott_s(\alpha_w)\|_2^2 \le O(\frac{k \log^2 Nk}{s}) \|(\alpha_H)_{H \in \mathcal{H}_{2N}}\|_2^2$. We bound the error between $\wh{z_w}[b]$ in Algorithm~\ref{alg:improper_learning} and the target $\alpha_w[b]$ to prove the correctness. 
    
    Since $\wh{z_w}[b]:=\wh{v^o_w}[f_b] \cdot e^{-\bi \pi (\theta_w-1) b}\sqrt{2}/M$ where $\wh{v^o_w}$ from Lemma~\ref{lem:sparseFFT} and its error relies on $\|\wh{v_w}\|_2$ and $\|\bott_s(\wh{v_w})\|_2$, we bound $\|\wh{v_w}\|_2$ and $\|\bott_s(\wh{v_w})\|_2$ at first. 
    
    In the analysis, let $v'_w \in \mathbb{C}^M$ with $v'_w[j]:=\zeta_w(\theta_w + 2j/M_w-1)$ for convenience. Claim~\ref{clm:def_zeta} implies
    \begin{equation}\label{eq:norm_v_W}
    \E[\|v'_w\|_2^2] = \frac{M}{2} \cdot \int_{-1}^1 |\zeta_w(u)|^2 \mathrm{d} u \le \frac{M}{2} \cdot \|y\|_{[-1,1]}^2.
    \end{equation}
        Also, Claim~\ref{clm:approx_zeta}
    bounds $\E_{\theta_w}[\|v_w-v_w'\|_2^2] = \E_{\theta_w}[\sum_j |v_w[j]-v_w'[j]|^2]$ as 
    \begin{align*}
        \E\left[\sum_{j \in [M]} |\tilde{z}(\theta_w + 2j/M-1)-\zeta_w(\theta_w + 2j/M-1)|^2\right] 
        & = \frac{M}{2} \sum_{j \in [M]} \int_{2j/M-1}^{2j/M-1+2/M} \E |\tilde{z}(u)-\zeta_w(u)|^2 \mathrm{d} u  \\
        & = \frac{M}{2} \cdot \frac{O(\log k)}{R} \cdot \|y\|_{[-1,1]}^2. 
    \end{align*}
    This implies \begin{equation}\label{eq:error_v_v'_freq}
        \E_{\theta_w} \left[ \sum_f \left| \wh{v_w}[f]/M-\wh{v'_w}[f]/M \right|^2 \right] = \frac{O(\log k)}{R} \cdot \|y\|_{[-1,1]}^2 \text{ and } \E[\|v_w\|_2^2] = O(M) \cdot \|y\|_{[-1,1]}^2.
    \end{equation}
    
    On the other hand, let  $\wh{v'_w}[f]:=\sum_j v'_w[j] \cdot e^{-2 \pi \bi fj/M}$ such that
    \begin{align*}
        \wh{v_w'}[f] &= \sum_j \zeta_w(\theta_w + 2j/M-1) e^{-2 \pi \bi fj/M} \\
        & = \sum_j (\sum_b \wh{\zeta_w}[b] e^{\bi \pi b \cdot (\theta_w + 2j/M-1)}/\sqrt{2}) \cdot e^{-2 \pi \bi fj/M}\\
        & = \sum_b \wh{\zeta_w}[b] e^{\bi \pi (\theta_w - 1) b}/\sqrt{2} \cdot (\sum_j  e^{2 \pi \bi (b/M - f/M) \cdot j}) \\
        & = \sum_{b \in f + M \cdot \mathbb{Z}} \wh{\zeta_w}[b] e^{\bi \pi (\theta_w - 1) b} \cdot M/\sqrt{2}.
    \end{align*}
    So each $\wh{\zeta_w}(b)$ for $b \in \{-N/\ell,\ldots,N/\ell-1\}$ maps to $f_b:=b \mod M$ in $\wh{v'_w}[f]$. Now we bound the error for every $b \in \{-N/\ell,\ldots,N/\ell-1\}$: 
    \begin{align}
        & \E_{\theta_w}\left[ \left| \wh{v'_w}[f_b] - \wh{\zeta_w}[b] e^{\bi \pi (\theta_w-1) b} \cdot \frac{M}{\sqrt{2}} \right|^2 \right] \notag \\
        = & \E_{\theta} \left[ \left| \sum_{a \in f_b + M \cdot \mathbb{Z}: a \neq b} \wh{\zeta_w}[a] e^{\bi \pi (\theta_w-1) a} \cdot \frac{M}{\sqrt{2}} \right|^2 \right] \notag \\
        = & \frac{M^2}{2}  \sum_{a \in f_b + M \cdot \mathbb{Z}: a \neq b} |  \wh{\zeta_w}[a]|^2 + \frac{M^2}{2}  \sum_{a,a'} \E[\wh{\zeta_w}(a) e^{\bi \pi (\theta_w-1) (a-a')} \overline{\wh{\zeta_w}[a']}] \notag \\
        = & \frac{M^2}{2}  \sum_{a \in f_b + M \cdot \mathbb{Z}: a \neq b} |  \wh{\zeta_w}[a]|^2. \label{eq:expected_error_v'}
    \end{align}
    The cross term has expectation 0 because $a-a' \in M \cdot \mathbb{Z} \setminus \{0\}$ and $\theta_w \sim [0,2/M)$. So
    \begin{equation}\label{eq:expected_error2_v'}
     \E_{\theta_w} \left[ \sum_{b \in G_w, f_b = b \mod M} \Big|\wh{v'_w}[f_b]- \wh{\zeta_w}(b) e^{\bi \pi (\theta_w-1) b} \cdot \tfrac{M}{\sqrt{2}} \Big|^2 \right] = \frac{M^2}{2} \cdot \sum_{a \notin \{-N/\ell,\ldots,N/\ell-1\}} |\wh{\zeta_w}(a)|^2.        
    \end{equation}

    Now we compare $\bott_s\big( \wh{v'_w} \big)$ and $\bott_s\big( \alpha_w \big)$. There are two errors: $\sum_{b \in G_w} |\wh{\zeta_w}[b] - \alpha_w(b)|^2 = \sum_b |\langle y-x, \phi_{w_b} \rangle|^2$ and $\E_{\theta_w}\big[ \sum_b \big| \wh{v'_w}[f_b] \cdot e^{-\bi \pi (\theta_w-1) b} \cdot \frac{\sqrt{2}}{M} - \wh{\zeta_w}[b] \big|^2 \big] = \sum_{a \notin \{-N/\ell,\ldots,N/\ell-1\}} |\wh{\zeta_w}(a)|^2$ from \eqref{eq:expected_error2_v'}. Let 
    \begin{align}
        \Err_w: &=\underset{a \notin \{-N/\ell,\ldots,N/\ell-1\}}{\sum} |\wh{\zeta_w}(a)|^2 + \underset{b \in \{-N/\ell,\ldots,N/\ell-1\}}{\sum} |\wh{\zeta_w}(b) - \langle x,\phi_{w_b}\rangle|^2 \notag \\
        & =  \underset{a \notin \{-N/\ell,\ldots,N/\ell-1\}}{\sum} |\langle y, \phi_{w_a} \rangle|^2 + \underset{b \in \{-N/\ell,\ldots,N/\ell-1\}}{\sum} |\langle y- x,\phi_{w_b}\rangle|^2   \label{def_Err_w}
    \end{align} capture these two errors among shifts of $w$. Furthermore, \eqref{eq:error_v_v'_freq} implies
    \begin{equation}\label{eq:error_vv'}
    \E_{\theta_w} \left[ \sum_{f} \Big|\tfrac{\sqrt{2}}{M} \cdot \wh{v_w}[f] - \tfrac{\sqrt{2}}{M}\cdot\wh{v'_w}[f] \Big|^2 \right] = \frac{O(\log k)}{R} \cdot \|y\|_{[-1,1]}^2.        
    \end{equation}
    
    Now we apply Fact~\ref{fact:bottom_vectors} and plug  $\Err_w:=\underset{a \notin \{-N/\ell,\ldots,N/\ell-1\}}{\sum} |\wh{\zeta_w}(a)|^2 + \underset{b \in \{-N/\ell,\ldots,N/\ell-1\}}{\sum} |\wh{\zeta_w}(b) - \langle x,\phi_{w_b}\rangle|^2$, \eqref{eq:expected_error2_v'}, and \eqref{eq:error_vv'}:
    \begin{align}
    \E_{\theta_w} \left[ \|\bott_s(\wh{v_w}[f] \cdot \tfrac{\sqrt{2}}{M})\|_2^2 \right] & = O(1) \cdot \bigg( \|\bott_s(\alpha_w)\|_2^2 + \|\alpha_w - \wh{\zeta_w}\|_2^2 \notag \\
    & \qquad + \E \|\wh{\zeta_w}-(\wh{v'_w}[f_b] e^{-\bi \pi (\theta_w-1) b} \cdot \tfrac{\sqrt{2}}{M})\|_2^2 + \E \|\wh{v'_w} \cdot \tfrac{\sqrt{2}}{M}-\wh{v_w} \cdot \tfrac{\sqrt{2}}{M}\|_2^2 \bigg) \notag \\
    & = O(1) \cdot \bigg( \|\bott_s(\alpha_w)\|_2^2 + \sum_b |\langle y-x,\phi_{w_b} \rangle|^2 \notag\\
    & \qquad + \sum_{a \notin \{-N/\ell,\ldots,N/\ell-1\}} |\wh{\zeta_w}(a)|^2 + \frac{O(\log k)}{R} \cdot \|y\|_{[-1,1]}^2\bigg) \notag\\
    & = O(1) \cdot \left( \|\bott_s(\alpha_w)\|_2^2 + \Err_w + \frac{O(\log k)}{R} \cdot \|y\|_{[-1,1]}^2 \right).   \label{eq:bot_hat_v_W}     
    \end{align}

    Since Lemma~\ref{lem:sparseFFT} implies 
    \[
        \|\wh{v^o_w} - \wh{v_w} \|_2^2 = O(1) \cdot \|\bott_s(\wh{v_w})\|_2^2 + 2 \delta^2 \cdot \|\wh{v_w}\|_2^2 = O(1) \cdot \|\bott_s(\wh{v_w})\|_2^2 + 2 \delta^2 M \cdot \|v_w\|_2^2,
    \]
    plugging $\wh{z_w}[b]:=\wh{v_w^o}[f_b] \cdot e^{-\bi \pi (\theta_w-1) b} \cdot \frac{\sqrt{2}}{M}$ for $f_b = b \mod M$ into the L.H.S. shows
    \begin{align*}
        & \E_{\theta_w} \left[ \left\| (\wh{z_w})_{b} - (\frac{\wh{v_w}[f_b]\cdot e^{-\bi \pi (\theta_w-1) b}}{M/\sqrt{2}})_b \right\|_2^2 \right] \\
        = & O(1) \cdot \E \|\bott_s(\wh{v_w}) \cdot \tfrac{\sqrt{2}}{M}\|_2^2 + 4\delta^2/M \cdot \E \|v_w\|_2^2 \\
        = & O(1) \cdot \left( \|\bott_s(\alpha_w)\|_2^2 + \Err_w + \frac{\log k}{R} \cdot \|y\|_{[-1,1]}^2 \right) + 4\delta^2/M \cdot \E \|v_w\|_2^2 \tag{by \eqref{eq:bot_hat_v_W}} \\
        = & O(1) \cdot \left( \|\bott_s(\alpha_w)\|_2^2 + \Err_w + \frac{\log k}{R} \cdot \|y\|_{[-1,1]}^2  + \delta^2 \cdot \|y\|_{[-1,1]}^2 \right). \tag{by \eqref{eq:error_v_v'_freq}}
    \end{align*}
    Moreover, \eqref{eq:expected_error2_v'} and \eqref{eq:error_vv'} imply
    \begin{align}
        \E_{\theta_w} \left[ \Big\| (\frac{\wh{v_w}[f_b]\cdot e^{-\bi \pi (\theta_w-1) b}}{M/\sqrt{2}})_{b \in G} - (\wh{\zeta_w}[b])_{b \in G} \Big\|_2^2 \right] & = O(1) \cdot \Big(\sum_{a \notin \{-N/\ell,\ldots,N/\ell-1\}} |\wh{\zeta_w}(a)|^2 + \frac{\log k}{R} \cdot \|y\|^2_{[-1,1]} \Big) \notag \\
        & = O(1) \cdot \Big(\Err_w + \frac{\log k}{R} \cdot \|y\|^2_{[-1,1]}\Big). \label{eq:error_}
    \end{align}

    So we have
    \begin{equation}\label{eq:error_hat_z_zeta}
        \E_{\theta_w} \left[ \left \|(\wh{z_w})_{b} - (\wh{\zeta_w}[b])_{b} \right\|_2^2 \right] = O(1) \cdot \left( \|\bott_s(\alpha_w)\|_2^2 + \Err_w + \frac{O(\log k)}{R} \cdot \|y\|_{[-1,1]}^2 + \delta^2 \cdot \|y\|_{[-1,1]}^2 \right) .
    \end{equation}

    The final error is the summation of \eqref{eq:error_hat_z_zeta} over all $w$:
    \begin{align}
        & \E_{\theta_w} \left[ \sum_w \sum_b  |\wh{z_w}[b] - \langle y, \phi_{w_b} \rangle |^2 \right] \notag \\
        = & \sum_w \E_{\theta_w} \left[ \left \|(\wh{z_w})_{b \in G} - (\wh{\zeta_w})_{b \in G} \right\|_2^2 \right] \notag \\
        = & O(1) \sum_w \Big(\|\bott_s(\alpha_w)\|_2^2 + \Err_w + \frac{O(\log k)}{R} \cdot \|y\|_{[-1,1]}^2 + \delta^2 \cdot \|y\|_{[-1,1]}^2\Big) \notag \\
        \le & O(1) \cdot \left( \|\bott_s(\alpha_H)\|_{[-1, 1]}^2 + \sum_w \Err_w + \frac{\log^2 k}{R} \cdot \|y\|_{[-1,1]}^2 + \delta^2 \log k \cdot \|y\|_{[-1,1]}^2 \right) \notag \\
        = & O(\epsilon) \cdot \|P_N x\|_{[-1, 1]}^2 + O(1) \cdot \|y-x\|_{[-1,1]}^2 + O(\epsilon) \cdot \|y\|_{[-1,1]}^2 \\
        = & O(\epsilon) \cdot \|x\|_{[-1,1]}^2 + O(1) \cdot \|\eta\|_{[-1,1]}^2.
    \end{align}
    A Markov's inequality shows with probability $0.95$,
    \[
        \sum_H |\beta_H - \langle y, \phi_{H} \rangle|^2 = O(\epsilon) \cdot \|x\|_{[-1,1]}^2 + O(1) \cdot \|\eta\|_{[-1,1]}^2.
    \]
    This bound $\|\sum_H \beta_H \phi_H - x\|_{[-1,1]}^2$  because $\sum_{H} |\alpha_H - \langle y, \phi_H \rangle|^2 = O(\epsilon) \cdot \|x\|_{[-1,1]}^2 + O(1) \cdot \|\eta\|_{[-1,1]}^2$.
    
    \paragraph{Time complexity and Sample complexity.} The time complexity is $\tilde{O}(k)$. $|W|=\log K + 1$ and the sparseFFT algorithm in line~\ref{line:sparseFFT} takes  $\tilde{O}(k/\epsilon)$ queries of $\tilde{z}$ and runs in time $\tilde{O}(k/\epsilon)$. Each query of $\tilde{z}(u)$ takes $R=\widetilde O(1/\epsilon)$ samples and time. The sample complexity is $\tilde{O}(k/\epsilon^2)$ also.
\end{proofof}    
    
\subsection{Proofs of Claim~\ref{clm:def_zeta} and Claim~\ref{clm:approx_zeta}}\label{sec:proofs_zeta}
We verify Claim~\ref{clm:def_zeta} at first.

\begin{proofof}{Claim~\ref{clm:def_zeta}}
By the definition of the Fourier transform of periodic signals:
\begin{align*}
    \wh{\zeta_w}[b]: & = \frac{1}{\sqrt{2}} \cdot \int_{-1}^1 \zeta_w(u) e^{- \bi \pi b u} \mathrm{d} u \\
    & = \int_{-1}^1 \frac{1}{\ell} \sum_{a \in [\ell]} y(\frac{u+2a}{\ell}) \overline{\phi_w(\frac{u+2a}{\ell})} \cdot e^{- \bi \pi b u} \mathrm{d} u \\
    & = \sum_{a \in [\ell]} \int_{-1}^1 y(\frac{u+2a}{\ell}) \overline{\phi_w(\frac{u+2a}{\ell})} \cdot e^{- \bi \pi b \ell \frac{u+2a}{\ell}} e^{\bi \pi 2ab} \mathrm{d} \frac{u+2a}{\ell} \\
    & = \int_{-1}^1 y(t) \overline{\phi_w(t)} e^{- \bi \pi b\ell t} \mathrm{d} t = \langle y, \phi_{w_b} \rangle.
\end{align*}
The second part follows from $\|\zeta_w\|_{[-1,1]}^2=\sum_n |\wh{\zeta_w}[n]|^2$ and $\sum_w \sum_{b \in \mathbb{Z}} |\langle y, \phi_{w_b} \rangle_{[-1, 1]}|^2 = \norm{y}_{[-1,1]}^2$.
\end{proofof}

Then we prove Claim~\ref{clm:approx_zeta}.

\begin{proofof}{Claim~\ref{clm:approx_zeta}}
    For convenience, let $t_a=(\frac{u+2a}{\ell} \mod 2)$ for $a \in [\ell]$ in this proof. We neglect the case $u \in \frac{\mathbb{Z}}{\ell}$ because it does not affect the integration and expectation.

    Then $t_0,\ldots,t_{\ell-1}$ have exactly $\ell/4$ points in $(0,1/2)$, $\ell/8$ points in $(1/2,3/4)$, $\ldots$, $1$ point in $(1-4/\ell,1-2/\ell)$, and $1$ point in $(1-2/\ell,1)$; and vice versa in $(-1,0)$. Moreover, the summation of $p_{\ell}(t_i)$ over $t_i$ in each interval is the same: $\frac{\ell/4}{\log_2 \ell}$. This also implies $\sum_{a \in [\ell]} p_{\ell}(t_i)=\ell/2$. So $u_r \sim \{t_0,\ldots,t_{\ell-1}\}$ is a random sampling of probability $p_{\ell}(t_i)$. 
    Furthermore, we have
    \[
    \E_{u_r} \left[\frac{y(u_{r})\cdot\overline{\phi_w(u_{r})}}{2 \cdot p_{\ell}(u_{r})}\right]=\sum_{i \in [\ell]} \frac{p_{\ell}(t_i)}{\ell/2} \cdot \frac{y(t_{i})\cdot\overline{\phi_w(t_{i})}}{\ell \cdot p_{\ell}(t_{i})/(\ell/2)} = \sum_i \frac{y(t_i) \overline{\phi_w(t_i)}}{\ell}.
    \]
    By definition, this shows $\E[\tilde{z}(u_r)]=\zeta_w(u)$. Then we analyze the variance.
    \begin{align*}
        \E[|\tilde{z}(u)-\zeta_w(u)|^2] & = \left( \E_{u_r}\left[ \left| \frac{y(u_{r})\cdot\overline{\phi_w(u_{r})}}{ \sqrt{2} \cdot p_{\ell}(u_{r})} \right|^2 \right] - |\zeta_w(u)|^2 \right)/R\\
        & \le \frac{2}{R} \sum_i \frac{p_\ell(t_i)}{\ell/2} \cdot \left| \frac{y(t_{i})\cdot\overline{\phi_w(t_{i})}}{\ell \cdot p_\ell(t_i)/(\ell/2)} \right|^2  \\
        & \le \frac{2}{R} \sum_i \frac{\left| y(t_{i})\cdot\overline{\phi_w(t_{i})}\right|^2}{\ell^2 \cdot p_\ell(t_i)/(\ell/2)} \\
        & \le \frac{2}{2R} \sum_i \frac{|y(t_{i})|^2}{\ell} \cdot \frac{|\overline{\phi_w(t_{i})}|^2}{p_\ell(t_i)}.
    \end{align*}
    We bound $\frac{|\overline{\phi_w(t_{i})}|^2}{p_\ell(t_i)}$ as follows: Property 4 of Theorem~\ref{thm:approximation} bounds $|\phi_w(t_i)|^2=O(\min\{\ell,\frac{1}{\ell(1-|t_i|)^2}\})$. When $|t_i| \le 1-1/\ell$, $\frac{|\phi_w(t_i)|^2}{p_\ell(t_i)} = O(\frac{\log_2 \ell}{\ell (1-|t_i|)}) = O(\log_2 \ell)$. Otherwise, $\frac{|\phi_w(t_i)|^2}{p_\ell(t_i)} = O(\frac{\ell}{\ell/(4 \log_2 \ell)})=O(\log_2 \ell)$ for $|t_i| > 1- 1/\ell$.

    Moreover, $\int_u \sum_i \frac{|y(t_{i})|^2}{\ell} \mathrm{d} u =\int_t |y(t)|^2 \mathrm{d} t$. Hence we have $\int_{-1}^1 \E[|\tilde{z}(u)-\zeta_w(u)|^2] du \le \frac{O(\log k)}{R} \cdot \|y\|_{[-1,1]}^2$.
    \end{proofof}

\section{Orthogonality between Fourier-sparse Signals}\label{sec:orthogonal}
In this section, we discuss Lemma~\ref{inform:almost_orthogonal_clusters}, which states that two Fourier-sparse signals are almost orthogonal on $[-1,1]$ if their frequencies are well-separated.
Moreover, we present several useful tools for proving the multiband approximation in Theorem~\ref{thm:learn_Fourier_sparse}. Here is a restatement of Lemma~\ref{inform:almost_orthogonal_clusters}.

\begin{lemma}
\label{lem:almost_orthogonal_clusters}
Let $0 < \delta \le 1/2$ and $C_M$ be a universal constant.
For two Fourier-sparse signals of sparsities $\ell$ and $r$, respectively, 
\[
w(t):=\sum_{j=1}^{\ell}\alpha_j e^{2\pi\bi f'_jt} \text{ and } z(t):=\sum_{j=1}^{r}\beta_j e^{2\pi\bi f_jt},
\]
if the separation between their frequencies $\min_{i,j} |f_i-f'_{j}| > C_M \cdot \big( \frac{\ell}{\delta} \cdot \min\{r,\frac{\ell}{\delta}\} \cdot \log \frac{\ell r}{\delta} + \frac{\ell r}{\delta} \big)$ 
then
\begin{equation}\label{eq:orthogonal_windows}
    |\langle w,z\rangle_{[-1,1]}|
    \le \delta \cdot \|w\|_{[-1,1]} \cdot \|z\|_{[-1,1]} .
\end{equation}
\end{lemma}
\begin{remark}
    The main term in the separation of Lemma~\ref{lem:almost_orthogonal_clusters} $C_M \cdot \big( \frac{\ell}{\delta} \cdot \min\{r,\frac{\ell}{\delta}\} \cdot \log \frac{\ell r}{\delta} + \frac{\ell r}{\delta} \big)$ is $O\left(\frac{\ell r}{\delta} \cdot \log \frac{\ell r}{\delta}\right)$. Lemma~\ref{inform:almost_orthogonal_clusters} uses $\Omega\left(\frac{\ell r}{\delta} \cdot \log \frac{\ell r}{\delta}\right)$ for ease of exposition.

    As mentioned earlier, a separation $\Omega({\ell r}/{\delta})$ is necessary. Because $k$-Fourier-sparse signals can approximate degree-$(k-1)$ polynomials arbitrarily close, extremal polynomials like the Chebyshev polynomials and the Legendre polynomials provide an example\footnote{One example  identified by GPT 5.6 sol was available at \url{https://vespri3.github.io/A_Sharp_Cross_Correlation_Lower_Bound_for_Separated_Fourier_Sparse_Signals.pdf}.} with two polynomials $p(t)$ of degree $\ell$ and $q(t)$ of degree $r$ such that  $\int_{[-1,1]} p(t) q(t) e^{-2 \pi \bi D t} \mathrm{d} t= \Omega(\frac{\ell r}{D}) \cdot \|p\|_{[-1,1]} \cdot \|q\|_{[-1,1]}$.
This demonstrates the frequency separation has to be $\Omega(\frac{\ell r}{\delta})$.
\end{remark}

To prove Lemma~\ref{lem:almost_orthogonal_clusters}, we show the main tool and several useful corollaries. This proof relies on a pair of filter functions $(M_{\ell,r,\delta},\wh{M_{\ell,r,\delta}})$. Its basic property is that $\wh{M_{\ell,r,\delta}}$ is compactly supported with 
\begin{equation}\label{eq:def_C_H}
    \supp(\wh{M_{\ell,r,\delta}}) \subset [- \Delta_M, \Delta_M] \text{ for } \Delta_M := \frac{C_M}{2} \cdot \big( \frac{\ell}{\delta} \cdot \min\{r,\frac{\ell}{\delta}\} \cdot \log \frac{\ell r}{\delta} + \frac{\ell r}{\delta} \big),
\end{equation} 
 and $M_{\ell,r,\delta}$ acts like a box function on the inner product $\langle M_{\ell,r,\delta} \cdot w, M_{\ell,r,\delta} \cdot z \rangle_2 \approx \langle w, z \rangle_{[-1,1]}$. Because Plancherel's identity implies $\langle M_{\ell,r,\delta} \cdot w, M_{\ell,r,\delta} \cdot z \rangle_2 = \langle \wh{M_{\ell,r,\delta} \cdot w}, \wh{M_{\ell,r,\delta} \cdot z} \rangle_2=0$, this implies $\langle w, z \rangle_{[-1,1]}$ is small.

We define this pair of filter functions formally in \eqref{eq:sqrt-localizer-def}. We first state its filtering properties for Fourier-sparse signals.
\begin{lemma}\label{lem:filter_sparse}
    Given any $\ell$ and $\delta$, $M_{\ell,\delta}:=M_{\ell,\ell,\delta}$ satisfies that for any $\ell$-Fourier-sparse signal $x$ in the time window $[-1,1]$, 
    \begin{enumerate}
        \item in the time domain, $\|M_{\ell,\delta} \cdot x - x \|_{[-1,1]}^2 \le \frac{\delta}{2} \cdot \|x\|_{[-1,1]}^2$ and $\|M_{\ell,\delta} \cdot x\|_{[-1,1]}^2 \ge (1-\delta) \|M_{\ell,\delta} \cdot x\|_2^2$; 
        \item in the frequency domain, $\supp(\wh{M_{\ell,\delta} \cdot x}) \subseteq \bigcup_{f_j \in \supp(\wh{x})} [f_j - \Delta_{\ell,\delta},f_j + \Delta_{\ell,\delta}]$ for $\Delta_{\ell,\delta} := C_M \cdot \frac{\ell^2}{\delta} \log \frac{\ell}{\delta}$.
    \end{enumerate}    
\end{lemma}
The proof of Lemma~\ref{lem:filter_sparse} is deferred to Section~\ref{sec:proof_filters}. We remark that $(M_{\ell,\delta},\wh{M_{\ell,\delta}})$ plays the same role as $(H,\wh{H})$ in previous works \cite{CKPS17, CP19_ICALP, SSWZ23}. For ease of exposition, we use $M$ instead of $H$ in this work. 

Lemma~\ref{lem:filter_sparse} is extremely useful in learning Fourier-sparse signals because it extends the observation of $x$ from the time window $[-1,1]$ to $\mathbb{R}$ such that the Fourier transform $\wh{M_{k,\delta} \cdot x}$ provides a method to estimate the frequencies in $x$. However, the error of the frequency recovery is $\tilde{\Omega}(k^2)$ because $|\supp(\wh{M_{k,\delta}})|=\tilde{\Theta}(k^2)$. This is the bottleneck on the sample complexity of previous learning algorithms. In Section~\ref{sec:multiband_approx}, we show a different way to recover $\wh{M_{k,\delta} \cdot x}$ via multiband approximations. The multiband approximation of a Fourier-sparse signal relies on the following corollary of Lemma~\ref{lem:almost_orthogonal_clusters} and Lemma~\ref{lem:filter_sparse}, whose proof is deferred to Section~\ref{sec:proofs_cor_filters}.

\begin{corollary}\label{cor:filter_functions}  
\begin{enumerate}
    \item  Let $\ell$, $k$, $\epsilon$, and $\delta$ be any parameters satisfying $\ell \le k$ and $\Delta_{\ell,\delta} \le \Delta_{k,\epsilon}$. For any $\ell$-Fourier-sparse signal $w$, $M_{\ell,\delta} \cdot w$ is a good approximation of $M_{k,\epsilon} \cdot w$:
    \begin{equation}\label{eq:approx_filters}
         \| M_{\ell,\delta} \cdot w - M_{k,\epsilon} \cdot w \|_2^2 \le O(\delta) \cdot \| M_{k,\epsilon} \cdot w\|_2^2.    
    \end{equation}
    
    \item Given any $k$, let $M_{k,\epsilon}$ denote the filter function constructed in Lemma~\ref{lem:filter_sparse} whose Fourier support is contained in $[-\Delta_{k,\epsilon},\Delta_{k,\epsilon}]$ and $C_H = 8 \cdot C_M$ for $C_M$ defined in  \eqref{eq:def_C_H}. For any $\ell$-Fourier-sparse signal $w$ and any $r$-Fourier-sparse signal $z$ with $\ell, r \le k$ whose frequencies are separated by $\min\Big\{C_H \cdot \frac{\ell r}{ \delta } \big(\log \frac{\ell r}{\delta} + \log \frac{k}{\epsilon} \big), 2\Delta_{k,\epsilon} \Big\}$, 
    \begin{equation}\label{eq:IP_filter_wz}
    \left| \langle M_{k,\epsilon} \cdot w, M_{k,\epsilon} \cdot z\rangle \right| \le \delta \cdot \|M_{k,\epsilon} \cdot w\|_2 \cdot \|M_{k,\epsilon} \cdot z\|_2. 
    \end{equation}
   
\end{enumerate}    
\end{corollary}

\paragraph{Constructions of $(M_{\ell,r,\delta},\wh{M_{\ell,r,\delta}})$.} We use $g\cdot h$ to denote the point-wise dot product $g(t) \cdot h(t)$ and $g^k$ to denote $g(t) \cdot g(t) \cdots g(t) \cdot g(t)$. Similarly, $g*h$ denotes the convolution of $g$ and $h$: $(g*h)(t) := \int g(x)h(t-x) \mathrm{d} x$ and $g^{*k}$ denotes the convolution $g*g*\cdots * g$.  
Let $C \ge 64$ be an even integer in this section and
\begin{align}
    h & :=\frac{\ell}{\delta} \cdot \min \left\{r,\frac{\ell}{\delta}\right\}, \label{eq:sqrt-localizer-scale}\\
    \alpha_{M} & :=1-\frac{1}{Ch},\nonumber\\
    g_{M}(t) & :=
    \sinc(C \cdot ht)^{C \lceil \log \frac{r}{\delta} \rceil}
    \prod_{i=0}^{\lceil\log r\rceil}
    \sinc \left(
       \frac{Ct}{h^{-1}+4^i/r^2}
    \right)^{C2^i}.
    \label{eq:sqrt-localizer-kernel}
\end{align}
Then we define
\begin{equation}
    M_{\ell,r,\delta}(t) :=s_{M} \cdot \bigl(\rect_{2\alpha_{M}} * g_{M}\bigr)(t), \label{eq:sqrt-localizer-def}
\end{equation}
where $s_{M}>0$ is chosen so that $M_{\ell,r,\delta}(0)=1$. Its Fourier transform 
\[
\wh{M_{\ell,r,\delta}}(f):=s_M \cdot \sinc( 2\alpha_M \cdot f) \cdot (\rect_{C h}(f)^{* C \lceil \log \frac{r}{\delta} \rceil}) * (\rect_{\frac{C}{h^{-1}+4^i/r^2}}(f)^{*C 2^i})_{i=0,\ldots,\lceil\log_2 r\rceil}.
\]

\paragraph{Analyses of $(M_{\ell,r,\delta},\wh{M_{\ell,r,\delta}})$.} We state the basic properties of  $M_{\ell,r,\delta}$ to prove Lemma~\ref{lem:almost_orthogonal_clusters}, Lemma~\ref{lem:filter_sparse}, and Corollary~\ref{cor:filter_functions}.

\begin{claim}
\label{claim:sqrt-localizer-bounds}
For $M = M_{\ell,r,\delta}$ with $1 \le \ell \le r$, $0 < \delta \le 1 / 2$, the following properties hold:
\begin{enumerate}
\item \label{item:sqrt-localizer-plateau}  $|1-M(t)| \le (\delta/r)^{C/2}$ when $|t|\le 1 - 2/(Ch)$.
\item \label{item:sqrt-localizer-global}
$0 \le M(t) \le 1 + (\delta/r)^{C/2}$ for all $t \in \R$. 
\item \label{item:sqrt-localizer-near-tail} 
For $1 \le |t| \le 1 + 1/C$, $M(t) \le (\delta/r)^{C/2} \exp \bigl(-\frac{C}{8}r\sqrt{|t|-1}\bigr)$.
\item \label{item:sqrt-localizer-far-tail}
For $|t| \ge 1+1/C$, $M(t) \le (\delta/r)^{C/2} |\pi t/2|^{-Cr/4}$.
\item \label{item:sqrt-localizer-support}
$\supp(\widehat{M})\subseteq[-\Delta_{M},\Delta_{M}]$ where
\begin{align*}
    \Delta_{M} := \frac{C^2}{2} \left( h \left\lceil\log \frac r\delta\right\rceil + \sum_{i=0}^{\lceil \log r\rceil} \frac{2^i}{h^{-1}+4^i/r^2} \right)
    \le 2 C^2 \left( \frac{\ell}{\delta} \cdot \min\{r,\frac{\ell}{\delta}\} \cdot \log \frac{\ell r}{\delta} + \frac{\ell r}{\delta} \right).
\end{align*}
\end{enumerate}
\end{claim}
One remark is that the construction of $M_{\ell,r,\delta}$ and Claim~\ref{claim:sqrt-localizer-bounds} are almost identical to the filter functions constructed in \cite{CP19_ICALP} and Claim 7.1 therein. The difference is that Property~\ref{item:square_root_bound_outside_interval} proved by Zhang \cite{zhang2026optimalextrapolationboundssparse} shows a smaller growth rate of Fourier-sparse signals outside $[-1,1]$. This small growth rate is the key to prove a nearly tight separation $\Omega\left(\frac{\ell r}{\delta} \cdot \log \frac{\ell r}{\delta}\right)$ where the rate in \cite{CP19_ICALP} requires $\Omega(\ell r^2/\delta)$. Because the proof of Claim~\ref{claim:sqrt-localizer-bounds} is also similar to the proof of Claim 7.1 in \cite{CP19_ICALP}, we defer its proof to Appendix~\ref{sec:properties_filters}.

Given Claim~\ref{claim:sqrt-localizer-bounds}, the idea of the proof of Lemma~\ref{lem:almost_orthogonal_clusters} is as follows. We consider $M_{\ell,r,\delta} \cdot w$  and $M_{\ell,r,\delta} \cdot z$ instead of $w$ and $z$. In the time domain, $\langle w,z\rangle_{[-1,1]} \approx \langle M_{\ell,r,\delta} \cdot w, M_{\ell,r,\delta} \cdot z \rangle$ over $\R$ by the properties of $M_{\ell,r,\delta}$. In the frequency domain, when the separation is large enough, $\wh{M_{\ell,r,\delta} \cdot w}$ and $\wh{M_{\ell,r,\delta} \cdot z}$ have disjoint supports such that their inner product is 0. We leave this calculation to Appendix~\ref{sec:almost_orthogonal_clusters}. 
We prove Lemma~\ref{lem:filter_sparse} and  Corollary~\ref{cor:filter_functions} in Appendix~\ref{sec:proof_filters} and Appendix~\ref{sec:proofs_cor_filters}, respectively.

\section{Multiband Approximations of Fourier-sparse Signals}\label{sec:multiband_approx}

\begin{theorem}\label{thm:multiband_approx}
    Given any $k$-Fourier-sparse signal $x(t)=\sum_{j=1}^k \alpha_j e^{2 \pi i f_j t}$ and $\epsilon$, there exists a multiband signal $x'$ such that
    \begin{enumerate}
        \item $\|x'\|_{[-1,1]}^2 \ge (1-\epsilon) \|x'\|_{2}^2$.
        \item $\|x-x'\|_{[-1,1]}^2 \le \epsilon \cdot \|x\|_{[-1,1]}^2$.
        \item $\supp(\wh{x'}) \subseteq I_1 \cup \cdots \cup I_n$ for $n \le k$ and $|I_1|+\cdots+|I_n|=\tilde{O}(k^2/\epsilon)$.
    \end{enumerate}
\end{theorem}

In the rest of this section, we finish the proof of Theorem~\ref{thm:multiband_approx}.  For convenience, given $k$ and $\epsilon$, let $M_{k,\epsilon}$ be the filter function constructed in Lemma~\ref{lem:filter_sparse} and 
$\Delta_{k,\epsilon} = O(\frac{k^2 \log (k/\epsilon)}{\epsilon})$ denote its support size such that $\supp(\wh{M_{k,\epsilon}}) \subseteq [-\Delta_{k,\epsilon},\Delta_{k,\epsilon}]$.
The starting point is that $M_{k,\epsilon} \cdot x$ is a multiband approximation of size $\tilde{O}(k^3/\epsilon)$. Lemma~\ref{lem:filter_sparse} implies this approximation has a small error.
At the same time, the Fourier transform of $\wh{M_{k,\epsilon} \cdot x}=\wh{M_{k,\epsilon}} * \wh{x}$ has at most $k$ intervals of length $\Delta_{k,\epsilon}=\tilde{O}(k^2/\epsilon)$ because $x$ is $k$-sparse. So the support size of $\wh{M_{k,\epsilon} \cdot x}$ is  $\tilde{O}(k^3/\epsilon)$. To reduce the size, we partition the $k$ frequencies in $\wh{x}$ into almost orthogonal clusters and apply Corollary~\ref{cor:filter_functions} to approximate each cluster. 

Here are a few notations about clusters in $\wh{x}$. For a fixed $k$-Fourier-sparse signal $x(t):=\sum_{j=1}^k \alpha_j e^{2 \pi \bi f_j t}$, a cluster of $\wh{x}$ (and $x$) is a subset of frequencies with coefficients $\cC \subset \{(f_1,\alpha_1),\ldots,(f_k,\alpha_k)\}$. We define its corresponding signal $x_{\cC}(t):=\sum_{(f,\alpha) \in \cC} \alpha \cdot e^{2 \pi \bi f t}$. We plan to partition the $k$ frequencies in $\wh{x}$ (with their coefficients) into as many clusters as possible while guaranteeing every pair of clusters is almost orthogonal: $\langle x_{\cC}, x_{\cC'} \rangle_{[-1,1]} \approx 0$.  For convenience, we define $\dist(\cC,\cC'):=\min_{(f,\alpha)\in\cC,\,(f',\alpha')\in\cC'} |f-f'|$ to be the frequency separation of these two clusters and $range(\cC)=[\min_{(f,\alpha)\in\cC} f, \max_{(f,\alpha)\in\cC} f]$ to be the range of  frequencies in a cluster $\cC$.

Given Lemma~\ref{lem:almost_orthogonal_clusters} and Corollary~\ref{cor:filter_functions} about orthogonality, a sufficient condition is that the frequency separation between every pair of clusters $\cC$ and $\cC'$ is at least $\tilde{\Omega}(\frac{|\cC| \cdot |\cC'|}{\epsilon})$. We start with $k$ clusters such that each contains a distinct frequency (with its coefficient) and keep merging them if any pair violates the separation condition. We strengthen the separation between $\cC$ and $\cC'$ to be $\tilde{\Omega}(k \cdot \min\{|\cC|,|\cC'|\} / \epsilon)$ in order to guarantee the main property (Lemma~\ref{lem:total_energy} below). We describe this partition procedure in Algorithm~\ref{alg:cluster}.

\begin{algorithm} 
    \caption{Partition Frequencies into Clusters \label{alg:cluster}}
    \begin{algorithmic}
        \Procedure{PartitionClusters}{frequencies $f_1,\ldots,f_k$ with amplitudes $\alpha_1,\ldots,\alpha_k$}      
        \State Define $k$ clusters $\cC_i:=\{(f_i,\alpha_i)\}$ and $d_{min}:=\frac{C_H k \log^{1.5} (k/\epsilon)}{\epsilon}$ for the constant $C_H=O(1)$ defined in Corollary~\ref{cor:filter_functions}
        \While{ $\exists~\cC_i$ and $\cC_j$ such that $       
        \operatorname{dist}(\cC_i,\cC_j) \le \min\bigg\{ d_{min} \cdot \min\{|\cC_i|,|\cC_j|\}, 2 \Delta_{k,\epsilon} \bigg\}
        $}
        \State  merge all clusters whose frequencies lie between $\cC_i$ and
    $\cC_j$ into one cluster
        \EndWhile
        \State Return all remaining clusters $\cC$
    \EndProcedure
    \end{algorithmic}
\end{algorithm}

We set the frequency separation in Algorithm~\ref{alg:cluster} in order to guarantee Lemma~\ref{lem:total_energy} stated below.
Basically, there are two terms in this separation: $d_{min} \cdot \min\{|\cC_i|,|\cC_j|\}$ and $2\Delta_{k,\epsilon}$. Since the algorithm will consider $M_{k,\epsilon} \cdot x$ whose Fourier transform is $\wh{M_{k,\epsilon}}*\wh{x}$ with $\supp(\wh{M_{k,\epsilon}}) \subseteq [-\Delta_{k,\epsilon},\Delta_{k,\epsilon}]$, $M_{k,\epsilon} \cdot x_{\cC}$ and $M_{k,\epsilon} \cdot x_{\cC'}$ are orthogonal (over $\mathbb{R}$) when their separation is $2 \Delta_{k,\epsilon}$. Otherwise, we call two clusters $\cC$ and $\cC'$ \emph{correlated} if their separation is less than $2 \Delta_{k,\epsilon}$. Another term $d_{min} \cdot \min\{|\cC_i|,|\cC_j|\}$ with $d_{min}:=\frac{C_H k \log^{1.5} (k/\epsilon)}{\epsilon}$ is chosen to bound the total correlation between a cluster $\cC$ and all the rest clusters.

\begin{lemma}
\label{lem:total_energy}
Let $d_{min}:=\frac{C_H k \log^{1.5} (k/\epsilon)}{\epsilon}$ and $\cC_1,\ldots,\cC_n$ be $n$ clusters with $\dist(\cC_i,\cC_j) \ge  \min\bigg\{ d_{min} \cdot \min\{|\cC_i|,|\cC_j|\}, 2\Delta_{k,\epsilon} \bigg\}$ for any two $\cC_i$ and $\cC_j$. For every $S\subseteq \{1, \ldots, n\}$,
    \[
    \left\|\sum_{i\in S}x_{\cC_i}\right\|_{[-1,1]}^2
    = \left(1\pm O(\epsilon)\right) \sum_{i\in S}\|x_{\cC_i}\|_{[-1,1]}^2 \text{ and }    \left\|M_{k,\epsilon} \cdot\sum_{i\in S}x_{\cC_i}\right\|_2^2
    = \left(1\pm O(\epsilon)\right) \sum_{i\in S}\|M_{k,\epsilon} \cdot x_{\cC_i}\|_2^2.
\]
In particular, $
    \|x\|_{[-1,1]}^2 
    = \left(1\pm O(\epsilon)\right) \sum_{i=1}^n\|x_{\cC_i}\|_{[-1,1]}^2 = \left(1\pm O(\epsilon)\right) \sum_{i=1}^n\|M_{k,\epsilon} \cdot x_{\cC_i}\|_2^2 $.
\end{lemma}

For completeness, we provide a proof of Lemma~\ref{lem:total_energy} in Section~\ref{sec:proof_total_energy}.  
Finally, we state an upper bound on the size of a cluster.
\begin{claim} \label{clm:bound_range_clusters}
From Algorithm~\ref{alg:cluster}, a cluster with $\ell$ frequencies has a range of length at most 
\[
r_\ell \le
\begin{cases}
d_{min}\cdot O(\ell \log \ell),   &\ell \le \frac{2\Delta_{k,\epsilon}}{d_{min}},\\
d_{min}\cdot O(\ell) \cdot \log \frac{2\Delta_{k,\epsilon}}{d_{min}},    & \ell >\frac{2\Delta_{k,\epsilon}}{d_{min}}.
\end{cases}
\]
Moreover, $\sum_{i \in \{1, \ldots, n\}} |range(\cC_i)| = \tilde{O}(k) \cdot d_{min} = \tilde{O}(k^2/\epsilon)$. 
\end{claim}

Now we are ready to finish the proof of Theorem~\ref{thm:multiband_approx}. The main idea is to apply Corollary~\ref{cor:filter_functions} to approximating each $M_{k,\epsilon} \cdot x_
{\cC_i}$ by $M_{|\cC_i|,\delta_i} \cdot x_
{\cC_i}$ for a proper $\delta_i$. On the one hand, because most clusters $\cC_i$ have $|\cC_i| = o(k)$, the Fourier support of $M_{|\cC_i|,\delta_i}$ is of size $\Delta_{|\cC_i|,\delta_i}=\tilde{O}(|\cC_i|^2/\delta_i)$. This could be much smaller than $\Delta_{k,\epsilon}$ if $\delta_i$ is not too small. On the other hand, Corollary~\ref{cor:filter_functions} bounds the relative error of replacing $M_{k,\epsilon} \cdot x_
{\cC_i}$ by $M_{|\cC_i|,\delta_i} \cdot x_
{\cC_i}$ as $O(\delta_i)$. The actual proof sets $\delta_i:=\frac{\epsilon \cdot |\cC_i|}{k}$ carefully to bound the total relative error over all clusters via a Cauchy--Schwarz inequality based on Lemma~\ref{lem:total_energy}.

\begin{proofof}{Theorem~\ref{thm:multiband_approx}}
    Let $\theta_s=\Theta(k)$ be the smallest integer such that $\Delta_{\theta_s,\frac{\epsilon \cdot \theta_s}{k}} \ge \Delta_{k,\epsilon}$. Because $\Delta_{\ell,\delta}=\Theta(\frac{\ell^2}{\delta} \log \frac{\ell}{\delta})$ from Lemma~\ref{lem:filter_sparse}, $\Delta_{k,\epsilon}=\Theta(\frac{k^2}{\epsilon} \log \frac{k}{\epsilon})$ and 
    \[
    \Delta_{\theta_s,\frac{\epsilon \cdot \theta_s}{k}}
    =\Theta\left(\frac{\theta_s^2}{\epsilon \cdot \theta_s / k} \log \frac{\theta_s}{\epsilon \cdot \theta_s / k}\right)
    =\Theta\left(\frac{k \theta_s}{\epsilon} \log \frac{k}{\epsilon}\right).
    \]
    This implies that $\theta_s=\Theta(k)$.

    Let $\cC_1,\ldots,\cC_n$ be the remaining clusters in Algorithm~\ref{alg:cluster}. Now we consider a multiband approximation of $\wh{M_{k,\epsilon} \cdot x}=\wh{M_{k,\epsilon}}*\wh{x}=\sum_{i=1}^n \wh{M_{k,\epsilon}}*\wh{x_{\cC_i}}$:
    \[
    \wh{z}:=\sum_{i:|\cC_i|<\theta_s} \wh{M_{|\cC_i|,\frac{\epsilon |\cC_i|}{k}}}*\wh{x_{\cC_i}} + \sum_{i:|\cC_i| \ge \theta_s} \wh{M_{k,\epsilon}}*\wh{x_{\cC_i}}.
    \]
    So the multiband is     \begin{equation}\label{eqdef:multiband_x}
        I_i:= \begin{cases} 
            \Big[\min_{(f, \alpha) \in \cC_i} f - \Delta_{|\cC_i|,\frac{\epsilon |\cC_i|}{k}}, \quad \max_{(f, \alpha) \in \cC_i} f + \Delta_{|\cC_i|, \frac{\epsilon |\cC_i|}{k}}\Big],             & |\cC_i|<\theta_s, \\
            \Big[\min_{(f, \alpha) \in \cC_i} f - \Delta_{k,\epsilon}, \quad \max_{(f, \alpha) \in \cC_i} f + \Delta_{k,\epsilon}\Big],             & |\cC_i| \ge \theta_s
        \end{cases}  \qquad \text{ for } i =1,\ldots,n.
    \end{equation}
    We bound $\|\wh{z} - \wh{M_{k,\epsilon} \cdot x}\|_2=\|z-M_{k,\epsilon} \cdot x\|_2$ as follows. By the definition,
    \[
    z:=\sum_{i:|\cC_i|<\theta_s} M_{|\cC_i|,\frac{\epsilon |\cC_i|}{k}} \cdot x_{\cC_i} + \sum_{i:|\cC_i| \ge \theta_s} M_{k,\epsilon} \cdot x_{\cC_i}
    \]
    such that
    \begin{equation}\label{eq:sum_diff_filters}
    \|z-M_{k,\epsilon} \cdot x\|_2 
    = \norm{\sum_{i:|\cC_i|<\theta_s} (M_{|\cC_i|,\frac{\epsilon |\cC_i|}{k}} - M_{k,\epsilon}) \cdot x_{\cC_i}}_2 
    \le \sum_{i:|\cC_i|<\theta_s} \norm{(M_{|\cC_i|,\frac{\epsilon |\cC_i|}{k}} - M_{k,\epsilon}) \cdot x_{\cC_i}}_2.        
    \end{equation}
    Since $\Delta_{|\cC_i|,\frac{\epsilon |\cC_i|}{k}} < \Delta_{k,\epsilon}$ by the definition of $\theta_s$, the first property of Corollary~\ref{cor:filter_functions} bounds each 
    \[
    \norm{\Big(M_{|\cC_i|,\frac{\epsilon |\cC_i|}{k}} - M_{k,\epsilon}\Big) \cdot x_{\cC_i} }_2
    = \sqrt{O(\epsilon |\cC_i|/k)} \cdot \|M_{k,\epsilon} \cdot x_{\cC_i}\|_2. \]
    Substituting this bound into \eqref{eq:sum_diff_filters} and applying the Cauchy--Schwarz inequality gives
    \[
    O(\sqrt{\epsilon}) \cdot \sum_{i:|\cC_i|<\theta_s} \sqrt{\frac{|\cC_i|}{k}} \cdot \|M_{k,\epsilon} \cdot x_{\cC_i}\|_2 
    \le O(\sqrt{\epsilon}) \cdot \Big(\sum_i \frac{|\cC_i|}{k}\Big)^{1/2} \cdot (\sum_i \|M_{k,\epsilon} \cdot x_{\cC_i}\|^2_2)^{1/2}.
    \]
    By Lemma~\ref{lem:total_energy}, this is at most 
    \[
    O(\sqrt{\epsilon}) \cdot \|\sum_i M_{k,\epsilon} \cdot x_{\cC_i}\|_2 = O(\sqrt{\epsilon}) \cdot \|M_{k,\epsilon} \cdot x\|_2.
    \]
    So
    $\|z - M_{k,\epsilon} x\|_{[-1,1]}^2 \le \|z - M_{k,\epsilon} x\|_{2}^2=O(\epsilon) \cdot \|M_{k,\epsilon} x\|_2^2$. 
                By Lemma~\ref{lem:filter_sparse}, 
    \[
    \|z\|_{\mathbb R\setminus[-1,1]}^2
    \le 2\|z-M_{k,\epsilon}x\|_2^2+2\|M_{k,\epsilon}x\|_{\mathbb R\setminus[-1,1]}^2
    = O(\epsilon)\|M_{k,\epsilon}x\|_2^2.
    \] 
    The reverse triangle inequality implies $\|z\|_2\ge(1-O(\sqrt{\epsilon}))\|M_{k,\epsilon}x\|_2$.
    It follows that $\|z\|_{\mathbb R\setminus[-1,1]}^2=O(\epsilon)\|z\|_2^2$, and hence $\|z\|_{[-1,1]}^2\ge(1-O(\epsilon))\|z\|_2^2$.

    Then we show the second desired property 
                \begin{align*}
    \|z-x\|_{[-1,1]}^2
    &\le2\|z-M_{k,\epsilon}x\|_{[-1,1]}^2+2\|M_{k,\epsilon}x-x\|_{[-1,1]}^2\\
    &\le2\|z-M_{k,\epsilon}x\|_2^2+\epsilon\|x\|_{[-1,1]}^2\\
    &=O(\epsilon)\|x\|_{[-1,1]}^2,
    \end{align*} 
    where the last step is by Lemma~\ref{lem:total_energy}.

    Next, we bound the $|\supp(\wh{z})|$. This is at most
    \begin{align*}
        \sum_i |I_i| \le
        & \sum_i |range(\cC_i)| + \sum_{i: |\cC_i| < \theta_s}2\Delta_{|\cC_i|,\frac{\epsilon \cdot |\cC_i|}{k}} + 2 \Delta_{k,\epsilon} \cdot \frac{k}{\theta_s} \\
        \le & \tilde{O}(k^2/\epsilon) + \sum_{i: |\cC_i| < \theta_s} O\Big(\frac{k |\cC_i|}{\epsilon} \log \frac{k}{\epsilon}\Big) + O(\Delta_{k,\epsilon}) = \Tilde{O}(k^2/\epsilon).
    \end{align*}

    Finally we remark that rescaling $\epsilon$ by a constant factor and set $x':=z$ will finish this proof.
\end{proofof}

\subsection{Proof of Lemma~\ref{lem:total_energy}}\label{sec:proof_total_energy}

Let $M=M_{k,\epsilon}$ in this section. 
We first prove the filtered estimate.
Assume that $\cC_1,\ldots,\cC_n$ are sorted by their frequencies. Recall that
two distinct clusters are correlated if their separation is less than
$2\Delta_{k,\epsilon}$. For $i\ne j$, let
\[
D_{i,j}:=\frac{k\cdot\dist(\cC_i,\cC_j)}{d_{min}}.
\]
For correlated pairs, define
\begin{equation}\label{eq:def_K}
K_{i,j}:=\frac{|\cC_i|\cdot|\cC_j|}{D_{i,j}}.
\end{equation}
Set $K_{i,j}=0$ on the diagonal and for all other pairs. Since
$\Delta_{k,\epsilon}=C_M k^2\log(k/\epsilon)/\epsilon$ and $C_H=8C_M$,
\begin{equation}\label{eq:max D in correlation}
\frac{2\Delta_{k,\epsilon}}{d_{min}}
=\frac{k}{4\sqrt{\log(k/\epsilon)}}.
\end{equation}
For a correlated pair the separation assumption therefore implies
\begin{equation}\label{eq:lower_bound_D}
 k\min\{|\cC_i|,|\cC_j|\}\le D_{i,j}
 <\frac{2k\Delta_{k,\epsilon}}{d_{min}}\le k^2 .
\end{equation}

\begin{claim}\label{clm:satisbility_of_delta_construction}
For $i\ne j$, let
\begin{equation}\label{eq:selection_on_delta}
\delta_{i,j}:=\frac{4\epsilon}{\sqrt{\log k}}\cdot K_{i,j}.
\end{equation}
Then
\begin{equation}\label{eq:weighted_pair_correlation}
|\langle M\cdot x_{\cC_i},M\cdot x_{\cC_j}\rangle|
\le \delta_{i,j}\|M\cdot x_{\cC_i}\|_2\|M\cdot x_{\cC_j}\|_2 .
\end{equation}
\end{claim}

\begin{proof}
If $\cC_i$ and $\cC_j$ are not correlated, the Fourier supports of
$M\cdot x_{\cC_i}$ and $M\cdot x_{\cC_j}$ are disjoint.
So their inner product is zero by Plancherel.

Suppose that $\cC_i$ and $\cC_j$ are correlated. By \eqref{eq:lower_bound_D}, $|\cC_i||\cC_j|/D_{i,j}\le1$.
Thus $0<\delta_{i,j}\le8\epsilon/\sqrt{\log k}\le1/2$. Moreover, by
\eqref{eq:selection_on_delta},
\[
\frac{|\cC_i||\cC_j|}{\delta_{i,j}}
\le \frac{D_{i,j}\sqrt{\log k}}{4\epsilon}
\le (k/\epsilon)^3 .
\]
It follows that
\begin{align*}
C_H\frac{|\cC_i||\cC_j|}{\delta_{i,j}}
\left(\log\frac{|\cC_i||\cC_j|}{\delta_{i,j}}+\log(k/\epsilon)\right)
\le \frac{C_H D_{i,j}\sqrt{\log k}}{\epsilon}\log(k/\epsilon)
\le D_{i,j}\frac{d_{min}}{k}=\dist(\cC_i,\cC_j).
\end{align*}
So the second property of Corollary~\ref{cor:filter_functions} proves
\eqref{eq:weighted_pair_correlation}.
\end{proof}

For a subset $S\subseteq\{1,\ldots,n\}$, let $K_S:=K[S,S]$ and $q_{max}:=\lfloor\log_2 k\rfloor$. 
We partition the indices into classes
\[
\mathcal I_q:=\{i\in S:2^q\le|\cC_i|<2^{q+1}\},
\qquad 0\le q\le q_{max}.
\]

\begin{lemma}\label{lem:bound_K_euclidean}
The submatrix $K_S=K[S,S]$ satisfies $\|K_S\|_{\op}=O(\sqrt{\log k})$.
\end{lemma}

\begin{proof}
Fix $i\in S$ and $0\le q\le\lfloor\log_2|\cC_i|\rfloor$.
We consider the indices $j\in\mathcal I_q$ for which $\cC_j$ is correlated with
$\cC_i$ and $|\cC_j|\le|\cC_i|$. 
Suppose the indices on the right of $\cC_i$ by $j_1,\ldots,j_p$ in increasing order.

For $1\le x\le p$, the distance between any two successive clusters in the list $\cC_i,\cC_{j_1},\ldots,\cC_{j_x}$ is at most $\dist(\cC_i,\cC_{j_x})$, which is less than $2\Delta_{k,\epsilon}$. 
And by the separation assumption and the definition of $\mathcal I_q$, each distance is at least $d_{min}2^q$.
Thus $D_{i,j_x} \ge \sum_{r=2}^x D_{i, j_r}  \ge kx2^q$. 
The same bound holds on the left side of $i$.

Since  $|\cC_j|<2^{q+1}$ for $j\in\mathcal I_q$, summing the two sides implies
\[
\sum_{\substack{j\in\mathcal I_q\\|\cC_j|\le|\cC_i|}} K_{i,j}^2
\le 2\sum_{x=1}^\infty
 \left(\frac{2|\cC_i|\,2^{q+1}}{kx2^q}\right)^2
\le \frac{64|\cC_i|^2}{k^2}.
\]
The diagonal and the non-correlated pairs contribute zero. 
And every unordered pair is counted from its larger cluster.
Therefore,
\[
\|K_S\|_{\op}^2\le\|K_S\|_F^2
\le \frac{128(q_{max}+1)}{k^2}\sum_{i\in S}|\cC_i|^2
=O(\log k),
\]
where the last step uses $\sum_{i\in S}|\cC_i|^2\le(\sum_{i\in S}|\cC_i|)^2\le k^2$.
\end{proof}

\begin{proofof}{Lemma~\ref{lem:total_energy}}
We fix $S$ and denote $z_i=\|M\cdot x_{\cC_i}\|_2$ for $i\in S$. By
Claim~\ref{clm:satisbility_of_delta_construction} and
Lemma~\ref{lem:bound_K_euclidean}, we have
\begin{align*}
&\left|\left\|M\cdot\sum_{i\in S}x_{\cC_i}\right\|_2^2
            -\sum_{i\in S}\|M\cdot x_{\cC_i}\|_2^2\right|\\
\quad \le & \sum_{\substack{i,j\in S\\ i\ne j}}
       |\langle M\cdot x_{\cC_i},M\cdot x_{\cC_j}\rangle|
 \le \frac{4\epsilon}{\sqrt{\log k}}z^\top K_S z\\
\quad\le & \frac{4\epsilon}{\sqrt{\log k}}\|K_S\|_{\op}\|z\|_2^2
 = O(\epsilon)\sum_{i\in S}\|M\cdot x_{\cC_i}\|_2^2 .
\end{align*}
This proves the filtered estimate.

We next prove the estimate on $\left\|\sum_{i\in S}x_{\cC_i}\right\|_{[-1,1]}^2$. For any signal $w$ with at most $k$
frequencies, we first show that
\[
\|M\cdot w\|_2^2=(1\pm O(\epsilon))\|w\|_{[-1,1]}^2 .
\]
As in the proof of Lemma~\ref{lem:filter_sparse}, let $I=[-1+2\epsilon/(Ck^2),\,1-2\epsilon/(Ck^2)]$.
Property~\ref{item:sqrt-localizer-plateau} of Claim~\ref{claim:sqrt-localizer-bounds} gives $|M(t)^2-1|\le3(\epsilon/k)^{C/2}$ on $I$, while Property~\ref{item:sqrt-localizer-global} in the same claim shows $|M(t)^2-1|\le1$ on $[-1,1]\setminus I$. 
Hence Lemma~\ref{lemma:bounds_Fourier_sparse_signals} implies
\begin{align*}
\left|\|M\cdot w\|_{[-1,1]}^2-\|w\|_{[-1,1]}^2\right|
 &\le 3(\epsilon/k)^{C/2}\|w\|_{[-1,1]}^2
      +\frac{4\epsilon}{Ck^2}\sup_{t\in[-1,1]}|w(t)|^2\\
 &\le \left(3(\epsilon/k)^{C/2}+\frac{\pi^2\epsilon}{C}\right)
               \|w\|_{[-1,1]}^2
 \le \frac{\epsilon}{2}\|w\|_{[-1,1]}^2.
\end{align*}
By Lemma~\ref{lem:filter_sparse},
\[
\|M\cdot w\|_{[-1,1]}^2\le\|M\cdot w\|_2^2
\le (1-\epsilon)^{-1}\|M\cdot w\|_{[-1,1]}^2 .
\]
Since the above estimates hold for any $k$-sparse Fourier signal, we have
\begin{align*}
\left\|\sum_{i\in S}x_{\cC_i}\right\|_{[-1,1]}^2
 &= (1\pm O(\epsilon))
       \left\|M\cdot\sum_{i\in S}x_{\cC_i}\right\|_2^2\\
 &= (1\pm O(\epsilon))\sum_{i\in S}\|M\cdot x_{\cC_i}\|_2^2\\
 &= (1\pm O(\epsilon))\sum_{i\in S}\|x_{\cC_i}\|_{[-1,1]}^2 .
\end{align*}
\end{proofof}

\subsection{Proof of Claim~\ref{clm:bound_range_clusters}}

Let $r_\ell$ denote the maximum possible length of the range of any cluster with at most $\ell$ frequencies. Since a cluster is generated by merging two smaller clusters, we have
\[
r_\ell \le \max\limits_{1 \le i \le \ell-1} \Big\{r_i+r_{\ell-i}+\min\big\{ d_{min}\cdot \min\{i,\ell-i\}, 2\Delta_{k,\epsilon}\big\}\Big\}
\]
for any $\ell \ge 2$. Based on symmetry, the upper limit of $i$ can be replaced with $\lfloor\ell/2\rfloor$.

Let $i_0:= \frac{2\Delta_{k,\epsilon}}{d_{min}}$ such that $d_{min}\cdot i_0 = 2\Delta_{k,\epsilon}$, and $C' := \frac{1}{2\log 2}$.

We prove the following hypothesis by induction in $\ell$:
\[
r_\ell \le
\begin{cases}
d_{min}\cdot C' \ell\log \ell,   &\ell\le i_0,\\
d_{min}\cdot (C' \ell\log i_0 +\ell-i_0),    & \ell>i_0.
\end{cases}
\]

The base case $\ell=1$ follows from $r_1=0$. 

Let $p_{\ell,i}:=\left(r_i+r_{\ell-i}+\min\big\{d_{min}\cdot \min\{i,\ell-i\}, 2\Delta_{k,\epsilon}\}\right) / d_{min}$. It suffices to bound different cases of $p_{\ell,i}$ for the inductive step of $r_\ell$.

\begin{enumerate}
\item If $0<i \le  \ell -i \le i_0$, we have $\ell \le 2i_0$. Thus,
\begin{align*}
p_{\ell,i}
&\le C' i\log i + C'(\ell -i)\log(\ell -i)+ i
\le C' \ell \log \ell. 
\end{align*}
Because $x\ln x +(1-x)\ln(1-x)\le 4\ln2\cdot x(x-1)$ for $0<x<1$ and $2x(x-1)+x = 2x(x-1/2)\le 0$ for $0< x <1/2$, the last inequality holds.

Since $1-\frac{1}{x} \ge \frac{\log_2 {x}}{2}$ for $1\le x\le 2$, we have $C'\ell\log i_0 + \ell - i_0\ge C'\ell\log \ell$ for $i_0\le \ell\le 2i_0$, then
$$
p_{\ell,i} \le
\begin{cases}
C'\ell \log \ell, &\ell \le i_0,\\
C'\ell\log i_0 +\ell  -i _0, &\ell > i_0.
\end{cases}
$$
    
\item If $0< i\le i_0 < \ell -i$, we have $\ell > i_0$. Thus, 
\begin{align*}
p_{\ell,i}
\le C' i\log i + (C'(\ell-i)\log i_0 + (\ell - i) - i_0) + i
\le C'\ell \log i_0 +\ell - i_0.
\end{align*}    

\item If $i_0 < i \le \ell -i $, we have $\ell > i_0$. Thus,
\begin{align*}
p_{\ell,i}
\le  (C'i\log i_0 + i - i_0) +  (C'(\ell-i)\log i_0 + (\ell - i) - i_0)  + i_0 
\le C'\ell \log i_0 + \ell - i_0.
\end{align*}
\end{enumerate}
Combining all of the above cases with $r_\ell \le d_{min} \cdot \max_{1\le i\le\lfloor\ell/2\rfloor}p_{\ell,i}$ finishes the proof of $r_{\ell}$.

Since $r_x+r_y \le r_{x+y}$ for any size $x,y$, $\sum_{i \in \{1, \ldots, n\}} |range(\cC_i)|$ is upper bounded by $r_{k}$. Combining $d_{min}:=\frac{C_H k \log^{1.5} (k/\epsilon)}{\epsilon} $ and \eqref{eq:max D in correlation} indicate a bound of $r_k$ as follows:
\begin{align*}
\sum_{i \in \{1, \ldots, n\}} |range(\cC_i)| 
& \le r_k \le d_{min}\cdot\Big(\frac{k}{2\log 2} \cdot  \log \frac{2\Delta_{k,\epsilon}}{d_{min}}+ k\Big) \\
& = O \left(\frac{k^2\log^{1.5} (k/\epsilon)}{\epsilon} 
\cdot \log k\right) = \tilde{O}(k^{2}/\epsilon).
\end{align*}

\section{Frequency Recovery}\label{sec:freq_recovery}
We show how to recover frequencies in a Fourier-sparse signal efficiently for Theorem~\ref{inform:improper_learning}
and the proper learning algorithm in Theorem~\ref{inform:stat_bound}. While Theorem~\ref{thm:multiband_approx} indicates that $x$ has a multiband approximation whose Fourier support is on $I_1 \cup I_2 \cdots \cup I_n$ of size $\tilde{O}(k^2)$, it does not give an algorithm to recover these bands. Applying the multiband recovery algorithm in Theorem~\ref{inform:learn_multiband} takes $\tilde{O}(k^2)$ samples because it does not use the fact that $x$ is $k$-Fourier-sparse. In this section, we show how to recover $I_1 \cup I_2 \cdots \cup I_n$ within $\tilde{O}(k)$ samples and $\tilde{O}(k)$ time based on the Haar-Fourier learning algorithm in Theorem~\ref{thm:improper_learning}.

Our frequency recovery is stated with respect to clusters defined in Algorithm~\ref{alg:cluster}. Let $\cC_1,\ldots,\cC_n$ be the clusters of a given Fourier-sparse signal. We still use $x_{\cC_j}:=\sum_{(\alpha,f) \in \cC_j} \alpha \cdot e^{2 \pi \bi f t}$ to be the signal induced by these clusters and $\operatorname{range}(\mathcal C_j)=[\min_{(f,\alpha)\in\mathcal C_j}f,\max_{(f,\alpha)\in\mathcal C_j}f]$ to be the range of  frequencies in a cluster $\cC$. Back to the frequency recovery, the Haar-Fourier learning algorithm in Theorem~\ref{thm:improper_learning} outputs $\tilde{O}(k)$ heavy coefficients that provide $\tilde{O}(k)$ regions (blocks) of size $O(k)$ in $\mathbb{R}$. While each heavy cluster would have a heavy Haar-Fourier coefficient nearby, it is unclear whether $\tilde{O}(k)$ regions generated by the heavy Haar-Fourier coefficients cover all frequencies or not. So the key is a postprocessing procedure to guarantee the frequencies would fall into these regions.
\begin{theorem}\label{thm:efficient-frequency-cover}
Given any $F$, sparsity $k$, error $\epsilon$, and any $y(t):=x(t)+\eta(t)$ with a $k$-Fourier-sparse signal $x(t)$ and noise $\eta$, 
Algorithm~\ref{alg:frequencies-cover}, in $\tilde{O}(k/\epsilon^2)$ samples and $\tilde{O}(k/\epsilon^2)$ time, outputs a multiband $I'=\bigcup_{r=1}^{n'} I_r'$ with $n' = \tilde{O}(k/\epsilon)$ such that w.h.p., for $x$'s clusters $\cC_1,\ldots,\cC_n$,
\begin{equation}
|I'| \le O\!\left(\frac{k^2\log^{5/2}(k/\epsilon) \log^2 F}{\epsilon^2}\right) \text{ and }
\left\|x-\sum_{\cC_j: range(\cC_j) \subset I'} x_{\cC_j}\right\|_{[-1,1]}^2  \le O(\epsilon) \cdot \|x\|_{[-1,1]}^2 + O(1) \cdot \|\eta\|_{[-1,1]}^2. \label{eq:cover-error2}
\end{equation}
In particular, the rest frequencies are negligible in the time window: $ \| \sum_{\cC_j: range(\cC_j) \not\subset I'} x_{\cC_j}\|_{[-1,1]}^2 \le O(\epsilon) \cdot \|x\|_{[-1,1]}^2 + O(1) \cdot \|\eta\|_{[-1,1]}^2$.
\end{theorem}
  Theorem~\ref{thm:efficient-frequency-cover} proves the frequency recovery part of Theorem~\ref{inform:improper_learning}. 
 Another remark is that we cannot recover all clusters because some clusters may be negligible in the time window.  In the rest of this section, we finish the proof of Theorem~\ref{thm:efficient-frequency-cover}. The proof of Lemma~\ref{lem:shift_one_freq} is deferred to Section~\ref{sec:proof_shifting}; and the proper learning algorithm in Theorem~\ref{inform:stat_bound} is deferred to Section~\ref{sec:proper_learning}.

\begin{algorithm}
\caption{Covering the principal frequency clusters}\label{alg:frequencies-cover}
\begin{algorithmic}
\Procedure{FrequenciesCover}{$y,\epsilon$}
\State $\sum_{H \in \mathcal{H}_{2N}} \beta_H \cdot \phi_H \gets $ \textsc{ImproperLearning} in Algorithm~\ref{alg:improper_learning} to $y$
\State For each $H$ with $\beta_H \neq 0$, for the support of its Haar vector $\{a_i,a_i+1,\ldots,b_i-1, b_i\} \subset \mathbb{Z}$, create a region $\widetilde{I}_i$ as $[\ell_i,r_i] \gets [a_i/2-1/4,b_i/2+1/4]$ and set its energy $m_i \gets |\beta_H|^2$
\State Merge regions $[\ell_i,r_i]$ with non-empty intersections until they are disjoint --- the new energy $m_i$ is the summation of energies over merged regions

\State Sort these regions as $\ell_1\le r_1\le\cdots\le \ell_{n'} \le r_{n'}$.
\State $\tau\gets\frac{c\epsilon^2\sum_i m_i}{k^2\log^{5/2}(k/\epsilon)\log^2 N}$ for a sufficiently small constant $c>0$.
\State $R_1 \leftarrow r_1+m_1/\tau$, $L_{n'} \leftarrow \ell_{n’}-m_{n'}/\tau$. 
\For{$i = 2$ to $n'$}
\State $R_i \leftarrow \max\{R_{i-1},r_i\}+m_i/\tau$.
\EndFor
\For{$i = n' - 1$ to $1$}
\State $L_i \leftarrow  \min\{L_{i+1},\ell_i\}-m_i/\tau$.
\EndFor
\State \Return $I' \leftarrow \bigcup_{i=1}^{n'} [L_i, R_i] $.
\EndProcedure
\end{algorithmic}
\end{algorithm}

Let $(I_i)_{1 \le i \le n}$ be the intervals of the multiband approximation defined in \eqref{eqdef:multiband_x} of Theorem~\ref{thm:multiband_approx} such that $\operatorname{range}(\mathcal{C}_i)\subseteq I_i$ and $\sum_i|I_i| = \tilde{O}(k^2)$. Based on the results in Section~\ref{sec:orthogonal} and \ref{sec:multiband_approx}, we have the following estimate, whose proof is deferred to the end of this section.
\begin{lemma} \label{lem:discare-estimate}
For any approximate signal $w$ and any set $S$ of cluster indices,
\begin{equation}\label{eq:discard}
\left\|\sum_{i\in S}x_{\mathcal{C}_i}\right\|_{[-1,1]}^2
\lesssim \int_{\bigcup_{i\in S}I_i}|\widehat w(f)|^2\,df + \|x-w\|_{[-1,1]}^2+\|w\|_{\mathbb R\setminus[-1,1]}^2
+\epsilon\|x\|_{[-1,1]}^2.
\end{equation} 
\end{lemma}
Thus it suffices to bound the error in the corresponding regions on an approximate signal.
A natural idea is to use the nearby region covered by the output of Theorem~\ref{thm:improper_learning}.
However, a good reconstruction does not essentially imply that its energy on these regions is bounded.
So we need a proper postprocess on the cover:

\begin{claim}\label{lem:principal-frequencies-cover}
For any $\tau > 0$, $I'$ returned by Algorithm~\ref{alg:frequencies-cover} satisfies
\begin{equation}\label{eq:cover-length}
|I'|\le\sum_i|[\ell_i,r_i]|+\frac{2\sum_i m_i}{\tau}.
\end{equation}
Moreover, for every $f_0\notin I'$ and $r\ge1$,
\begin{equation}\label{eq:coefficient-local-energy2}
\sum_{n:\,|n/2-f_0|\le r}|\widehat{\tilde{x}}[n]|^2\lesssim\tau r.
\end{equation}
\end{claim}
We defer the proof of this claim to the end of this section.

\begin{proof}[Proof of Theorem~\ref{thm:efficient-frequency-cover}]

Recall $I_i$ defined in \eqref{eqdef:multiband_x} of Theorem~\ref{thm:multiband_approx} have the following properties:
\begin{equation}\label{eq:cluster-enclosures}
\operatorname{range}(\mathcal{C}_i)\subseteq I_i, \qquad
|I_i|\ge k, \qquad
\sum_i|I_i|\lesssim\frac{k^2}{\epsilon}\log^{5/2}(2k/\epsilon).
\end{equation}

We run Algorithm~\ref{alg:improper_learning} on $y$ to obtain an reconstruction $\tilde{x} = \sum_H \beta_H \cdot \phi_H$.
By Theorem~\ref{thm:improper_learning}, $\tilde{x}$ has $\tilde{O}(k/\epsilon)$ non-zero entries in Haar-Fourier basis. 
Moreover, each $\phi_H$ has at most $O(k)$ consecutive Fourier coefficients of at most two kinds.
Thus $\tilde x(t)=\frac1{\sqrt2}\sum_{n=-N}^{N-1}\widehat{\tilde x}[n]e^{\pi\bi nt}$ has $\tilde{O}(k^2/\epsilon)$ non-zero Fourier coefficients. 
Moreover, there are $\widetilde O(k/\epsilon)$ segments on which $(-1)^n\widehat{\tilde x}[n]$ is constant.

For a segment of $\tilde x$ from $a_i$ to $b_i$, we assign $\widetilde{I}_i := [a_i/2-1/4,b_i/2+1/4]$ and $m_i := \sum_{n=a_i}^{b_i}|\widehat{\tilde x}[n]|^2$ (based on the Haar-Fourier orthonormal basis), which can be computed in $\tilde{O}(k/\epsilon)$ time using the representation $\sum_H \alpha_H \cdot \phi_H$.
And they satisfy
\begin{equation}\label{eq:coefficient-bands}
\sum_i m_i=\|\tilde x\|_{[-1,1]}^2,
\qquad \sum_i|\widetilde{I}_i|=\widetilde O(k^2/\epsilon).
\end{equation}
We next apply the forward and backward scans in Algorithm~\ref{alg:frequencies-cover} to $(\widetilde{I}_i,m_i)_i$ with the following threshold.
\[
\tau=\frac{c\epsilon^2\|\tilde x\|_{[-1,1]}^2}
{k^2\log^{5/2}(k/\epsilon)\log^2 N}.
\]
Claim~\ref{lem:principal-frequencies-cover} yields the stated length of $I'$.

Let $w=\mathbf1_{[-1,1]}\tilde x$. 
To bound the first term in \eqref{eq:discard}, it remains to bound $\int_{I_i}|\widehat w(f)|^2 df$ for bands $I_i$ not contained in $I'$.
By convolution, we have the identity
\[
\widehat w(f)= \sqrt{2} \sum_{n = -N}^{N-1}\widehat{\tilde x}[n] \cdot \sinc(2f - n). \]
For any $f_0\notin I'$ and $r\ge1$, Claim~\ref{lem:principal-frequencies-cover} shows
\begin{equation*}\label{eq:coefficient-local-energy}
\sum_{n:|n/2-f_0|\le r}|\widehat{\tilde x}[n]|^2 \lesssim \tau r.
\end{equation*}
Moreover, by Plancherel,
\[
\int_{-\infty}^\infty \sinc(2f - n) \sinc(2f - m) d f = \frac{1}{4} \int_{-\infty}^\infty \mathbf{1}_{[-1, 1]} e^{\pi \bi (n - m) t} d t = \frac{1}{2}\sinc(n-m) = \frac{1}{2} \mathbf{1}_{n = m}.
\]
For any interval $I_i\ni f_0$, we bound $\wh{w}(f)$ for $f \in I_i$ as follows.
\begin{align*}
\int_{I_i}\left| \sqrt{2} \sum_{n:|n/2-f_0|\le 2|I_i|} \widehat{\tilde x}[n] \cdot \sinc(2f - n)\right|^2df 
& \le \int_{-\infty}^\infty \left| \sqrt{2} \sum_{n:|n/2-f_0|\le 2|I_i|} \widehat{\tilde x}[n] \cdot \sinc(2f - n) \right|^2df \\
& \le \sum_{n:|n/2-f_0|\le 2|I_i|} |\widehat{\tilde x}[n]|^2 \tag{by orthogonality of $\sinc$} \\ 
& \lesssim \tau|I_i|. \tag{by \eqref{eq:coefficient-local-energy}}
\end{align*}
When $2^j|I_i|<|n/2-f_0|\le2^{j+1}|I_i|$, we have $|n/2-f| \gtrsim 2^j|I_i|$ and $\left| \sinc(2f - n) \right| \lesssim \frac{1}{2^j|I_i|}$ for $f \in I_i$. 
By Cauchy--Schwarz inequality,
\begin{align*}
& \int_{I_i}\left| \sqrt{2} \sum_{2^j|I_i|<|n/2-f_0|\le2^{j+1}|I_i|} \widehat{\tilde x}[n] \cdot \sinc(2f - n)\right|^2df
\\ \le & 2 \int_{I_i} \sum_{2^j|I_i|<|n/2-f_0|\le2^{j+1}|I_i|} |\widehat{\tilde x}[n]|^2 \cdot \sum_{2^j|I_i|<|n/2-f_0|\le2^{j+1}|I_i|} \left|\sinc(2f - n)\right|^2df
\\ \le & 2 |I_i| \cdot \tau 2^{j+1}|I_i| \cdot  2^{j+1}|I_i| \left( \frac{1}{2^j|I_i|} \right)^2 \lesssim \tau|I_i|.
\end{align*}
Since there are at most $2 N$ frequencies in $\tilde{x}$,
\[
\int_{I_i}|\widehat w(f)|^2\,df \lesssim\tau|I_i|\log^2 N.
\]
For $S=\{i:I_i\nsubseteq I'\}$, every $I_i$ contains such an $f_0$.
Hence
\[
\int_{\bigcup_{i\in S}I_i}|\widehat w(f)|^2 d f
\lesssim\tau\log^2 N \sum_i|I_i|
\lesssim\epsilon\|\tilde x\|_{[-1,1]}^2,
\]
where the last step is by $\sum_i|I_i|\lesssim\frac{k^2}{\epsilon}\log^{5/2}(2k/\epsilon)$.
Take $S=\{i:\operatorname{range}(\cC_i)\nsubseteq I'\}$ in Lemma~\ref{lem:discare-estimate}. Each such $I_i$ contains an uncovered point. Since $w$ vanishes outside $[-1,1]$, the lemma bounds the uncovered energy by $O(\epsilon)\|\tilde x\|_{[-1,1]}^2+O(1)\|x-\tilde x\|_{[-1,1]}^2+O(\epsilon)\|x\|_{[-1,1]}^2$. The triangle inequality and Theorem~\ref{thm:improper_learning} implies \eqref{eq:cover-error2}.
\end{proof}

\subsection{Proof of Lemma~\ref{lem:shift_one_freq}}
\label{sec:proof_shifting}

We present the proof of Lemma~\ref{lem:shift_one_freq} and discuss one important corollary of this result. A direct corollary of Lemma~\ref{lem:shift_one_freq} provides a net on ``off-grid" frequencies, which will be used in the proper learning algorithm of Theorem~\ref{inform:stat_bound}. Since its proof is identical to the proof of Lemma 2.1 in \cite{CKPS17}, we omit it here.

\begin{theorem}\label{thm:net_frequency}
Given any $\epsilon$, let $\mathcal{N}:=\frac{\epsilon}{C k^2} \cdot\mathbb{Z} \cap [-F,F]$ be the net of frequencies for a large constant $C$. For any $x(t):=\sum_{j=1}^k \alpha_j e^{2 \pi \bi f_j t}$ with $k$ arbitrary frequencies in $[-F,F]$, there exists $x'(t):=\sum_{j=1}^k \alpha'_j e^{2 \pi \bi f'_j t}$ whose frequencies $f'_1,\ldots,f'_k$ are in $\mathcal{N}$ such that
\[
\|x'-x\|_{[-1,1]} \le \epsilon \cdot \|x\|_{[-1,1]}.
\]
\end{theorem}

We now prove Lemma~\ref{lem:shift_one_freq}. Following \cite{CKPS17}, let $v_j(t)=e^{2\pi\bi f_jt}$ and $V_i=\operatorname{span}\{v_1,\ldots,v_i\}$, with $V_0=\{0\}$.
We define $v_k^{\parallel}$ and $v_{k+1}^{\parallel}$ to be the  projections of $v_k$ and $v_{k+1}$ onto $V_{k-1}$. Then $v_k^{\perp}:=v_k-v_k^{\parallel}$ and $v_{k+1}^{\perp}:=v_{k+1}-v_{k+1}^{\parallel}$ are the orthogonal parts.

Furthermore, let $w$ be the component of $v_{k+1}^{\perp}$ orthogonal to $v_k^{\perp}$ such that $v_{k+1}^{\perp} - w$ is parallel to $v_k^{\perp}$.  
Because both $v_k^{\perp}$ and $v_{k+1}^{\perp}$ lie in $V_{k-1}^{\perp}$, $w$ is orthogonal to both $V_{k-1}$ and $v_k^{\perp}$. So it is the component of $v_{k+1}$ orthogonal to $V_k=V_{k-1}+\mathrm{span}\{v_k^{\perp}\}$.

As shown in Lemma 8.5 of \cite{CKPS17}, the relative error of replacing $f_k$ by $f_{k+1}$ is at most $\frac{\|w\|_{[-1,1]}^2}{\|v_{k+1}^{\bot}\|^2_{[-1,1]}}$. Explicitly,
\begin{align}
w = v_{k+1}^{\perp} -\frac{\langle v_{k+1}^{\perp},v_k^{\perp}\rangle_{[-1,1]}} {\|v_k^{\perp}\|_{[-1,1]}^2}v_k^{\perp} \quad \text{and} \quad \|w\|_{[-1,1]}^2 = \|v_{k+1}^{\perp}\|_{[-1,1]}^2 -\frac{|\langle v_k^{\perp},v_{k+1}^{\perp}\rangle_{[-1,1]}|^2} {\|v_k^{\perp}\|_{[-1,1]}^2 }. \label{eq:w-orthogonality-bound}
\end{align}
Observe that for $b:=\frac{\langle v_k^\bot,v_{k+1}^\bot\rangle_{[-1,1]}}{\|v_{k+1}^\bot\|^2_{[-1,1]}}$,
\[\min_{z \in \mathbb C} \|v_k^{\perp}-z v_{k+1}^{\perp}\|^2_{[-1,1]}
=\|v_k^{\perp}-b v_{k+1}^{\perp}\|^2_{[-1,1]}
=\|v_k^{\perp}\|_{[-1,1]}^2 - \frac{|\langle v_k^{\bot},v_{k+1}^{\bot} \rangle|^2}{\|v_{k+1}^\bot\|^2_{[-1,1]}} 
.\] So we set $x'= \sum_{j=1}^{k-1} \alpha'_j v_j + \alpha_k b v_{k+1}^{\perp}$ to approximate $x= \sum_{j=1}^k \alpha_j v_j$, which replaces $f_k$ by $f_{k+1}$ and keeps the component in $V_{k-1}$ the same. Hence,
\begin{align}
\|x'-x\|_{[-1,1]}^2 = |\alpha_k|^2 \cdot \|v_k^{\perp}-b v_{k+1}^{\perp}\|^2_{[-1,1]}
\le \frac{\|v_k^{\perp}-b v_{k+1}^{\perp}\|^2_{[-1,1]}}{\|v_k^{\perp}\|_{[-1,1]}^2} \cdot \|x\|^2_{[-1,1]}, \label{eq:one-shift-first-bound}
\end{align}
where we use that $|\alpha_k|\|v_k^{\perp}\|_{[-1,1]}\le\|x\|_{[-1,1]}$ in the last step. Moreover, by the guarantee of $b$,
\begin{align}
\frac{\|v_k^{\perp}-b v_{k+1}^{\perp}\|_{[-1,1]}^2}{\|v_k^{\bot}\|_{[-1,1]}^2}
= 1  - \frac{|\langle v_k^{\perp},v_{k+1}^{\perp}\rangle_{[-1,1]}|^2} {\|v_k^{\perp}\|_{[-1,1]}^2 \cdot \|v_{k+1}^{\perp}\|_{[-1,1]}^2} = \frac{\|w\|_{[-1,1]}^2}{\|v_{k+1}^\bot\|_{[-1,1]}^2}. \label{eq:b-orthogonality-bound}
\end{align}

Unlike Lemma 8.5 in \cite{CKPS17}, the rest of this proof provides a new bound on the quotient $\|w\|^2_{[-1,1]} / \|v_{k+1}^{\perp}\|^2_{[-1,1]}$, based on the Poincare inequality via the Volterra operator. The proof idea is to present a decomposition of $v_{k+1}$ into two vectors $w$ and $(v_{k+1} - w)$ in $\mathrm{span}\{v_1,\ldots,v_k\}$. Since any upper bound on $\|w\|$ provides an upper bound on the component of $v^{\bot}_{k+1}$ orthogonal to $v^{\bot}_k$, $w \bot \mathrm{span}\{v_1,\ldots,v_k\}$ is no longer necessary. 
From now on, we fix $f_1,\ldots,f_k$ and $f_{k+1}$. For any $g:[-1,1] \to \mathbb{C}$ and $t \in [-1,1]$, we define 
\begin{equation}\label{eq:def_operator_S}
(Sg)(t)=e^{2\pi\bi f_kt} \int_0^t e^{-2\pi\bi f_ks}g(s)\,\mathrm{d}s. 
\end{equation}
For any $j \neq k$ with $v_j(t)=e^{2 \pi \bi f_j t}$,
\begin{equation}
S v_j(t) = \frac{e^{2 \pi \bi f_j t} - e^{2 \pi \bi f_k t} }{2 \pi \bi(f_j-f_k)}=(v_j-v_k)/(2\pi\bi(f_j-f_k)). \label{eq:S_v_i}
\end{equation}
So $S(u) \in V_k$ for any $u \in V_{k-1}$.

Since \eqref{eq:S_v_i} also holds for $v_{k+1}$, rearranging \eqref{eq:S_v_i} for $v_{k+1}$ shows
\begin{align}
    v_{k+1} = & 2 \pi \bi (f_{k+1} - f_k) \cdot S v_{k+1} +  v_k \notag \\
    = & 2 \pi \bi (f_{k+1} - f_k) \cdot S v_{k+1}^{\perp} + 2 \pi \bi (f_{k+1} - f_k) \cdot S v_{k+1}^{\parallel} + v_k \notag \\
    = & 2 \pi \bi (f_{k+1} - f_k) \cdot ( S v_{k+1}^{\perp} - r \cdot v_k) + \underbrace{2 \pi \bi (f_{k+1} - f_k) \cdot S v_{k+1}^{\parallel} + [ 2 \pi \bi (f_{k+1} - f_k) r + 1 ] \cdot v_k}_{\in V_k}, \label{eq:decom_v_k_1}
\end{align}
where $r$ is a parameter chosen later.

Because $v_{k+1}- \underbrace{2 \pi \bi (f_{k+1} - f_k) \cdot ( S v_{k+1}^{\perp} - r \cdot v_k)}_{\text{ first part in \eqref{eq:decom_v_k_1}}} \in V_k$ by \eqref{eq:decom_v_k_1}, $w$ as the component of $v_{k+1}$ orthogonal to $V_k$ satisfies
\begin{equation}
    \| w \|_{[-1,1]} \le \| 2 \pi \bi (f_{k+1} - f_k) \cdot ( S v_{k+1}^{\perp} - r \cdot v_k) \|_{[-1,1]} \label{eq:w-orthogonality-bound2}.
\end{equation}
For convenience, let $(Tg)(t):=\int_0^t e^{-2\pi\bi f_ks}g(s)\,\mathrm{d}s$ be the integral operator in \eqref{eq:def_operator_S} such that \eqref{eq:w-orthogonality-bound2} becomes
\[
\| w \|_{[-1,1]}  \le 2 \pi | f_{k+1} - f_k | \cdot \left\| T v_{k+1}^{\perp} - r \right\|_{[-1,1]}.
\]

Finally, we apply Poincar\'{e}'s inequality with a proper $r$.
\begin{lemma}[Poincar\'{e}'s inequality \cite{stein2011fourier}]
    If $f$ is continuously differentiable on $[-1, 1]$, 
    \begin{align*}
        \int_{-1}^{1} \left| f(t) - \frac{1}{2} \int_{-1}^{1} f(s) ds \right|^2 dt \le \frac{4}{\pi^2} \int_{-1}^{1} |f'(t)|^2 dt.
    \end{align*}
\end{lemma}

Applying Poincar\'{e}'s inequality with $f:=T v_{k+1}^{\perp}$ and $r := \frac{1}{2} \int_{-1}^{1} f(s) ds$, we have
\begin{align*}
    \left\| T v_{k+1}^{\perp} - r \right\|_{[-1,1]} \le \frac{2}{\pi} \| v_{k+1}^{\perp} \|_{[-1,1]}.
\end{align*}

Therefore, with \eqref{eq:w-orthogonality-bound2},
\begin{align*}
    \| w \|_{[-1,1]} \le 4 | f_{k+1} - f_k | \cdot \| v_{k+1}^{\perp} \|_{[-1,1]}.
\end{align*}
Combining this estimate with \eqref{eq:one-shift-first-bound} and \eqref{eq:b-orthogonality-bound} gives \[\|x'-x\|_{[-1,1]}\le4|f_{k+1}-f_k|\cdot\|x\|_{[-1,1]},\] proving Lemma~\ref{lem:shift_one_freq}.

\subsection{Proper Learning}\label{sec:proper_learning}
\begin{corollary}\label{cor:query_complexity_learning}
    Given any $F\ge0$ and $0<\epsilon<1$, let $y(t):=x(t)+\eta(t)$ for $x(t):=\sum_{j=1}^k \alpha_j e^{2 \pi \bi f_j t}$ be our observation over $[-1,1]$ with $k$ arbitrary frequencies $f_1,\ldots,f_k \in [-F,F]$ and $\|\eta(t)\|_{[-1,1]}^2 \le \epsilon \cdot \|x(t)\|_{[-1,1]}^2$. There exists an algorithm that takes $\tilde{O}(k/\epsilon^2)$ samples and $\tilde{O}(k/\epsilon)^{O(k)}$ time to output a $k$-Fourier-sparse signal $\tilde{x}$ such that with probability 0.6,
    \[
    \|\tilde{x}-x\|_{[-1,1]}^2 = O(\epsilon \cdot \|x\|_{[-1,1]}^2 + \|\eta\|_{[-1,1]}^2).
    \]
\end{corollary}

\begin{algorithm}[t]
\caption{Recover k-sparse-Fourier signal}\label{alg:net_learning}
\begin{algorithmic}[1]
\Procedure{SparseFT}{$y,\epsilon$}
\State $I \gets\textsc{FrequenciesCover}(y,\epsilon)$.
\State $\rho \gets \sqrt{\epsilon}/(Ck^2)$ for a large constant $C$.
\State $\mathcal{N}' \gets \rho \cdot \mathbb Z\cap (I + [-k\rho, k\rho])$, $m \gets k \cdot (\log Fk)^{O(1)}$.
\State Draw $t_1, \cdots, t_m$ independently from the density $D_{2k}$.
\State Query $y(t_1), \cdots, y(t_m)$.
\ForAll{$A \subset \mathcal{N}'$ with $|A|\le k$}
    \State $x_A\gets\argmin_{z \in V_A}\|y-z\|_{D_{2k}, m}$ where $V_A :=\operatorname{span}\{e^{2\pi\bi ft}:f\in A\}$
\EndFor
\State \Return $\tilde{x} \gets \argmin_{x_A: A \subset \mathcal{N}', |A| \le k}\|y- x_A\|_{D_{2k}, m}$
\EndProcedure
\end{algorithmic}
\end{algorithm}

\begin{proof}[Proof of Corollary~\ref{cor:query_complexity_learning}]
Suppose $x = \sum_{i=1}^k \alpha_i e^{2 \pi \bi f_i t}$ with $f_1<\cdots<f_k$.
We may assume that all of its frequency clusters lie on $I$; otherwise, we treat the component of the external frequency clusters as a part of noise and pay an extra error, $O(\epsilon \cdot \|x\|_{[-1,1]}^2 + \|\eta\|_{[-1,1]}^2)$, with Theorem~\ref{thm:efficient-frequency-cover}.
We define
\begin{align*}
f'_1 = \rho \cdot \lceil f_1/\rho \rceil , \qquad f'_j=\max\big\{\rho \cdot \lceil f_j/\rho \rceil ,f'_{j-1}+\rho\big\} \quad (2\le j\le k).
\end{align*}
Each frequency moves by at most $k \rho$. Applying Lemma~\ref{lem:shift_one_freq} at each step yields a signal $x'$ such that
\begin{equation}\label{eq:proper-grid-approximation}
    \norm{x'-x}
    \le \bigl[(1+ O(k \rho))^k - 1\bigr] \|x\|_{[-1,1]}
    \le (e^{O(\sqrt{\epsilon})}-1) \cdot \|x\|_{[-1,1]}
    \le \sqrt{\epsilon} \cdot \|x\|_{[-1,1]},
\end{equation}

Apply Theorem~\ref{inform:stat_bound} with sparsity $2k$ and accuracy $1/2$ to $\tilde{x}$.
With probability at least $0.98$ and $m=O\bigl(k\log^7(Fk)\bigr)$, for every signal $x'$ with bandlimit $F$,
\begin{equation}\label{eq:proper-sample-bound}
    \frac12\|\tilde{x}-x'\|_{[-1,1]}^2
    \le \|\tilde{x}-x'\|_{D_{2k}, m}^2
    \le \frac32\|\tilde{x}-x'\|_{[-1,1]}^2.
\end{equation}

Let $r=y-x'=\eta+(x-x')$. 
With $\mathbb E[\|r\|_{D_{2k}, m}^2] = \|r\|_{[-1,1]}^2$,
Markov's inequality shows that, with probability at least $0.99$,
\begin{equation}\label{eq:proper-noise-bound}
    \|r\|_{D_{2k}, m}^2  \le 100 \|r\|_{[-1,1]}^2.
\end{equation}
The events in \eqref{eq:proper-sample-bound} and \eqref{eq:proper-noise-bound} therefore hold simultaneously with probability at least $0.97$.

Assume both events happen. By the choice of $\tilde{x}$ and the feasibility of $x'$,
$\|\tilde{x}-y\|_{D_{2k}, m}\le\|x'-y\|_{D_{2k}, m}$. Hence,
\begin{align*}
\|\tilde{x}-x'\|_{[-1,1]}
&\le\sqrt2\|\tilde{x}-x'\|_{D_{2k}, m} \tag{by event \eqref{eq:proper-sample-bound}}\\
&\le\sqrt2\bigl( \|\tilde{x}-y\|_{D_{2k}, m} +\|y-x'\|_{D_{2k}, m}\bigr) \\
&\le 2 \sqrt2 \|y-x'\|_{D_{2k}, m} \tag{$\|\tilde{x}-y\|_{D_{2k}, m} \le \|y-x'\|_{D_{2k}, m}$}\\
& = 2 \sqrt2 \|r\|_{D_{2k}, m} \\
&\le 20 \sqrt2 \|r\|_{[-1,1]}. \tag{by event \eqref{eq:proper-noise-bound}}
\end{align*}
Combining this with \eqref{eq:proper-grid-approximation} implies
\[
\|\tilde{x}-x\|_{[-1,1]}^2 \lesssim \|\tilde{x}-x'\|_{[-1,1]}^2 + \norm{x'-x}_{[-1,1]}^2 = O\bigl(\|\eta\|_{[-1,1]}^2+\epsilon\|x\|_{[-1,1]}^2\bigr).
\]

Finally, $|\mathcal N'| = \tilde{O}(k^4/\epsilon^{5/2})$ in Algorithm~\ref{alg:net_learning}, and the number of supports $A$ is at most $|\mathcal N'|^k$.
Thus the total running time is $\tilde{O}(k/\epsilon)^{O(k)}$.
\end{proof}

\subsection{Proofs of Lemma~\ref{lem:discare-estimate} and Claim~\ref{lem:principal-frequencies-cover}}

\begin{proofof}{Lemma~\ref{lem:discare-estimate}}
By Lemma~\ref{lem:filter_sparse} and the triangle inequality, 
\[\left(1-\sqrt{\epsilon/2}\right)\left\|\sum_{i\in S}x_{\mathcal{C}_i}\right\|_{[-1,1]}\le\left\|M_{k,\epsilon}\sum_{i\in S}x_{\mathcal{C}_i}\right\|_2.\] 
Thus $\|\sum_{i\in S}x_{\mathcal{C}_i}\|_{[-1,1]}^2\lesssim\|M_{k,\epsilon}\sum_{i\in S}x_{\mathcal{C}_i}\|_2^2$.
And
\begin{align}
\left\|M_{k,\epsilon}\sum_{i\in S}x_{\mathcal{C}_i}\right\|_2^2 
& \le \left| \left\langle M_{k,\epsilon}\sum_{i\in S}x_{\mathcal{C}_i},M_{k,\epsilon}x \right\rangle \right| + \left| \left\langle M_{k,\epsilon}\sum_{i\in S}x_{\mathcal{C}_i},M_{k,\epsilon} \sum_{i\notin S}x_{\mathcal{C}_i} \right\rangle \right| \notag \\
& \le \left| \left\langle M_{k,\epsilon}\sum_{i\in S}x_{\mathcal{C}_i},M_{k,\epsilon}x \right\rangle \right| + O(\epsilon)\|x\|_{[-1,1]}^2, \label{eq:self-pairing}
\end{align}
where we use Lemma~\ref{lem:total_energy} to bound the cross term in the last step.
Then by triangle inequality and Cauchy--Schwarz inequality, we obtain
\begin{align}
\left| \left\langle M_{k,\epsilon}\sum_{i\in S}x_{\mathcal{C}_i},M_{k,\epsilon}x \right\rangle \right| 
& \le \left| \left\langle M_{k,\epsilon}\sum_{i\in S}x_{\mathcal{C}_i},w \right\rangle \right|
+ \left| \left\langle M_{k,\epsilon}\sum_{i\in S}x_{\mathcal{C}_i}, M_{k,\epsilon}x-w \right\rangle \right| \notag \\
& \le \left| \left\langle M_{k,\epsilon}\sum_{i\in S}x_{\mathcal{C}_i}, w \right\rangle \right|
+ \left\| M_{k,\epsilon}\sum_{i\in S}x_{\mathcal{C}_i} \right\|_2 \|M_{k,\epsilon}x-w\|_2. \label{eq:discare-estimate-step1}
\end{align}

To bound the first term in \eqref{eq:discare-estimate-step1}, let $z := \sum_{i\in S}M_{|\mathcal{C}_i|,2\epsilon|\mathcal{C}_i|/k}x_{\mathcal{C}_i}$.
Then Plancherel and Cauchy--Schwarz show
\begin{align}
\left|\left\langle M_{k,\epsilon}\sum_{i\in S}x_{\mathcal{C}_i},w\right\rangle\right|
\le & \left|\left\langle z,w\right\rangle\right| + \left|\left\langle z - M_{k,\epsilon}\sum_{i\in S}x_{\mathcal{C}_i},w\right\rangle\right| \tag{by $\supp(\widehat{z}) \subseteq \bigcup_{i\in S}I_i$} \\
\le & \| z \|_2 \left[\int_{\bigcup_{i\in S}I_i}|\widehat w(f)|^2\,df\right]^{1/2} + \left\|z-M_{k,\epsilon}\sum_{i\in S}x_{\mathcal{C}_i}\right\|_2 \| w \|_2
\label{eq:observed-pairing}
\end{align}
And we have
\begin{align*}
\left\|z-M_{k,\epsilon}\sum_{i\in S}x_{\mathcal{C}_i}\right\|_2
& \lesssim\sum_{i\in S} \sqrt{\frac{\epsilon|\mathcal{C}_i|}{k}}\, \|M_{k,\epsilon}x_{\mathcal{C}_i}\|_2 \tag{by Corollary~\ref{cor:filter_functions}} \\
&\le\sqrt{\frac{\epsilon\sum_{i\in S}|\mathcal{C}_i|}{k}} \left(\sum_{i\in S}\|M_{k,\epsilon}x_{\mathcal{C}_i}\|_2^2\right)^{1/2} \tag{by Cauchy--Schwarz inequality} \\
&\lesssim\sqrt\epsilon\left\|M_{k,\epsilon}\sum_{i\in S}x_{\mathcal{C}_i}\right\|_2. \tag{by Lemma~\ref{lem:total_energy}}
\end{align*}

Substituting these estimates into \eqref{eq:self-pairing}, we obtain
\[
\left\| M_{k,\epsilon}\sum_{i\in S}x_{\mathcal{C}_i} \right\|_2^2
\lesssim \left\| M_{k,\epsilon}\sum_{i\in S}x_{\mathcal{C}_i} \right\|_2
\left( \|M_{k,\epsilon}x-w\|_2
+ \left[ \int_{\bigcup_{i\in S}I_i}|\widehat w(f)|^2\,df \right]^{1/2} + \sqrt\epsilon\|w\|_2 \right)
+ \epsilon \|x\|_{[-1,1]}^2.
\]
Since $ab\le \lambda a^2+b^2/(4\lambda)$, with $\lambda$ sufficiently small, we bound the first term on the right-hand by $O\Big( \left\| M_{k,\epsilon}\sum_{i\in S}x_{\mathcal{C}_i} \right\|_2^2 + 
\|M_{k,\epsilon}x-w\|_2^2
+ \int_{\bigcup_{i\in S}I_i}|\widehat w(f)|^2\,df 
+ \epsilon\|w\|_2^2 \Big)$.
With a sufficient small hided constant, we can absorb the factor of $\left\| M_{k,\epsilon}\sum_{i\in S}x_{\mathcal{C}_i} \right\|_2^2$ into the left-hand side. 
So
\[
\left\| M_{k,\epsilon}\sum_{i\in S}x_{\mathcal{C}_i} \right\|_2^2
\lesssim 
\|M_{k,\epsilon}x-w\|_2^2
+ \int_{\bigcup_{i\in S}I_i}|\widehat w(f)|^2\,df + \epsilon\|w\|_2^2 
+ \epsilon \|x\|_{[-1,1]}^2.
\]
By Lemma~\ref{lem:filter_sparse},
\begin{align*}
\|M_{k,\epsilon}x-w\|_2 & \le \|M_{k,\epsilon}x-w\|_{[-1, 1]} + \|M_{k,\epsilon}x-w\|_{\R \setminus [-1, 1]} \\
& \le \|M_{k,\epsilon}x-x\|_{[-1, 1]} + \|x - w\|_{[-1, 1]} + \|M_{k,\epsilon}x\|_{\R \setminus [-1, 1]} + \|w\|_{\R \setminus [-1, 1]} \\
& \lesssim \|x - w\|_{[-1, 1]} + \|w\|_{\R \setminus [-1, 1]} + \sqrt{\epsilon} \|x\|_{[-1, 1]}.
\end{align*}
And $\|w\|_2^2 \lesssim \|x\|_{[-1, 1]}^2 + \|x - w\|_{[-1, 1]}^2 + \|w\|^2_{\R \setminus [-1,1]}$.

Combining all the above estimates proves~\eqref{eq:discard}.
\end{proofof}

\begin{proof}[Proof of Claim~\ref{lem:principal-frequencies-cover}] 
Each forward extension adds at most $m_i/\tau$ to $\bigcup_i[r_i,R_i]$, and the backward scan has the same bound. 
Therefore
\[
\left|I'\right|\le\sum_i|I_i'|+2\sum_i m_i/\tau,
\]
which proves~\eqref{eq:cover-length}.

Moreover, the cover $I'$ output by Algorithm~\ref{alg:frequencies-cover} is equal to
\begin{equation}\label{eq:def_I'}
    I'' = \bigcup \bigg\{ [\ell_j-\sum_{h=i}^j m_h/\tau, r_i+\sum_{h=i}^j m_h/\tau]: 1 \le i \le j \le n', \ell_j - r_i \le \sum_{h=i}^j m_h/\tau \bigg\}
\end{equation}
By the definition of $L_i$ and $R_i$, $R_i=\max_{j \le i}(r_j+\sum_{h=j}^i m_h/\tau)$ and $L_i=\min_{j \ge i}(\ell_j-\sum_{h=i}^j m_h/\tau)$.
So $[\ell_i, R_i]$ and $[L_i,r_i]$ lie in $I''$.
Conversely, for any new point $f \in I'' \setminus \bigcup_{i=1}^{n'} [\ell_i, r_i]$ with $r_h<f<\ell_{h+1}$, at least one of $f-\ell_j \le \sum_{u=j}^h m_u/\tau$ and $r_j-f \le \sum_{u=h+1}^j m_u / \tau$ holds, 
which implies that $f$ is in either $[r_h,R_h]$ or $[L_{h+1},\ell_{h+1}]$.

Then consider an arbitrary interval $J = [s, t]$ with $\sum_{n:\,n/2\in[s,t]}|\widehat{\tilde x}[n]|^2>\tau(t-s)$.
Let $\widetilde{I}_i$ and $\widetilde{I}_j$ be the first and last bands that have non-empty intersections with $J$. Then we have $s \le r_i$ and $\ell_j \le t$. Moreover, $\ell_j - r_i \le t-s \le \frac{\sum_{n:\,n/2\in[s,t]}|\widehat{\tilde x}[n]|^2}{\tau} \le \frac{\sum_{h=i}^j m_h}{\tau}$ such that $[\ell_j-\sum_{h=i}^j m_h/\tau, r_i+\sum_{h=i}^j m_h/\tau]$ has been added to the list.
By \eqref{eq:def_I'}, we obtain
\[
J = [s,t] \subseteq \left[t-\frac{\sum_{h=i}^j m_h}{\tau},\ s+\frac{\sum_{h=i}^j m_h}{\tau}\right] 
\subseteq \left[\ell_j-\frac{\sum_{h=i}^j m_h}{\tau},\ r_i+\frac{\sum_{h=i}^j m_h}{\tau}\right].
\]
Since $f_0\notin I'$, the interval $[f_0-r,f_0+r]$ is not contained in $I'$. The above argument therefore implies $\sum_{n:|n/2-f_0|\le r}|\widehat{\tilde x}[n]|^2\le2\tau r$, which proves \eqref{eq:coefficient-local-energy}.
\end{proof}

\section{Multiband Recovery}\label{sec:multiband_recovery}
In this section,
we show how to recover locations of bands in $\wh{x}$ for Theorem~\ref{inform:learn_multiband}. The full proof of Theorem~\ref{inform:learn_multiband} is deferred to Appendix~\ref{sec:proof_main}.
For ease of exposition, we assume that the $n$ intervals in $\wh{x}$ have length at most a given parameter $R$. Recall that $mid(J)=(a+b)/2$ denotes the midpoint of an interval  $J=[a,b]$. For the noisy observation $w=x+\eta$, we call $J$ heavy only if it satisfies
\begin{equation}\label{eq:cond_heavy}    
\int_{J} |\wh{w}(f)|^2 \mathrm{d} f \ge \frac{\epsilon}{4n} \cdot \|w\|_2^2.
\end{equation}

\begin{theorem}\label{thm:locations_multiband}
Given parameters $F$, $R$, $n$, and $\epsilon$, let $x$ be a multiband signal whose Fourier spectrum $\supp(\wh{x})=J_1 \cup J_2 \cup \cdots \cup J_n$ in $[-F,F]$ has $n$ \emph{unknown} intervals of length at most $R$. Given the time window $[-1,1]$, suppose $x$ satisfies                \begin{equation}\label{cond:multiband_recovery_intro2} 
        \int_{-1}^1 |x(t)|^2 \mathrm{d} t \ge (1-\epsilon) \int |x(t)|^2 \mathrm{d} t.
    \end{equation} 
    
     For any noisy observation $w(t)=x(t)+\eta(t)$ in $[-1,1]$ with $\|\eta\|_{[-1,1]}^2 \le \epsilon \cdot \|x\|_{[-1,1]}^2$,     there exists an efficient algorithm to return $\ell=\tilde{O}(n)$ frequencies $\tilde{f}_1,\ldots,\tilde{f}_\ell$ in $\tilde{O}(n \cdot R)$ samples and $\tilde{O}(n \cdot R)$ time such that with probability 0.99, for any heavy $J_j$ satisfying \eqref{eq:cond_heavy}, $\exists \tilde{f}_i$ with $|\tilde{f}_i - mid(J_j)| = O(R)$.
\end{theorem}
\begin{remark}\label{remark:recover_multiband}
Condition \eqref{eq:cond_heavy} is necessary to recover $J_j$ accurately because the adversarial noise $\eta$ may erase an interval $J_j$ in $\wh{x}$ (and $\wh{w}$). Here is a counter-example. For any $M = \Omega(\sqrt{1/\epsilon})$ and any $\delta \le \epsilon/M$, consider two shifted and rotated box functions $g(t)=e^{2 \pi \bi (t - 1)M}$ for $t \in [1 - \delta, 3-\delta]$ and $h(t)=e^{-2 \pi \bi (t - 1) M}$ for $t \in [1 - \delta, 3-\delta]$ ($g(t)=h(t)=0$ for $t\notin [1 - \delta, 3-\delta]$). In the time domain, they are almost the same in $[-1,1]$ since $\delta M \le \epsilon$.
In the Fourier domain, $\wh{g}$ and $\wh{h}$ are concentrated around $[M - O(\sqrt{\frac{1}{\epsilon}}),M + O(\sqrt{\frac{1}{\epsilon}})]$ and $[-M - O(\sqrt{\frac{1}{\epsilon}}), -M + O(\sqrt{\frac{1}{\epsilon}})]$ respectively. After truncating $\wh{g}$ and $\wh{h}$ into single-bands, it is still impossible to distinguish $g$ and $h$ in the time window under adversarial noise. Thus, $\supp(\wh{g})$ and $\supp(\wh{h})$ can not be recovered accurately unless we know $\delta$. On the other hand, $\delta$ is exactly a quantity identified in Corollary~\ref{cor:interpolating_multiband} for relaxing \eqref{cond:multiband_recovery_intro2}.

Moreover, our algorithm works for the case when $\supp(\wh{x})=I_1\cup \cdots \cup I_n$ has some intervals of length more than $R$. Because this algorithm only needs $R$ and $n$, given a proper parameter $R$ in this case, one partitions $I_i$ into $\lceil \frac{|I_i|}{R} \rceil$ intervals of length $R$ and sets $n$ to be a upper bound on $\sum_{I_i} \lceil \frac{|I_i|}{R} \rceil$.    
\end{remark}

The key components of Theorem~\ref{thm:locations_multiband} are two procedures: \textsc{HashToBins} and \textsc{GetSamples} in Algorithm~\ref{alg:recover_cluster1}. Basically, \textsc{HashToBins} uses $G$ and $\wh{G}$ defined below to hash all bands of $\wh{x}$ into $O(n)$ buckets; and \textsc{GetSamples} provides an efficient method to estimate the band in every bucket under this hash.
In the rest of this section, we analyze \textsc{HashToBins} and \textsc{GetSamples} before proving Theorem~\ref{thm:locations_multiband}. Recall that $w(t)=x(t)+\eta(t)$ with noise. Algorithm~\ref{alg:recover_cluster1} assumes $w(t)=0$ for $t \notin [-1,1]$ --- this is because Condition~\eqref{cond:multiband_recovery_intro2}
only increases $\|\eta\|_2^2 \le 2\epsilon \cdot \|x\|_2^2$. 

\paragraph{Properties of $(G,\wh{G})$.} Classical works in sparse Fourier transforms \cite{HIKP,PS15,K16,CKPS17} rely on a pair of filter functions $G$  and $\wh{G}$ such that $G$ is compact and $\wh{G}$ looks like a box function. We state its constructions and properties in the continuous setting here from \cite{CKPS17}.

Given parameters $B=O(n)$, error $\delta>0$, a constant $\alpha < 0.1$, $\ell$, $s$, and a normalizer factor $b_0$ for $\wh{G}(0)=1$, let
\begin{align}
    G(t) & := b_0 \bigg(\rect_{\frac{B}{\alpha}}(t) \bigg)^{* \ell} \cdot \sinc\Big(\frac{t}{B \cdot s}\Big); \label{def:G} \\
    \wh{G}(f) & := b_0 \cdot \bigg(\sinc\Big(\frac{B}{\alpha}f\Big) \bigg)^{\ell} * \rect_{\frac{1}{B \cdot s}}(f). \label{def:hatG} \end{align} By setting $\ell:=\Theta(\log (n/\delta))$, $s:=(1-\alpha/2)^{-1}$, and $b_0:=\Theta(\sqrt{\ell}/\alpha)$, $(G,\wh{G})$ have the following properties. 
\begin{enumerate}
    \item $\supp(G) \subset [\frac{\ell}{2} \cdot \frac{-B}{\alpha},\frac{\ell}{2} \cdot \frac{B}{\alpha}]$ and $G(t) \le b_0 = O(\sqrt{\ell}/\alpha)$.     \item $|\wh{G}(f)| \in [1-(\delta/n),1]$ for $|f| \le (1-\alpha/2)\frac{1}{2B}$.     \item $|\wh{G}(f)| \in [0,1]$ for $|f| \in \big( (1-\alpha/2)\frac{1}{2B}, \frac{1}{2B} \big)$.
    \item $|\wh{G}(f)| \le \delta/n$ for $|f| > \frac{1}{2B}$.     \item $|\wh{G}(f)| \le b_0 \Big(\frac{2 \alpha}{\pi B |f|}\Big)^{\ell}$ for $|f| \ge 1/2$.
\end{enumerate}

We set the parameters as follows: in $(G,\wh{G})$, $c<1$ is a small constant, $\alpha:=0.01$, $\delta := O(c\epsilon / n)$,
$B := O(n/(c\epsilon))$, $\ell := O(\log (n/\delta))$. Then hash parameters
$\sigma \sim \left[ \frac{c}{20BR},\frac{c}{10BR} \right]$ and
$b \sim [0, 1/(\sigma B)]$ are sampled uniformly at random. The parameter $a$ in \textsc{HashToBins} will be chosen in \textsc{GetSamples}. Finally, we describe how to choose parameter $\beta$ of \textsc{GetSamples} later (in Lemma~\ref{lem:estimate_mid_J}).

\begin{algorithm} [ht]
    \caption{Recover Heavy Intervals\label{alg:recover_cluster1}} 
    \begin{algorithmic}[1]
        \Procedure{HashToBins}{$\sigma,b,a$} 
        \State Compute $G(t)$ from \eqref{def:G} for $t \in [\frac{- \ell B}{2 \alpha},\frac{\ell B}{2 \alpha}] \cap \Z$ with parameters $\alpha=0.01$, $\delta = O(c\epsilon / n)$, $B = O(n/c \epsilon)$, $\ell = O(\log (n/\delta))$.
        \State Compute a vector in $\mathbb{C}^B$: $u[j]:=\underset{i \in \mathbb{Z}: |j+iB| \le \frac{\ell B}{2 \alpha}}{\sum}  w(\sigma(j+iB-a)) \cdot e^{2 \pi \bi b\sigma (j+iB)} \cdot G(j+iB)$
        \State Apply DFT to compute $\wh{u}[j]:=\underset{r \in [B]}{\sum} u[r] \cdot e^{-2 \pi \bi \frac{j \cdot r}{B}}$
        \State Return $\wh{u}[j]$
        \EndProcedure
        
        \Procedure{GetSamples}{$\sigma,b,\beta$}
            \Comment{This procedure replaces Procedure~\textsc{GetLegalKSample} in Algorithm 6 of \cite{CKPS17}}
            \State $m:=C_R \cdot R$ for a  large constant $C_R$
            \Comment{ Since $\sigma=\Theta(\frac{1}{BR})$, $m=\frac{O(1)}{B \sigma}$}
            \For{$i \in [m]$}
            \State Uniformly sample $t_i \in [-1.1,1.1]$ 
            \State Call Procedure~\textsc{HashToBins}($\sigma,b,a:=-t_i/\sigma$) to obtain $z^{(j)}(t_i)=\hat{u}[j]$            
            \State Call Procedure~\textsc{HashToBins}($\sigma,b,a:=-(t_i+\beta)/\sigma$) to obtain $z^{(j)}(t_i+\beta)=\hat{u}[j]$      \EndFor
            \For{$j \in [B]$}
                \State Set a distribution $D_j$ over $t_1,\ldots,t_m$ s.t. $D_j(t_i)$ is proportional to $|z^{(j)}(t_i)|^2$
                \State Sample $s_j \sim D_j$
            \EndFor
            \State Return $B$ complex numbers $\Big(\frac{z^{(0)}(s_0+\beta)}{z^{(0)}(s_0)},\ldots,\frac{z^{(B-1)}(s_{B-1}+\beta)}{z^{(B-1)}(s_{B-1})}\Big)$
                    \EndProcedure
    \end{algorithmic}
\end{algorithm}

\paragraph{HashToBins.} 
Procedure~\textsc{HashToBins} in Algorithm~\ref{alg:recover_cluster1} reduces the continuous Fourier transforms to the discrete case. Here are a few basic facts about this reduction.
\begin{fact}\label{fact:FFT_DFT}
    Let $\rect_s(t)$ denote the box function of width $s$: $\rect_s(t)=1/s$ for $|t|\leq s/2$ and $0$ otherwise; let $\sinc(sf)$ denote its Fourier transform $\frac{\sin(\pi s f)}{\pi s f}$.    
    
    Let $\comb_s(t):=\sum_{j \in \mathbb{Z}}\delta_{js}(t)$ for the Dirac-delta function $\delta$.
    
    Let $w:\mathbb{R} \to \mathbb{C}$ be a signal. Define $A:\mathbb{Z} \to \mathbb{C}$ as $A[i]:=w(i)$ and $B:[n] \to \mathbb{C}$ as $B[i]:=\sum_{j \in \mathbb{Z}} A[i+jn]$. Then
    \[
    DTFT~\wh{A}(f)=\sum_{i \in \mathbb{Z}} \wh{w}(f+i) \text{ for } f \in [0,1) \text{ and } DFT~\wh{B}[j]=\sum_{i \in \mathbb{Z}} \wh{w}(j/n+i) \text{ for } j \in [n].
    \]
\end{fact}

We state the main properties of $u[j]$ and $\wh{u}[j]$ computed in Procedure~\textsc{HashToBins} as follows. Since $G$'s support is $[-\frac{\ell B}{2 \alpha},\frac{\ell B}{2 \alpha}]$, it takes $O(\ell B+ B\log B)$ time to compute $u$ (in line 3) and its discrete Fourier transform $\wh{u}$ (in line 4). More importantly, vector $u$ is equivalent to
\begin{equation}\label{eq:def_u}
u[j]:=\sum_{i \in \Z} w(\sigma(j+iB-a)) \cdot e^{2 \pi \bi b\sigma (j+iB)} \cdot G(j+iB).
\end{equation}
This implies the following fact of Procedure~\textsc{HashToBins}~\cite{HIKP,CKPS17,SSWZ23}: given hash parameters $\sigma$ and $b$, all frequencies in $\wh{w}$ are mapped into $B$ buckets $\wh{u}[0],\ldots,\wh{u}[B-1]$ such that $\wh{u}[j]$ provides a way to estimate frequencies in the $j$th bucket of this hash. Specifically, this is achieved via $B$ filter functions $\wh{G}^{(j)}_{\sigma,b}$ for $j \in [B]$ that are based on $\wh{G}$. We provide a formal statement here, whose proof is deferred to Section~\ref{sec:hashtobins}. 

\begin{lemma}\label{lem:property_hash_bins}
     Given $j \in [B]$, $\sigma>0$, and $b$, let $\wh{G}^{(j)}_{\sigma,b}(f):=\sum_{i \in \Z} \wh{G}(i + j/B -\sigma f - \sigma b)$ such that $G^{(j)}_{\sigma,b}(t)= G(-t/\sigma)\cdot e^{2 \pi \bi (j/B-\sigma b)\cdot t/\sigma} \cdot \comb_{\sigma}(t)$.      
     Let $z^{(0)},\ldots,z^{(B-1)}$ denote $B$ signals defined as: $\wh{z^{(j)}}:=\wh{w} \cdot \wh{G}^{(j)}_{\sigma, b}$ and $z^{(j)}:=w*G^{(j)}_{\sigma,b}$. $u \in \C^B$ defined in \eqref{eq:def_u} satisfies $\wh{u}[j]=z^{(j)}(-\sigma a)$.
\end{lemma}

\paragraph{Procedure~\textsc{GetSamples}.} The main technical lemma of this section proves that \textsc{GetSamples} estimates $e^{2 \pi \bi \cdot mid(J_j) \beta}$ of each well-hashed interval $J_j$ efficiently. This provides an efficient search algorithm to estimate $mid(J_j)$ from previous works \cite{CKPS17,CP19_ICALP,SSWZ23}.

Here are some notations in its analysis. Let $round(x)$ denote the nearest integer of $x$ (break the ties arbitrarily when $x \in 0.5 + \mathbb{Z}$). After fixing $\sigma$ and $b$, for a frequency $f$, let 
\[
h_{\sigma,b}(f):=round\bigg( \big( (\sigma f + \sigma b) \mod 1 \big) \cdot B \bigg) \mod B
\] denote the bucket in $[B]$ containing $f$. For any frequency interval $J:=[f_1,f_2]$, let 
\[
h_{\sigma,b}(J):=h_{\sigma,b}(mid(J))=h_{\sigma,b}(\frac{f_1+f_2}{2})
\]
be its bucket in this hash. For each heavy interval $J_j$ in $\wh{x}$, let $h_j:=h_{\sigma,b}(J_j)$ and $i_j$ denote the closest integer to $\sigma \cdot mid(J_j) + \sigma b - \frac{h_j}{B}$ such that $i_j+h_j/B-\sigma b - \sigma \cdot mid(J_j) \in [-1/2B,1/2B]$ --- the main energy part of $\wh{G}$.
we call $J_j$ well-hashed given $\sigma$ and $b$ only if $\sigma$ and $b$ satisfy the following two conditions for some small constant $c<1$:
\begin{align}
    i_j + h_j/B - \sigma b - \sigma J_j \subset [-(1-\alpha)/2B,(1-\alpha)/2B] \label{eq:good_condition_1} \\
    \int^{\frac{i_j+h_j/B+1/2B-\sigma b}{\sigma}}_{\frac{i_j+h_j/B-1/2B-\sigma b}{\sigma}} |\wh{z^{(h_j)}}(f)|^2 \mathrm{d} f \ge (1-c) \int |\wh{z^{(h_j)}}(f)|^2 \mathrm{d} f  \label{eq:good_condition_2} 
\end{align}

\begin{lemma}\label{lem:estimate_mid_J}
    For some fixed small constant $c$, let $B=O(\frac{n}{c \cdot \epsilon})$, $\sigma \sim [\frac{c}{20 B \cdot R},\frac{c}{10 B \cdot R}]$, and $b \sim [0,\frac{1}{B \sigma}]$. Then each heavy interval $J_j$ in $\wh{w}$ is well-hashed (satisfying \eqref{eq:good_condition_1} and \eqref{eq:good_condition_2}) with probability $0.9$ over $\sigma$ and $b$.

    Conditioned on that $J_j$ is well-hashed, for any $0 < \beta \le c_0 \cdot B \sigma$ (with some universal constant $c_0<1$), 
    \[
    \left|\frac{z^{(h_j)}(s_{h_j}+\beta)}{z^{(h_j)}(s_{h_j})} - e^{2 \pi \bi \cdot mid(J_j) \beta}\right| \le 0.3 \text{ with probability at least $0.6$ (over $t_1,\ldots,t_m$ and $s_{h_j} \sim D_{h_j}$).}\]
\end{lemma}

Finally, we defer the proof of Lemma~\ref{lem:estimate_mid_J} to Section~\ref{sec:proof_lem_estimate_J} and finish the proof of Theorem~\ref{thm:locations_multiband} here.

\begin{proofof}{Theorem~\ref{thm:locations_multiband}}
    Plugging Lemma~\ref{lem:estimate_mid_J} into previous algorithms \cite{CKPS17,SSWZ23} provides an efficient algorithm to estimate $J_j$ within error $O(R)$. Specifically, one could use the framework of \cite{CKPS17} but replace line 17 of Procedure~\textsc{LocateKInner} in Algorithm 6 calling \textsc{GetLegalKSample} (to obtain $\hat{u}[j]/\hat{u'}[j]$) by \textsc{GetSamples} (or use Procedure~\textsc{ArySearch} in Algorithm 3 of \cite{SSWZ23} but replace line 13 by \textsc{GetSamples}). 
    
    For completeness, we show an outline of the algorithm and its analysis in this proof: after sampling $\sigma$ and $b$, the algorithm calls Procedure~\textsc{GetSamples} with various parameters of $\beta$ to search well-hashed $J_j$ satisfying \eqref{eq:good_condition_1} and \eqref{eq:good_condition_2}. Specifically, let $L=[-F,F]$     be the current search range of $J_j$ that will be cut by a constant factor in every round. Every round will set $\beta:=\frac{1}{2\cdot |L|}$ for the current range $L$ to guarantee that $e^{2 \pi \bi f \beta}$ maps $L$ to half of the unit circle $S^1$ in the complex plane. Now Lemma~\ref{lem:estimate_mid_J} shows \[
    \left|\frac{z^{(h_j)}(s_{h_j}+\beta)}{z^{(h_j)}(s_{h_j})} - e^{2 \pi \bi \cdot mid(J_j) \beta}\right| \le 0.3 \text{ with probability at least $0.6$.}
    \]
    
    This provides an estimation of $mid(J_j) \cdot \beta$ within error $\frac{0.4}{\pi} \cdot |L|$ on the unit circle $S^1$. Moreover, one could repeat this sampling $O(\log (B \log F))$ rounds and take the median to amplify the success probability to $1-(B \log F)^{-\Omega(1)}$. After amplification, via the estimation of $mid(J_j) \cdot \beta$ on the unit circle $S^1$, the algorithm reduces the length of $L$ by a factor $0.8/\pi$ in each round until Lemma~\ref{lem:estimate_mid_J} stops at $\beta > c_0 \cdot B \sigma$. This stop condition implies the final error of estimating $mid(J_j)$ is $O(\frac{1}{c_0 \cdot B \sigma})=O(R)$.

    To make sure that we will recover all heavy intervals, we repeat the above procedure $O(\log n)$     times and take the union bound because each time it recovers one heavy interval $J_j$ with probability 0.9 (over $\sigma$ and $b$).
\end{proofof}

\subsection{Proof of Lemma~\ref{lem:estimate_mid_J}}\label{sec:proof_lem_estimate_J}
We finish the proof of Lemma~\ref{lem:estimate_mid_J}  in the rest of this section. Recall that Procedure \textsc{HashToBins} provides $\wh{u}[j]=z^{(j)}(-\sigma a)$ for $\wh{z^{(j)}}:=\wh{w} \cdot \wh{G}^{(j)}_{\sigma,b}$.

We bound the probability that $\sigma$ and $b$ satisfy \eqref{eq:good_condition_1} and  \eqref{eq:good_condition_2} separately.

\begin{claim}\label{clm:condition1_holds}
\eqref{eq:good_condition_1} holds with probability $1-\alpha - \frac{c}{5}$ over $\sigma$ and $b$.
\end{claim}
\begin{proof}
    Consider the median $f$ of this interval $J_j$. Since $\sigma b \sim [0,\frac{1}{B}]$ by the definition $b \sim [0,\frac{1}{B \sigma}]$, $i_j + \frac{h_j}{B} - \sigma b - \sigma mid(J_j)$ is a uniformly random variable in $[-\frac{1}{2B},\frac{1}{2B}]$. Condition~\eqref{eq:good_condition_1} holds iff 
    \[
    \Big| i_j + \frac{h_j}{B} - \sigma b - \sigma mid(J_j) \Big| \le \frac{1-\alpha}{2B} - \frac{\sigma \cdot |J_j|}{2}.\] 
    Since $\sigma \sim [\frac{c}{20 B \cdot R},\frac{c}{10 B \cdot R}]$,
\begin{equation}\label{eq:prob_condition_1}
    \Pr\left[ \Big| i_j + \frac{h_j}{B} - \sigma b - \sigma mid(J_j) \Big| \le \frac{1-\alpha}{2B} - \frac{\sigma \cdot |J_j|}{2} \right] \ge 1 - \alpha - \frac{\sigma |J_j|}{1/2B} = 1 - \alpha - \frac{c}{5}.
\end{equation}

\end{proof}

\begin{claim}\label{clm:probability_condition_freq}    
    \eqref{eq:good_condition_2} holds with probability $0.95$ over $\sigma$ and $b$.
\end{claim}
The intuition behind Claim~\ref{clm:probability_condition_freq} is that $\wh G^{(j)}_{\sigma,b}$ is almost 1 on only $1/B$ fraction otherwise it is almost 0. 
We defer the proof to Section~\ref{sec:proof_cond_freq} and first establish a support bound for every $z^{(j)}$.

\begin{claim}\label{clm:outside_time_window}
    For any $t \notin [-1-\frac{\ell}{\alpha R},1+\frac{\ell}{\alpha R}]$, $z^{(j)}(t)=0$.
\end{claim}
\begin{proof}
    For any $t$, let us consider \[
    z^{(j)}(t)=w*G^{(j)}_{\sigma,b} = \int_{\supp(G^{(j)}_{\sigma,b})} w(t-\tau) \cdot G^{(j)}_{\sigma,b}(\tau) \mathrm{d} \tau.\]
    Because $G^{(j)}_{\sigma,b}(t):=G(-t/\sigma) \cdot e^{2 \pi \bi (j/B-\sigma b)\cdot t/\sigma} \cdot \comb_\sigma(t)$ from Lemma~\ref{lem:property_hash_bins}, $\supp(G^{(j)}_{\sigma,b})=\sigma \cdot (\supp(G) \cap \mathbb{Z})$. Moreover, $\supp(G)=[\frac{\ell}{2} \cdot \frac{-B}{\alpha},\frac{\ell}{2} \cdot \frac{B}{\alpha}]$ from the construction of $G$ implies $\supp(G^{(j)}_{\sigma,b})=\sigma \cdot ([-\frac{\ell B}{2 \alpha},  \frac{\ell B}{2 \alpha}] \cap \mathbb{Z})$, which is in $[-\frac{\ell}{20 \alpha R}, \frac{\ell}{20 \alpha R}]$ given $\sigma \le \frac{c}{10 B \cdot R}$. 

    On the other hand, because $w(t)=0$ for any $t \notin [-1,1]$ by the definition, $w(t-\tau) \cdot G^{(j)}_{\sigma,b}(\tau)$ is not zero only if $t-\tau \in [-1,1]$ and $\tau \in [-\frac{\ell}{20 \alpha R}, \frac{\ell}{20 \alpha R}]$.

        \end{proof} 

Claim~\ref{clm:condition1_holds} and Claim~\ref{clm:probability_condition_freq} imply the first part of Lemma~\ref{lem:estimate_mid_J} about Conditions \eqref{eq:good_condition_1} and \eqref{eq:good_condition_2}.
Then we analyze $\frac{z^{(h_j)}(s_{h_j}+\beta)}{z^{(h_j)}(s_{h_j})}$ and $e^{2 \pi \bi \cdot mid(J_j) \beta}$ conditioned on these two properties for the second part of Lemma~\ref{lem:estimate_mid_J}.

For convenience, we decompose $z^{(h_j)}$ into two parts: 
\begin{align}
\wh{h}(f) & :=\wh{z^{(h_j)}}(f) \text{ for } f \in \left[\frac{i_j+h_j/B-1/2B-\sigma b}{\sigma},\frac{i_j+h_j/B+1/2B-\sigma b}{\sigma}\right] \text{ and 0 otherwise.} \\
\wh{e}(f) & = \wh{z^{(h_j)}}(f) \text{ for } f \notin \left[\frac{i_j+h_j/B-1/2B-\sigma b}{\sigma},\frac{i_j+h_j/B+1/2B-\sigma b}{\sigma}\right] \text{ and 0 otherwise.}
\end{align}
So $z^{(h_j)}(t)=h(t)+e(t)$ by the definition. A useful property of $h$ is that 
\begin{align}
\sup_t |h(t)|^2 & \le \left(\int_{f \in \supp(\wh{h})} |\wh{h}(f)| \mathrm{d} f\right)^2 \notag \\    
& \le \left(\int_{f \in \supp(\wh{h})} |\wh{h}(f)|^2 \mathrm{d} f\right) \cdot \left(\int_{f \in \supp(\wh{h})} 1 \mathrm{d} f\right) \notag \\
& \le \|h\|_2^2 \cdot \frac{1/B}{\sigma} \le \frac{20R}{c} \cdot \|h\|_2^2. \label{eq:upper_bound_sup_h}
\end{align}

Moreover, \eqref{eq:good_condition_2} implies $\|e\|_2 \le \sqrt{c} \cdot \|z^{(h_j)}\|_2$. Because 
\[
\|z^{(h_j)}\|_2 =\|z^{(h_j)}\|_{[-1.1,1.1]}=\|h+e\|_{[-1.1,1.1]},\]
this implies 
\[
\|h\|_{[-1.1,1.1]} \ge \|z^{(h_j)}\|_2 - \|e\|_{[-1.1,1.1]} \ge \|z^{(h_j)}\|_2 - \|e\|_2 \ge (1-\sqrt{c}) \|z^{(h_j)}\|_2.
\]
This is at least $(1-\sqrt{c}) \cdot \|h\|_2$ by the definition of $\wh{h}$.

\begin{claim}\label{clm:guarantees_samples1}
    With probability $0.89$ over $t_1,\ldots,t_m$, $\sum_{i} |z^{(h_j)}(t_i)^2|/m \ge \frac{0.9}{2.2} \cdot \|z^{(h_j)}\|_2^2$ when $c$ is a sufficiently small constant.
\end{claim}
\begin{proof}
    For each random sample $t_i \sim [-1.1,1.1]$ in Procedure~\textsc{GetSamples}, $\E[|h(t_i)|^2]=\|h\|_{[-1.1,1.1]}^2/2.2 \ge \frac{(1-\sqrt{c})^2}{2.2} \cdot \|h\|^2_2$. So \eqref{eq:upper_bound_sup_h} implies $\frac{\sup_t |h(t)|^2}{\E[|h(t_i)|^2]} \le \frac{10R \cdot 2.2}{c \cdot (1-\sqrt{c})^2}=O(R/c)$. Since the number of samples $m:=C_R \cdot R$ for a large constant $C_R$, the Chernoff bound implies that 
\[
\text{with probability }0.99, \frac{\sum_{i=1}^m |h(t_i)|^2}{m} \ge 0.95 \cdot \E[|h(t_i)|^2] \ge \frac{0.95 \cdot (1-\sqrt{c})^2}{2.2} \cdot \|z^{(h_j)}\|_2^2.
\] 
At the same time, \[
\E\left[ \frac{\sum_{i=1}^m |e(t_i)|^2}{m} \right]=\E[|e(t_i)|^2]=\|e\|_{[-1.1,1.1]}^2/2.2 \le \|e\|_2^2/2.2.\]
Markov's inequality implies that with probability $0.9$, $\frac{\sum_{i=1}^m |e(t_i)|^2}{m} \le 10 \cdot \|e\|_2^2/2.2 \le 10 c \cdot \|z^{(h_j)}\|_2^2/2.2$.
\begin{align*}
    \sum_i |z^{(h_j)}(t_i)|^2 & \ge \sum_i \left(|h(t_i)|^2 - 2 |h(t_i)| \cdot |e(t_i)| + |e(t_i)|^2 \right) \\
    & \ge \sum_i |h(t_i)|^2 - 2 (\sum_i |h(t_i)|^2)^{1/2} \cdot (\sum_i |e(t_i)|^2)^{1/2} \\
    & \ge m \left( \frac{0.95 \cdot (1-\sqrt{c})^2}{2.2} \cdot \|z^{(h_j)}\|_2^2 - \frac{\sqrt{0.95} \cdot (1-\sqrt{c}) \cdot \sqrt{10 c}}{2.2} \cdot \|z^{(h_j)}\|_2^2\right)\\
    & \ge m \cdot \frac{0.9}{2.2} \cdot \|z^{(h_j)}\|_2^2.
\end{align*}
\end{proof}

\begin{claim}\label{clm:diff_beta}
    For any $\beta \le c_0 \cdot B \sigma$, $g(t):=z^{(h_j)}(t) \cdot e^{2 \pi \bi \cdot mid(J_j) \beta} - z^{(h_j)}(t+\beta)$ satisfies $\|g\|_{2}^2 \le O(c_0^2 + c) \cdot \|z^{(h_j)}\|_2^2$ and $\E_{t_i}[|g(t_i)|^2] = O(c_0^2 + c) \cdot \|z^{(h_j)}\|_2^2$.

\end{claim}
\begin{proof}
    Let us consider 
    \begin{equation}\label{eq:small_error}
    \wh{g}(f)=\wh{z^{(h_j)}}(f) \cdot e^{2 \pi \bi \cdot mid(J_j) \beta} - \wh{z^{(h_j)}}(f) \cdot e^{2 \pi \bi \cdot f \beta} = \wh{z^{(h_j)}}(f) \cdot (e^{2 \pi \bi \cdot mid(J_j) \beta} - e^{2 \pi \bi \cdot f \beta}).        
    \end{equation}
    Condition~\eqref{eq:good_condition_1} implies $mid(J_j) \in [\frac{i_j+h_j/B-1/2B-\sigma b}{\sigma},\frac{i_j+h_j/B+1/2B-\sigma b}{\sigma}]$ such that 
    \[
    \forall f \in [\frac{i_j+h_j/B-1/2B-\sigma b}{\sigma},\frac{i_j+h_j/B+1/2B-\sigma b}{\sigma}], |mid(J_j) \beta - f \beta| \le \frac{1/B}{\sigma} \cdot \beta \le c_0.
    \]
    This implies $|e^{2 \pi \bi \cdot mid(J_j) \beta} - e^{2 \pi \bi \cdot f \beta}| \le 4 \pi \cdot c_0$ and
 \eqref{eq:small_error} is at most   
    $\le |\wh{z^{(h_j)}}(f)| \cdot 4\pi c_0$.

    For $f$ not in this interval, we bound \eqref{eq:small_error} by $2 \cdot |\wh{z^{(h_j)}}(f)|$.    From all discussion above, 
    \begin{align*}
    \|g\|_2^2 & = \int_{f \in [\frac{i_j+h_j/B-1/2B-\sigma b}{\sigma},\frac{i_j+h_j/B+1/2B-\sigma b}{\sigma}]} |\wh{g}(f)|^2 \mathrm{d} f + \int_{f \notin [\frac{i_j+h_j/B-1/2B-\sigma b}{\sigma},\frac{i_j+h_j/B+1/2B-\sigma b}{\sigma}]} |\wh{g}(f)|^2 \mathrm{d} f \\
        & \le (4\pi c_0)^2 \cdot \|\wh{z^{(h_j)}}\|_2^2 + 4c \cdot \|\wh{z^{(h_j)}}\|_2^2 \\
        & = ((4\pi c_0)^2 + 4c) \cdot \|z^{(h_j)}\|_2^2.
    \end{align*}
    
Since $\|g\|_{[-1,1]}^2 \le \|g\|_2^2$ and $\E[|g(t_i)|^2]=\frac{\|g\|_{[-1.1,1.1]}^2}{2.2}$, $\E[|g(t_i)|^2] \le O(c_0^2 + c) \cdot \frac{\|z^{(h_j)}\|_{2}^2}{2.2}$.
    
\end{proof}

Back to the proof of Lemma~\ref{lem:estimate_mid_J}, we assume Claim~\ref{clm:guarantees_samples1} and Claim~\ref{clm:diff_beta} to bound $\frac{z^{(h_j)}(s_{h_j}+\beta)}{z^{(h_j)}(s_{h_j})} - e^{2 \pi \bi \cdot mid(J_j) \beta}$. For $s_j \sim D_{h_j}$, we bound
\begin{align}
\E_{t_1,\ldots,t_m} \left[ \E_{s_j \sim D_j}\left[ \frac{|z^{(h_j)}(s_j)e^{2 \pi \bi mid(J_j)\beta}-z^{(h_j)}(s_j+\beta)|^2}{|z^{(h_j)}(s_j)|^2} \right] \right] & =\E_{t_1,\ldots,t_m} \left[ \E_{s_j} \frac{|g(s_j)|^2}{|z^{(h_j)}(s_j)|^2} \right] \notag \\
& = \E_{t_1,\ldots,t_m} \left[ \sum_{i=1}^m \frac{|z^{(h_j)}(t_i)|^2}{\sum_{i'} |z^{(h_j)}(t_{i'})|^2} \cdot \frac{|g(t_i)|^2}{|z^{(h_j)}(t_i)|^2} \right] \notag \\
& \le \E_{t_1,\ldots,t_m} \left[ \frac{\sum_{i=1}^m |g(t_i)|^2}{\sum_{i'} |z^{(h_j)}(t_{i'})|^2} \right]. \label{eq:exp_z}
\end{align}

Assuming Claim~\ref{clm:guarantees_samples1} that $\sum_{i'} |z^{(h_j)}(t_{i'})|^2 \ge \frac{0.9 m}{2.2} \|z^{(h_j)}\|_2^2$ holds (with probability 0.89), Claim~\ref{clm:diff_beta} implies that this expectation is at most 
\[
\frac{m\E_{t_i}[|g(t_i)|^2]/0.89}{0.9 m\cdot \|z^{(h_j)}\|_2^2/2.2} = \frac{O(c_0^2+c) \cdot \|z^{(h_j)}\|_2^2/0.89}{0.9 \cdot \|z^{(h_j)}\|_2^2/2.2} = O(c_0^2+c).
\] 
Markov's inequality implies it is at most $0.09$ with probability $4/5$ when $c_0$ and $c$ are small constants. Hence with probability $\ge 0.6$ (assuming Claim~\ref{clm:guarantees_samples1} and the Markov's inequality over \eqref{eq:exp_z}), 
\[
\frac{|z^{(h_j)}(s_j)e^{2 \pi \bi mid(J)\beta}-z^{(h_j)}(s_j+\beta)|}{|z^{(h_j)}(s_j)|}\le 0.3.
\]

\subsection{Proof of Claim~\ref{clm:probability_condition_freq}}\label{sec:proof_cond_freq}

Recall that \[
\wh{G}^{(j)}_{\sigma,b}(f) = \sum_{q\in\Z}\wh{G}(q+j/B-\sigma f-\sigma b).
\]
Then
\[
G^{(j)}_{\sigma,b}(t) = G(-t/\sigma) e^{2\pi\bi(j/B-\sigma b)t/\sigma}\comb_\sigma(t).\label{eq:Gj_comb}
\]
\begin{lemma} \label{lem:properties_hatG_periodized}
Given $j\in[B]$, $\sigma>0$, and $b\in\R$. The filter function $(\wh{G}^{(j)}_{\sigma,b}, G^{(j)}_{\sigma,b})$ has the following properties: 
\begin{enumerate}
    \item If $\min_{q\in\Z}|q+j/B-\sigma f-\sigma b|\leq\frac{1-\alpha}{2B}$, $\wh{G}^{(j)}_{\sigma,b}(f) \in[1-3\delta/n,1+3\delta/n]$.
    \item If $\min_{q\in\Z}|q+j/B-\sigma f-\sigma b| > \frac{1}{2B}$, $0\leq\wh{G}^{(j)}_{\sigma,b}(f)\leq 3\delta/n$.
    \item $\supp(G^{(j)}_{\sigma,b}) \subseteq \sigma \cdot \left(\Z\cap\left[-\frac{\ell B}{2\alpha},\frac{\ell B}{2\alpha}\right]\right)$.
\end{enumerate}
\end{lemma}

\begin{proof}
First we prove an auxiliary estimation for $\sum_{q\in\Z\setminus\{0\}} |\wh{G}(f+q)|$ when $f \in [-1/2,1/2]$.

We use Property 5 of $(G,\wh{G})$ that 
\begin{equation}
    |\wh{G}(x)| \leq b_0\left(\frac{2\alpha}{\pi B|x|}\right)^{\ell}
\end{equation}
for $|x|\geq 1/2$. 
For $f\in[-1/2,1/2]$ and $q \in \mathbb{Z} \setminus 0$, we have $|f+q|\geq |q|/2$. 
Therefore,
\[
\sum_{q\in\Z\setminus\{0\}}|\wh{G}(f+q)| \leq 2b_0\left(\frac{4\alpha}{\pi B}\right)^{\ell} \sum_{q=1}^{\infty}q^{-\ell} \leq \delta/n.
\]

Consider the integer $q_0 = \arg\min_{q\in\Z}|q+j/B-\sigma f-\sigma b|$ for arbitrary $f$. Since $q+j/B-\sigma f-\sigma b \notin (-1/2,1/2)$ for $q \ne q_0$, the terms in $\sum_{q\in\Z}\wh{G}(q+j/B-\sigma f-\sigma b)$ beside $q_0$ can only contribute $(\delta/n)^{O(1)}$ in total by the above estimation. Hence, $|\wh{G}^{(j)} - \wh{G}(q_0+j/B-\sigma f-\sigma b)| \leq \delta / n$. The first two properties follow by the properties of $\wh{G}$. 

And we have proved the third property in Claim~\ref{clm:outside_time_window}.
\end{proof}

\begin{proof}
Recall that $h_j := h_{\sigma,b}(J_j)$, $i_j$ denotes the closest integer such that $i_j+h_j/B-\sigma b - \sigma \cdot mid(J_j) \in [-1/2B,1/2B]$, and $\wh{G}^{(j)}_{\sigma,b}(f) := \sum_{i \in \Z} \wh{G}(i + j/B -\sigma f - \sigma b)$.

For convenience, let $I_{j}:=\left[ \frac{i_j+h_j/B-1/(2B)-\sigma b}{\sigma}, \frac{i_j+h_j/B+1/(2B)-\sigma b}{\sigma} \right]$. And we define
\begin{align*}
E_{in} &:=\int_{I_j}|\wh w(f) \cdot \wh{G}^{(h_j)}_{\sigma,b}(f) |^2\mathrm{d}f,\\
E_{out} &:=\int_{\R\setminus I_j}|\wh w(f) \cdot \wh{G}^{(h_j)}_{\sigma,b}(f)  |^2\mathrm{d}f.
\end{align*}
By Lemma~\ref{lem:property_hash_bins}, we have $E_{in}+E_{out}=\|z^{(h_j)}\|_2^2$. 
And it suffices to prove that $E_{in}\geq(1-c)(E_{in}+E_{out})$.

 For $f \notin I_j$, $\wh{G}^{(h_j)}_{\sigma,b}(f) \le 3\delta/n$ if there does not exist an integer $m$ such that $\left|m+h_j/B-\sigma b-\sigma f\right|\leq\frac1{2B}$.
The definition of $i_j$ also implies $\left|i_j+h_j/B-\sigma b-\sigma \cdot mid(J_j)\right|\leq\frac1{2B}$.
Subtracting the last two inequalities yields that for any $f \notin I_j$  with $h_{\sigma,b}(f) = h_j$,
\[
\left|\sigma(f-mid(J_j))-(m-i_j)\right|\leq \frac{1}{B}.
\]
Since $f \notin I_j$ and so $m-i_j$ is a nonzero integer,
\begin{equation} \label{eq:collision-freqency-distance-condition}
    \sigma|f-mid(J_j)| \geq 1 - \frac{1}{B}.
\end{equation}

We consider the random variable $X := \sigma |f-mid(J_j)|$ for a fixed $f \notin I_j$. 
Since $\sigma \sim [\frac{c}{20 B \cdot R},\frac{c}{10 B \cdot R}]$, $X$ is uniform on an interval $[A, 2A]$ where $A := \frac{c |f-mid(J_j)|}{20 B \cdot R}$. 
And we have $2 A \geq \sigma|f-mid(J_j)| \geq 1 - \frac{1}{B}$ by \eqref{eq:collision-freqency-distance-condition}.
The interval $[A,2A]$ contains at most $A+3$ integer.  
Hence,
\begin{align*}
\Pr[h_{\sigma,b}(f) = h_j]
\le & \Pr\Big[\min_{m \in \Z} \left|\sigma(f-mid(J_j))-(m-i_j)\right|\leq\frac{1}{B} \Big] \\
\le & \Pr\Big[\min_{m' \in \Z} \left|X - m'\right|\leq\frac{1}{B} \Big] \\
\le & (A+3) \cdot \frac{2/B}{A}
= O\Big(\frac{1}{B}\Big).
\end{align*}

Therefore, for every $f \notin I_j$,
\begin{align*}
    \E_{\sigma,b} \left[ |\wh{G}^{(h_j)}(f)|^2\right] = & \E_{\sigma,b} \left[ |\wh{G}^{(h_j)}(f)|^2 \Big| h_{\sigma,b}(f) \ne h_j \right] \Pr[h_{\sigma,b}(f) \ne h_j] \\
    & + \E_{\sigma,b} \left[ |\wh{G}^{(h_j)}(f)|^2 \Big| h_{\sigma,b}(f) = h_j \right] \Pr[h_{\sigma,b}(f) = h_j].
\end{align*}
If $h_{\sigma,b}(f) \ne h_j$, that infers $\left|i_j+h_j/B-\sigma b-\sigma f\right| \geq \frac{1}{2B}$, then $|\wh{G}^{(h_j)}(f)|^2 \le (3\delta/n)^2$ by Lemma~\ref{lem:properties_hatG_periodized}. Meanwhile, Lemma~\ref{lem:properties_hatG_periodized} also shows $|\wh{G}^{(h_j)}(f)|^2 \le (1+3\delta/n)^2$ for any $f$. Combining the two bounds, we have
\begin{align*}
    \E_{\sigma,b} \left[ \mathbf1_{\{f\notin I_j\}} |\wh{G}^{(h_j)}(f)|^2\right] & \le (3\delta/n)^2 \cdot \Pr[h_{\sigma,b}(f) \ne h_j] + O(1) \cdot \Pr[h_{\sigma,b}(f) = h_j] \\
    & = (3\delta/n)^2 + O(1/B) = O(1/B).
\end{align*}
And it shows that
\begin{align*}
\E_{\sigma,b}[E_{out}]
=\int_{\R}|\wh w(f)|^2 \cdot \E_{\sigma,b} \Big[ \mathbf1_{\{f\notin I_j\}}  |\wh{G}^{(h_j)}_{\sigma,b}(f) |^2 \Big] \mathrm{d}f
=O\Big(\frac{1}{B}\Big) \cdot \|w\|_2^2.
\end{align*}
By Markov's inequality, with a proper choice $B:=O(\frac{n}{\epsilon c})$,
\begin{align*}
 \Pr \left[E_{out} > \frac{\epsilon c}{8n}\|w\|_2^2\right] \le 0.01.
\end{align*}

Meanwhile, since $J_j$ is heavy, when Condition~\eqref{eq:good_condition_1} holds,
\begin{align*}
E_{in} & \geq\int_{J_j}|\wh w(f)|^2 |\wh{G}^{(h_j)}_{\sigma,b}(f)|^2 \mathrm{d}f \\
& \geq(1-3\delta/n)^2\int_{J_j}|\wh w(f)|^2 \mathrm{d}f \geq\frac{\epsilon}{8n}\|w\|_2^2.
\end{align*}

Finally, a union bound implies
\[
\Pr[\text{Condition~\eqref{eq:good_condition_2} holds}] \ge 0.95.
\]
\end{proof}

\subsection{Proof of Lemma~\ref{lem:property_hash_bins}}\label{sec:hashtobins}
For convenience, let $s(t):=w(t) \cdot e^{2 \pi \bi b (t + \sigma a)} \cdot G(\frac{t + \sigma a}{\sigma})$
such that for $j \in [B]$, $u[j]:=\sum_{i \in \Z} s(\sigma(j+iB-a))$ by the definition~\eqref{eq:def_u}.

First, we verify the properties of $(\wh{G}^{(j)}_{\sigma,b},G^{(j)}_{\sigma,b})$.

\begin{claim}
Recall $\wh{G}^{(j)}_{\sigma,b}(f):=\sum_{i \in \Z} \wh{G}(i + j/B -\sigma f - \sigma b)$. Therefore, 
\[G^{(j)}_{\sigma,b}(t)= G(-t/\sigma)\cdot e^{2 \pi \bi (j/B-\sigma b)\cdot t/\sigma} \cdot \comb_{\sigma}(t). \]
\end{claim}
\begin{proof}
\begin{align*}
G^{(j)}_{\sigma,b}(t)
&=\int_{-\infty}^{+\infty} \wh{G}^{(j)}_{\sigma,b}(f)e^{2\pi \mathbf ift}\mathrm{d}f\\
&=\int_{-\infty}^{+\infty} \sum\limits_{i\in \mathbb Z}\wh{G}(i+j/B-\sigma f-\sigma b)e^{2\pi\mathbf ift}\mathrm{d}f.
\end{align*}

Define $u:=i+j/B-\sigma f-\sigma b$ such that:

\begin{equation}
\label{eq: substitute f for u}
f=\frac{i+j/B-\sigma b-u}{\sigma},\quad \mathrm{d}f=-\frac{\mathrm{d}u}{\sigma}.
\end{equation}

By substituting $f$ for $u$, we have:

\begin{align}
G^{(j)}_{\sigma,b}(t)
&=\int_{+\infty}^{-\infty} \sum\limits_{i\in \mathbb Z}\wh{G}(u)\exp(2\pi\mathbf i\frac{i+j/B-\sigma b-u}{\sigma}t)(-\frac{1}{\sigma})\mathrm{d} u\nonumber\\
&=\frac{1}{\sigma}\int_{-\infty}^{+\infty}\sum\limits_{i\in \mathbb Z}\wh{G}(u)\exp(2\pi\mathbf i t\frac{j/B-\sigma b}{\sigma})\exp(-2\pi\mathbf i t\frac{u}{\sigma})\exp(2\pi\mathbf i t\frac{i}{\sigma})\mathrm{d}u\nonumber\\
&=\frac{1}{\sigma}\exp\left(2\pi\mathbf i t\frac{j/B-\sigma b}{\sigma}\right)\left(\int_{-\infty}^{+\infty}\wh{G}(u)\exp(-2\pi\mathbf i t\frac{u}{\sigma})\mathrm{d}u\right)\left(\sum\limits_{i\in \mathbb Z}\exp(2\pi\mathbf i t\frac{i}{\sigma})\right)\nonumber\\
&=\frac{1}{\sigma}\exp\left(2\pi\mathbf i t\frac{j/B-\sigma b}{\sigma}\right)G(-t/\sigma)\sum\limits_{i\in \mathbb Z}\exp(2\pi\mathbf i t\frac{i}{\sigma}).\label{eq: expansion form of G}
\end{align}

By Fact~\ref{fact:FFT_DFT}, $\comb_\sigma(t)$ is a periodic function. We expand $\comb_\sigma(t)$ into Fourier series:

\begin{equation}
\label{eq: Fourier series of comb}
\comb_\sigma(t) = \sum_{n=-\infty}^{\infty} c_n \, e^{2\pi\mathbf i n \frac{t}{\sigma}}
\end{equation}

where 

\begin{equation}
\label{eq: coefficient of Fourier series of comb}
c_n = \frac{1}{\sigma} \int_{-\sigma/2}^{\sigma/2} \comb_\sigma(t) e^{-2\pi\mathbf i n \frac{t}{\sigma}}\mathrm{d}t=\frac{1}{\sigma}\int_{-\sigma/2}^{\sigma/2}\delta(t)e^{-2\pi\mathbf i n \frac{t}{\sigma}}\mathrm{d}t=\frac{1}{\sigma}e^{-2\pi \mathbf in\cdot 0}=\frac{1}{\sigma}
.
\end{equation}

Combine \eqref{eq: Fourier series of comb} and \eqref{eq: coefficient of Fourier series of comb}. We have:

\begin{equation}
\sum_{i\in\mathbb Z} \exp(2\pi\mathbf i \frac{i}{\sigma}t) =
\sigma\sum\limits_{n=-\infty}^{+\infty}\frac{1}{\sigma}\exp(2\pi\mathbf in\frac{t}{\sigma})
= \sigma\,\comb_\sigma(t).
\label{eq: connection between comb and exp}
\end{equation}

Plug \eqref{eq: connection between comb and exp} into \eqref{eq: expansion form of G}. We have:

\[
G^{(j)}_{\sigma,b}(t) = G(-t/\sigma)\, e^{2\pi\mathbf i (j/B - \sigma b)\cdot t/\sigma}\, \comb_\sigma(t).
\]
\end{proof}

We calculate $\wh{u}[j]$ in the following claim.

\begin{claim}

Recall in \eqref{eq:def_u}, $u[j]:=\sum_{i \in \Z} s(\sigma(j+iB-a))$. Therefore, 

\begin{equation}\label{eq:Fourier_coef_u}
\wh{u}[j]=\frac{1}{\sigma} \sum_{i \in \Z} \wh{s}(\frac{j/B + i}{\sigma}) \cdot e^{-2\pi\bi (j/B+i) \cdot a}.
\end{equation}
\end{claim}

\begin{proof}

By definition of $u[j]$,

\begin{align}
u[j]&=\sum_{i \in \mathbb Z} s(\sigma(j+iB-a))\nonumber\\
&=\sum\limits_{i\in\mathbb Z}\int_{-\infty}^{+\infty}\hat{s}(f)e^{2\pi \mathbf if\sigma(j+iB-a)}\mathrm{d}f\tag{definition of Fourier transform }\\
&=\int_{-\infty}^{+\infty}\hat{s}(f)e^{2\pi \mathbf if\sigma(j-a)}\left(\sum\limits_{i\in\mathbb Z}e^{2\pi \mathbf i if\sigma B}\right)\mathrm{d}f\nonumber\\
&=\int_{-\infty}^{+\infty}\hat{s}(f)e^{2\pi \mathbf if\sigma(j-a)}\left(\frac{1}{\sigma B}\sum\limits_{k\in \mathbb Z}\delta(f-\frac{k}{\sigma B})\right) \mathrm{d}f\tag{similar proof to \eqref{eq: connection between comb and exp}}\\
&=\frac{1}{\sigma B}\sum\limits_{k\in \Z}\hat{s}\left(\frac{k}{\sigma B}\right)e^{2\pi \mathbf i(j-a)\cdot k/B}.\label{eq: u_j calculation1}\end{align}

Suppose $k=qB+r$ such that $q \in \Z$ and $r\in [B]$. By plugging $k=qB+r$ into \eqref{eq: u_j calculation1}, we have:

\begin{align}
u[j]&=\frac{1}{\sigma B}\sum\limits_{r\in [B]}\sum\limits_{q\in \mathbb Z}\hat{s}\left(\frac{q+r/B}{\sigma}\right)e^{2\pi\mathbf i(j-a)(q+r/B)}\nonumber\\
&=\frac{1}{\sigma B}\sum\limits_{r\in [B]}\sum\limits_{q\in \mathbb Z}\hat{s}\left(\frac{q+r/B}{\sigma}\right)e^{-2\pi\mathbf ia(q+r/B)}e^{2\pi\mathbf i j\cdot r/B}\tag{$e^{2\pi \mathbf i jq}=1$ as $jq$ is an integer}\\
&=\frac{1}{\sigma B}\sum\limits_{r\in [B]}\left(\sum\limits_{q\in \mathbb Z}\hat{s}\left(\frac{q+r/B}{\sigma}\right)e^{-2\pi\mathbf i(q+r/B)a}\right)e^{2\pi\mathbf i j\cdot r/B}.\label{eq: u_j calculation2}
\end{align}

Recall in Fact~\ref{fact:FFT_DFT} that 
\begin{equation}
\label{eq: DFT of u}
u[j]=\frac{1}{B}\sum_{r\in [B]}\wh{u}[r]e^{2\pi\mathbf i j\cdot r/B}.
\end{equation}

Comparing \eqref{eq: u_j calculation2} with \eqref{eq: DFT of u}, we obtain
\begin{equation}
\label{eq:hat u (in process)}
\wh{u}[r]=\frac{1}{\sigma}\sum\limits_{q\in \mathbb Z}\wh{s}\left(\frac{q+r/B}{\sigma}\right)e^{-2\pi\mathbf i(q+r/B)a}.
\end{equation}

By substituting $i,j$ for $q,r$ in \eqref{eq:hat u (in process)} respectively, we have
\[
\wh{u}[j]=\frac{1}{\sigma} \sum_{i \in \Z} \wh{s}\left(\frac{j/B + i}{\sigma}\right) \cdot e^{-2\pi\bi (j/B+i) \cdot a}.
\]
\end{proof}

\begin{proofof}{Lemma~\ref{lem:property_hash_bins}}

As $s(t)=w(t) \cdot e^{2 \pi \bi b (t + \sigma a)} \cdot G(\frac{t + \sigma a}{\sigma})$, $\wh{s}(f)=\wh{w}(f) * \wh{\bigg( e^{2 \pi \bi b(t+ \sigma a)} \cdot G(\frac{t+a \sigma}{\sigma}) \bigg)} (f)$. Because
\begin{align*}
\wh{\bigg( e^{2 \pi \bi b(t+ \sigma a)} \cdot G(\frac{t+a \sigma}{\sigma}) \bigg)} (f) & = \int e^{2 \pi \bi b(t+ \sigma a)} \cdot G(\frac{t+a \sigma}{\sigma}) e^{-2\pi\bi f t}\mathrm{d} t\\
& = \int e^{2 \pi \bi b(t+ \sigma a)} \cdot G(\frac{t+a \sigma}{\sigma}) e^{-2\pi\bi (\sigma f) \cdot \frac{t+a \sigma}{\sigma} + 2 \pi \bi \sigma f a}\mathrm{d} t \\
& = e^{2 \pi \bi f a \sigma} \cdot \int G(\frac{t+a \sigma}{\sigma}) e^{-2\pi\bi (\sigma f - \sigma b) \cdot \frac{t+a \sigma}{\sigma} }\mathrm{d} t =\sigma \cdot e^{2 \pi \bi \sigma f a} \cdot \wh{G}(\sigma f - \sigma b),
\end{align*}
$\wh{s}(f)$ becomes
\[
\int_{-\infty}^{+\infty} \wh{w}(\tau) \cdot \sigma \cdot e^{2 \pi \bi \sigma (f - \tau) a} \cdot \wh{G}(\sigma f -\sigma b - \sigma \tau ) \mathrm{d} \tau
\]
such that \eqref{eq:Fourier_coef_u} equals
\begin{align*}
    \wh{u}[j] & =\frac{1}{\sigma} \sum_{i \in \Z} \left( \int_{-\infty}^{+\infty} \wh{w}(s) \cdot \sigma e^{2 \pi \bi \sigma (\frac{j/B + i}{\sigma} - s) a }  \cdot \wh{G}(j/B + i - \sigma s -\sigma b) \mathrm{d} s \right) \cdot e^{-2\pi\bi (j/B+i) \cdot a} \\
    & = \sum_{i \in \Z} \int \wh{w}(s) \cdot \wh{G}(j/B + i - \sigma s - \sigma b) \cdot e^{2 \pi \bi \sigma (-s) a} \mathrm{d} s\\
    & = \int \wh{w}(s) \cdot \sum_{i \in \Z} \wh{G}(j/B + i - \sigma s - \sigma b) \cdot e^{ - 2 \pi \bi \sigma s a} \mathrm{d} s.
\end{align*}   
Since $\sum_{i \in \Z} \wh{G}(j/B + i - \sigma s - \sigma b)=\wh{G}^{(j)}_{\sigma,b}(s)$,
we further simplify the above equation of \eqref{eq:Fourier_coef_u} as
\[
\int \wh{w}(s) \cdot \wh{G}^{(j)}_{\sigma,b}(s) \cdot e^{ - 2 \pi \bi \sigma s a} \mathrm{d} s=\int \wh{z^{(j)}}(s) \cdot e^{ - 2 \pi \bi \sigma s a} \mathrm{d} s=z^{(j)}(-a \sigma).
\]

\end{proofof}

\appendix

\appendix
\section{Related Work}\label{sec:related_work}
\paragraph{Sparse Fourier transforms in the discrete setting.} Sparse discrete Fourier transforms have a large literature, with rich connections to cryptography \cite{GL89} and coding theory \cite{AGS}. These results can be divided into two lines of work. The first line carefully chooses samples (measurements) to allow sublinear time recovery (to name a few \cite{GGIMS,GMS,HIKP,Iw13,IKP,K16}). Our result is closely related to this line --- our interpolation algorithm in Theorem~\ref{inform:improper_learning} is based on one such algorithm. Another line of research considers randomly chosen samples (measurements) and gives generic recovery algorithms such as $\ell_1$ minimization under the restricted isometry property \cite{RV08,HR16}. While the first line has better sample complexity and running time, the second line achieves smaller failure probabilities. For a discrete domain of size $N$, the best known results achieve $O(k\log N)$ samples \cite{IK} or $O(k\log^2 N)$ time \cite{HIKP} separately.

However, algorithms in the discrete setting cannot be applied directly to the continuous problem studied in this work. If the continuous problem has frequencies ``off-the-grid'', the discrete approximation becomes $\Omega(k^2)$-sparse. This is one reason that the main result (Theorem~\ref{inform:improper_learning}) applies a Haar transform to the Fourier series.

\paragraph{Sparse Fourier transforms in the continuous setting.} A line of research \cite{BCGLS,HK15,PS15,SSWZ22} has constructed sparse Fourier transform algorithms in the continuous setting directly. These algorithms learn each frequency $f_j$ and its amplitudes $\alpha_j$ in time $k \cdot (\log k F)^{O(1)}$ like the discrete algorithms. However, all these algorithms require that the $k$ frequencies in $x$ to have a separation $\min_{i \neq j}|f_i-f_j| \ge \frac{(\log k)^{\Omega(1)}}{T}$.

\paragraph{Super-resolution.} Learning $k$-Fourier-sparse signals is closely related to a fundamental task in imaging, called super-resolution. The task is to recover frequencies and amplitudes in $x(t):=\sum_{j=1}^{k} \alpha_j e^{2 \pi \bi f_j t}$. There are a variety of methods that work in the noiseless setting for $m=k$ samples, including Prony's method \cite{Prony}, Reed-Solomon decoding \cite{m69}, and the matrix pencil method \cite{BM86} (see more references in \cite{KATZ2024101687,DELZ24}). However, for exponentially small noise in the time window $[-1,1]$, Moitra \cite{Moitra} showed that it is impossible to recover each frequency accurately when the gap between frequencies is less than $\frac{1}{2}$. At the same time, Moitra provided an algorithm with $O(F)$ samples that tolerates polynomially small noise, when the gap between frequencies is at least $\frac{1}{2}$. Various algorithms (to name a few \cite{FL12,TBSR,CF14,YX15}) based on convex optimization and compressed sensing have been developed in the last two decades. However, all these algorithms require the gap between frequencies to be at least $\frac{1}{2}$ in order to recover frequencies.

While both sparse Fourier transform in the continuous setting and super-resolution study algorithms for recovering frequencies, they focus on different aspects. The goal of sparse Fourier transform is to optimize the running time (and sample complexity). On the other hand, super-resolution is concerned with how the gap between frequencies affects other parameters such as the sample complexity, robustness, and the length of the time window.

\paragraph{Interpolating Fourier-sparse signals.} Chen, Kane, Price, and Song \cite{CKPS17} showed that the gap between frequencies is not necessary for learning the whole signal. Their algorithm learns frequencies with error $k^{O(1)}$ and provides an interpolation with sparsity $(k \log F)^{O(1)}$ in $(k \log F)^{O(1)}$ time. This result has been improved significantly by subsequent works \cite{CP19_colt,CP19_ICALP,SSWZ23}. Specifically, Chen and Price proposed a weighted sampling distribution to reduce  sample complexity \cite{CP19_colt} and improved the construction of filter functions and sampling algorithms in \cite{CP19_ICALP}. Song, Sun, Weinstein, and Zhang \cite{SSWZ23} extended these techniques to provide an efficient interpolation in $m=\tilde{O}(k^4)$ samples and $m^{\omega+o(1)}$ time and learn the frequencies with error $\tilde{O}(k^3)$ in $\tilde{O}(k^2)$ time. Our interpolation algorithm is motivated by these works.

Moreover, the techniques developed in these works have found applications beyond sparse Fourier transforms. The authors of \cite{AKMMVZ18} showed how to reconstruct signals with simple Fourier spectra. While their result provides almost optimal sample complexity to guarantee $\tilde{x}(t) \approx x(t)$ in the time window, it assumes that the positions of the spectra are given. Another application is to learn eigenvalues of Toeplitz matrices \cite{ELMM20}. Moreover, properties of Fourier-sparse signals are closely related to random Fourier features in kernel methods  \cite{RR07_Fourier_features,AKMMVZ17,NEURIPS2020_012d9fe1}.

Two recent works \cite{LLM22, CDHNSY25} studied different approaches to interpolate Fourier-sparse signals. Li, Liu, and Moitra \cite{LLM22} showed how to obtain a $\tilde{O}(k)$-Fourier-sparse signal by post-processing the interpolation $\tilde{x}(t):=\sum_i e^{2 \pi \bi \tilde{f}_i t} \cdot p_i(t)$ from the above framework. While it is more efficient than the proper learning algorithm in Theorem~\ref{inform:stat_bound} (of time $\tilde{O}(k)^{O(k)}$), its guarantee becomes $\tilde{x} \approx x$ in a (slightly) smaller interval $[-1+\delta,1-\delta]$ instead of the time window $[-1,1]$. The authors of \cite{CDHNSY25} studied information-theoretic bounds for reconstruction of $\wh{x}$ under the Wasserstein distance and the heavy hitter distance.

\paragraph{Exponential Sums.} Various properties of Fourier-sparse signals have been studied in approximation theory and Fourier analysis \cite{TuranBook} in terms of inequalities of exponential sums. In particular, Erd\'{e}lyi \cite{erdelyi2016inequalities} proved tight bounds on $\underset{k\text{-Fourier-sparse } x}{\max} \frac{|x(t)|}{\|x\|_{[-1,1]}}$. These bounds are extremely useful in designing algorithms for Fourier-sparse signals (although weaker bounds were used in \cite{CKPS17,CP19_colt,CP19_ICALP}).
Very recently, Zhang~\cite{zhang2026optimalextrapolationboundssparse} proved a tight bound on the growth rate of Fourier-sparse signals outside the time window, which is used in this work to prove Lemma~\ref{inform:almost_orthogonal_clusters}.

\paragraph{Multiband signals.}  Let $|\supp(\wh{x})|$ denote the size of the support of $\wh{x}$ for convenience. For reconstructing single-band signals, early works by Nyquist, Shannon, and others \cite{wmttaker1915functions,nyquist1928certain,kotel1933carrying,shannon1949communication} showed that $O(|\supp(\wh{x})|)$ samples suffice. However, this does not provide a reconstructing algorithm directly. The seminal works by Slepian, Landau, and Pollak \cite{slepian1961prolate,landau1961prolate,landau1962prolate} showed an interpolation approach via $O(|\supp(\wh{x})|)$ prolate spheroidal wave functions. Subsequent works \cite{Landau1980EigenvalueDO,xiao2001prolate,shkolnisky2006approximation,OR14} provided efficient reconstructing algorithms with optimal sample complexity $O(|\supp(\wh{x})|)$. However, all these works assume $\supp(\wh{x})$ is known or given as prior knowledge. One exception is \cite{CP19_ICALP} that proposed the two conditions ($\kappa$ and $\tau$ discussed in Corollary~\ref{cor:interpolating_multiband}) to recover the location of $\supp(\wh{x})$.

For reconstructing multiband signals, Landau and Widom \cite{Landau1980EigenvalueDO} proved that $O(n \log \frac{1}{\epsilon} + \sum_{i=1}^n |I_i|)$ samples suffice. Later on, Avron et al. \cite{AKMMVZ18} provided an efficient reconstructing algorithm with $m=O(n \log \frac{1}{\epsilon} + \sum_{i=1}^n |I_i|)$ samples and $m^{O(1)}$ time. Actually, Avron et al. showed that their approach works for any signal with a simple Fourier spectrum (including the Gaussian density and the Cauchy-Lorentz density). However, this work \cite{AKMMVZ18} still required $\supp(\wh{x})$ as a priori knowledge and left the recovery of $\supp(\wh{x})$ as an important open question.

\section{Proofs}

\subsection{Proof of Fact \ref{fact:bottom_vectors}}\label{sec:proof_bottom}

For $u,v\in \C^N$, let $S_u$ and $S_v$ be the index set of the 
$N-s$ elements with the smallest absolute values in vector $u$ and $v$ respectively. This implies 

\begin{align*}
&\bott_s(u)=u_{S_u},&\bott_s(v)=v_{S_v}.
\end{align*}

Since $S_u$ contains the $N-s$ elements with the smallest absolute values in $u$, $S_u$ minimizes the $\ell_2$ norm among all index sets of size $N-s$. We have

\[
\|u_{S_u}\|_2^2\le \|u_{S_v}\|_2^2.
\]

Finally,

\begin{align*}
\|\bott_s(u)\|^2_2=\|u_{S_u}\|_2^2 \le \|u_{S_v}\|_2^2&\le \|u_{S_v}-v_{S_v}+v_{S_v}\|_2^2\\
&\le 2\|u_{S_v}-v_{S_v}\|_2^2+2\|v_{S_v}\|_2^2\tag{$\|a-b\|_2^2\le 2\|a\|_2^2+2\|b\|_2^2$ for any vector $a,b$}\\
&\le 2\|u-v\|_2^2+2\|\bott_s(v)\|^2_2\\
&=O(\|u-v\|^2_2 + \|\bott_s(v)\|^2_2 ).
\end{align*}

\subsection{Proof of Claim~\ref{claim:sqrt-localizer-bounds}}\label{sec:properties_filters}

We state two basic properties of $\rect$ and $\sinc$ for the analysis of Claim~\ref{claim:sqrt-localizer-bounds}.
\begin{fact} \label{fact:sinc_bounds}
    We have two bounds on the $\sinc$ function:
    \begin{enumerate}
    \item For any $|x| \ge \frac{1.2}{\pi}$, $|\sinc(x)| \le \frac{1}{\pi |x|}$.
    \item For any $|x| \le \frac{1.2}{\pi}$, $\sinc(x) \in [1 - \frac{\pi^2 |x|^2}{6}, 1 - \frac{\pi^2 |x|^2}{10}]$.
    \end{enumerate}
\end{fact}

We first verify the Fourier support. By the definition of $M$, 
\[
\wh{M}(f) = s_M \cdot \rect_{Ch}(f)^{*C\lceil\log(r/\delta)\rceil}
    * \rect_{\frac{C}{h^{-1}+r^{-2}}}(f)^{*C}
    * \cdots
    * \rect_{\frac{C}{h^{-1}+ 1}}(f)^{*C r}
    \cdot \sinc(2 \alpha_M f)
\]
Hence, the radius of the support is at most
\[
\Delta_{M} = \frac{C^2}{2}h\log\frac r\delta + \frac{C^2}{2} \cdot \sum_{i=0}^{\lceil\log  r\rceil} \frac{2^i}{h^{-1}+4^i/r^2}.
\]

It remains to estimate the summation in $\Delta_{M}$.  
For every $i$,
\begin{equation}
\frac{2^i}{h^{-1}+4^i/r^2}
\le \min \left\{2^ih,\frac{r^2}{2^i}\right\}.
\label{eq:sqrt-localizer-support-summand}
\end{equation}
If $r\le\ell/\delta$, then $h=\ell r / \delta \ge r^2$, and the second bound in \eqref{eq:sqrt-localizer-support-summand} implies
\[
\sum_{i=0}^{\lceil\log r\rceil} \frac{2^i}{h^{-1}+4^i/r^2}
\le r^2\sum_{i=0}^{\infty}2^{-i}
= 2r^2 \le\frac{2\ell r}{\delta}.
\]
If $r\ge\ell/\delta$, then $\sqrt h=\ell/\delta$.
Choose the integer $j \ge 0$ such that $2^j\le\frac{r}{\sqrt h}<2^{j+1}$.
Splitting the sum at $j$ and applying the two bounds in \eqref{eq:sqrt-localizer-support-summand} shows
\begin{align*}
\sum_{i=0}^{\lceil\log r\rceil} \frac{2^i}{h^{-1}+4^i/r^2}
\le h\sum_{i=0}^{j}2^i + r^2\sum_{i=j+1}^{\infty}2^{-i}
\le r\sqrt{h} \left( \frac{2^{j+1}\sqrt h}{r} +\frac{r}{2^j\sqrt h} \right)
\le 3 r \sqrt h,
\end{align*}
where the last inequality is by $2 x + x^{-1} \le 3$ for $x \in (1/2, 1]$.
This implies $C_M/2$ in \eqref{eq:def_C_H} is at most $2C^2$.

Next, we prove the properties of the filter $M$ in time domain.
Because every exponent is even, $g_{M}(t) \ge 0$.  

By the definition of $s_{M}$, $s_{M}^{-1} =  (2 \alpha_M)^{-1} \int_{-\alpha_{M}}^{\alpha_{M}}g_{M}(v) \mathrm{d}v$.
Fact~\ref{fact:sinc_bounds} and $1-u \ge \exp(-2u)$ for $0 \le u\le 1/2$ imply that, for $|v|\le 1/8\Delta_{M}$,
\begin{align*}
g_{M}(v)
\ge \exp \left[ - \frac{2 \pi^2}{6} \left(
      C\log\frac r\delta (Chv)^2
      +\sum_{i=0}^{\lceil\log r\rceil}
      C2^i\left(\frac{Cv}{h^{-1}+4^i/r^2}\right)^2
      \right) \right]
\ge \exp \big[ - 16 \Delta_{M}^2v^2 \big],
\end{align*}
where $C\log\frac r\delta (Ch)^2 + \sum_{i=0}^{\lceil\log r\rceil} C2^i\left(\frac{C}{h^{-1}+4^i/r^2}\right)^2 \le \left[C\log\frac r\delta Ch + \sum_{i=0}^{\lceil\log r\rceil} C2^i \frac{C}{h^{-1}+4^i/r^2} \right]^2 \le 4 \Delta_{M}^2$.

Since $1/8\Delta_{M}<\alpha_{M}$, it follows that
\begin{equation}
s_{M}^{-1} = \frac{1}{2 \alpha_M}\int_{-\alpha_{M}}^{\alpha_{M}}g_{M}(v) \mathrm{d}v
\ge \frac{1}{2 \alpha_M} \int_{-1/8\Delta_{M}}^{1/8\Delta_{M}} \exp(-16\Delta_{M}^2v^2) \mathrm{d}v
\ge e^{-1/4}/8\Delta_{M} 
\ge 1/12\Delta_{M}.
\label{eq:sqrt-localizer-normalization-lower}
\end{equation}

For $|v|\ge1/(Ch)$, applying Property~1 of Fact~\ref{fact:sinc_bounds} to the first $\sinc$ factor shows
\begin{align}
\int_{|v|\ge1/(Ch)}g_{M}(v) \mathrm{d}v
\le 2 \int_{1/(Ch)}^\infty (\pi Chv)^{-C\log(r/\delta)} \mathrm{d}v
= \frac{2 \cdot \pi^{-C\log(r/\delta)}} {Ch\bigl(C\log(r/\delta)-1\bigr)}.
\label{eq:sqrt-localizer-first-tail}
\end{align}

Combining \eqref{eq:sqrt-localizer-normalization-lower} and
\eqref{eq:sqrt-localizer-first-tail}, we have
\begin{equation}
s_M \int_{|v|\ge1/(Ch)}g_{M}(v) \mathrm{d}v \le \frac{24 \Delta_{M}}{Ch\bigl(C\log(r/\delta)-1\bigr)} \pi^{-C\log(r/\delta)} \le 128 \frac{r}{\delta} \cdot \left(\frac{\delta}{r}\right)^{C \log \pi} \le \left(\frac{\delta}{r}\right)^{C}.
\label{eq:sqrt-localizer-relative-tail}
\end{equation}

If $|t| \le 1-2/(Ch)$, then $[-1/(Ch),1/(Ch)] \subseteq[t-\alpha_{M},t+\alpha_{M}]$. 
Hence, by \eqref{eq:sqrt-localizer-relative-tail},
\begin{align*}
|1-M(t)| \le \frac{s_M}{2\alpha_M} \int_{|v|\ge1/(Ch)}g_{M}(v) \mathrm{d}v \le \left(\frac{\delta}{r}\right)^{C},
\end{align*}
which proves Property~\ref{item:sqrt-localizer-plateau}. And Property~\ref{item:sqrt-localizer-global} follows by
\begin{align*}
0 \le M(t) \le \frac{s_M}{2\alpha_M} \int_{\R}g_{M}(v) \mathrm{d}v \le 1 + \left(\frac{\delta}{r}\right)^{C},
\end{align*}

By symmetry, we suppose $1 \le t \le 1+1/C$, and let $x = t-1$.
If $C r\sqrt{x} \le 1$, \eqref{eq:sqrt-localizer-relative-tail} already shows that
\begin{align*}
    M(t) \le \left(\frac\delta r\right)^C \le\left(\frac\delta r\right)^{C/2} \exp(-Cr\sqrt{t-1}/8).
\end{align*}

For $(Cr)^{-2} \le x \le 1/C$, there exists an $i \in \{0,\ldots,\lceil\log r\rceil\}$ such that $2^i\le r\sqrt{Cx} < 2^{i+1}$.
Therefore, ${4^i}/{r^2} \le Cx$ and then, for $v \ge x + 1/(Ch)$,
\begin{align*}
\frac{Cv}{h^{-1}+4^i/r^2} \ge 1.
\end{align*}
Thus the $i$-th factor in the second term of $g_{M}$ satisfies
\[
\left|\sinc \left(
\frac{Cv}{h^{-1}+4^i/r^2}\right)\right|^{C2^i}
\le\pi^{-C2^i}
\le \exp \left(-\frac{C}{8}r\sqrt{t-1}\right).
\]

With \eqref{eq:sqrt-localizer-first-tail}, we obtain
\begin{align*}
M(t) \le & s_M \cdot \int_{x+1/(Ch)}^\infty g_{M}(v) \mathrm{d}v \\
\le & s_M \cdot \exp \left(-\frac{C}{8}r\sqrt{t-1}\right) \cdot \int_{1/(Ch)}^\infty |\sinc(Chv)|^{C\log(r/\delta)} \mathrm{d}v\\
\le & (\delta/r)^{C/2} \exp \left(-\frac{C}{8}r\sqrt{t-1}\right).
\end{align*}
This proves Property~\ref{item:sqrt-localizer-near-tail}.

Finally, we prove Property~\ref{item:sqrt-localizer-far-tail}. 
Again suppose $t\ge1+1/C$ by symmetry.
If $r = 1$, then $C \log(r/\delta) \ge Cr$. For $v\ge t-1+1/(Ch)$, we have
\[
Chv \ge Ch(t-1)+1
\ge \begin{cases}
1+h, & 1+1/C \le t \le2,\\
Ch t/2, & t \ge 2.
\end{cases}
\]
Hence, $Chv=\Omega(t)$ and
\[
|\sinc(Chv)|^{(C/2)\log(r/\delta)} \le |\pi t/2|^{-Cr/4}
\]
by the second sinc bound in Fact~\ref{fact:sinc_bounds}. So in this case, the first sinc factor shows
\begin{align*}
M(t) \le & s_M \cdot \int_{x+1/(Ch)}^\infty \sinc(Chv)^{C\log(r/\delta)} \mathrm{d}v \\
\le & s_M \cdot |\pi t/2|^{-Cr/4} \cdot \int_{1/(Ch)}^\infty |\sinc(Chv)|^{C\log(r/\delta)/2} \mathrm{d}v\\
\le & (\delta/r)^{C/2} |\pi t/2|^{-Cr/4}.
\end{align*}

We next consider $r \ge 2$. 
Since $\delta \le 1/2$, $h = \frac{\ell}{\delta} \cdot \min (\frac{\ell}{\delta}, r) \ge 4$. 
Then we choose $i$ so that $r/4 \le 2^i \le r/2$.
Then $C2^i \ge Cr/4$ and $h^{-1}+4^i/r^2 \le 1$.
For $v \ge t-\alpha_{M}=t-1+1/(Ch)$, the sinc bound implies
\[
\left|\sinc \left(\frac{Cv}{h^{-1}+4^i/r^2}\right)\right|^{C2^i}
\le \left(\pi C v\right)^{-C2^i}.
\]
If $1+1/C\le t\le2$, then $Cv \ge 1$ and $\pi^{-Cr/4} \le |\pi t/2|^{-Cr/4}$.
If $t \ge 2$, then $v \ge t-1 \ge t/2$, and $|\pi Ct|^{-Cr/4} \le |\pi t/2|^{-Cr/4}$.
Therefore,
\begin{align*}
M(t) \le & s_M \cdot \int_{x+1/(Ch)}^\infty \sinc(Chv)^{C\log(r/\delta)} \cdot \sinc \left(\frac{Cv}{h^{-1}+4^i/r^2}\right)^{C 2^i}\mathrm{d}v \\
\le & s_M \cdot |\pi t/2|^{-Cr/4} \cdot \int_{1/(Ch)}^\infty |\sinc(Chv)|^{C\log(r/\delta)/2} \mathrm{d}v\\
\le & (\delta/r)^{C/2} |\pi t/2|^{-Cr/4}.
\end{align*}
Combining the above two cases proves Property~\ref{item:sqrt-localizer-far-tail}.

\subsection{Proof of Lemma~\ref{lem:filter_sparse}}\label{sec:proof_filters}
We prove Lemma~\ref{lem:filter_sparse} in this section.
Let $M = M_{\ell,\delta}$. Since $0 < \delta \le 1/2$, the parameter $h$ in \eqref{eq:sqrt-localizer-scale} is $\ell^2/\delta$.

We first compare the energy on \([-1,1]\).  Let $I=\left[-1+\frac{2\delta}{C\ell^2}, 1-\frac{2\delta}{C\ell^2}\right]$.
By Property~\ref{item:uniform_bound_on_interval} of Lemma~\ref{lemma:bounds_Fourier_sparse_signals},
\begin{align*}
\int_{[-1,1]\setminus I}|x(t)|^2 dt
\le \frac{4\delta}{C\ell^2} \sup_{t\in[-1,1]}|x(t)|^2
\le \frac{\pi^2\delta}{C}\|x\|_{[-1,1]}^2.
\end{align*}
Properties~1 and 2 of Claim~\ref{claim:sqrt-localizer-bounds} show
\[
|M(t)-1|^2 \le
\begin{cases}
(\delta/\ell)^C, &t\in I,\\
1, &t\in[-1,1]\setminus I.
\end{cases}
\]
Consequently,
\begin{align}
\left| \|Mx - x\|_{[-1,1]}^2 \right| & \le \int_{-1}^1 (\delta/\ell)^{C} \cdot |x(t)|^2 \mathrm{d} t + \frac{\pi^2\delta}{C}\|x\|_{[-1,1]}^2 \\
& \le \left[ \left(\frac{\delta}{\ell}\right)^C + \frac{\pi^2 \delta}{C} \right] \cdot \|x\|_{[-1,1]}^2
\le \frac{\delta}{4} \cdot \|x\|_{[-1,1]}^2.
\label{eq:addendum-inside-energy}
\end{align}

Next, we bound the energy outside \([-1,1]\).  

For \(1\le |t|\le1+1/C\), Property~3 of Claim~\ref{claim:sqrt-localizer-bounds} and Property~\ref{item:square_root_bound_outside_interval} of Lemma~\ref{lemma:bounds_Fourier_sparse_signals} imply
\begin{align}
\int_{1\le|t|\le1+1/C}|M(t)x(t)|^2 dt 
& \le 2 (192 \ell)^2 \left(\frac{\delta}{\ell}\right)^{C}
    \int_0^{1/C} \exp \left[- \left( \frac{C}{4} - 6 \sqrt{2} \right) \ell\sqrt u\right] \mathrm{d} u \cdot \|x\|_{[-1,1]}^2 \notag \\    
& \le \left(\frac{\delta}{\ell}\right)^{C/2} \cdot \|x\|_{[-1,1]}^2.
\label{eq:addendum-single-near-integral}
\end{align}

For \(|t|\ge1+1/C\), Property~4 of Claim~\ref{claim:sqrt-localizer-bounds} and Property~\ref{item:polynomial_bound_outside_interval} of Lemma~\ref{lemma:bounds_Fourier_sparse_signals} show
\begin{align}
\int_{|t|\ge1+1/C}|M(t)x(t)|^2 dt
& \le \frac{\pi^2 \ell^2}{2} (2e)^{2\ell} \left( \frac{\pi}{2} \right)^{-C \ell /2} \left(\frac{\delta}{\ell}\right)^{C} \|x\|_{[-1,1]}^2
    \cdot \int_{1}^{\infty}t^{-\frac{C}{2} \ell+2\ell} dt
\le \left(\frac{\delta}{\ell}\right)^{C / 2} \|x\|_{[-1,1]}^2.
\label{eq:addendum-single-far-integral}
\end{align}
Hence, Equations \eqref{eq:addendum-single-near-integral} and \eqref{eq:addendum-single-far-integral} imply
\begin{equation}
\|Mx\|_{\mathbb R\setminus[-1,1]}^2 \le \frac{\delta}{4} \cdot \|x\|_{[-1,1]}^2.
\label{eq:addendum-single-tail}
\end{equation}
And we have 
\[
\begin{aligned}\left|\|Mx\|_{[-1,1]}^2-\|x\|_{[-1,1]}^2\right|&\le\int_{-1}^1|M(t)^2-1||x(t)|^2\,\mathrm dt\\&\le\left[3(\delta/\ell)^{C/2}+\frac{\pi^2\delta}{C}\right]\|x\|_{[-1,1]}^2\\&\le\frac\delta4\|x\|_{[-1,1]}^2\end{aligned}
\] for sufficiently large $C$. 
Combining this with \eqref{eq:addendum-single-tail} gives \[\left|\|Mx\|_2^2-\|x\|_{[-1,1]}^2\right|\le\frac\delta2\|x\|_{[-1,1]}^2.\]

Therefore,
\[
\|Mx\|_{\mathbb R\setminus[-1,1]}^2
\le\frac{\delta/4}{1-\delta/2} \cdot \|Mx\|_2^2
\le\delta\|Mx\|_2^2,
\]
and
\[
\|Mx\|_{[-1,1]}^2=\|Mx\|_2^2-\|Mx\|_{\mathbb R\setminus[-1,1]}^2\ge(1-\delta)\|Mx\|_2^2.
\]

Finally, one has $\widehat{Mx}(f)=\sum_j\alpha_j\widehat M(f-f_j)$. The construction gives $\supp(\widehat M)\subseteq[-\Delta_{\ell,\delta},\Delta_{\ell,\delta}]$, so $\supp(\widehat{Mx})\subseteq\bigcup_j[f_j-\Delta_{\ell,\delta},f_j+\Delta_{\ell,\delta}]$, as required.

\subsection{Proof of Lemma~\ref{lem:almost_orthogonal_clusters}}\label{sec:almost_orthogonal_clusters}
By symmetry, we suppose $\ell \le r$.
Let $M=M_{\ell,r,\delta}$ and $C_M := 4 C^2$.  

For $w(t)=\sum_{j=1}^{\ell}\alpha_j e^{2\pi\bi f'_jt}, z(t)=\sum_{j=1}^{r}\beta_j e^{2\pi\bi f_j t}$, we have
\begin{align}
\wh{Mw}(f)
& = \int_{\mathbb R}M(t) \sum_{j=1}^{\ell}\alpha_j e^{2\pi\bi f'_jt}e^{-2\pi\bi f t} dt \notag \\
& = \sum_{j=1}^{\ell}\alpha_j \int_{\mathbb R}M(t)e^{-2\pi\bi(f-f'_j)t} dt
= \sum_{j=1}^{\ell}\alpha_j\wh M(f-f'_j), \label{eq:addendum-direct-modulation-w}\\
\wh{Mz}(f) & = \sum_{j=1}^{r}\beta_j\wh M(f-f_j).
\label{eq:addendum-direct-modulation-z}
\end{align}
So the separation assumption and Property~\ref{item:sqrt-localizer-support} of
Claim~\ref{claim:sqrt-localizer-bounds} imply that \eqref{eq:addendum-direct-modulation-w} and \eqref{eq:addendum-direct-modulation-z} are disjoint, which also means
\[
\int_{\mathbb R}M(t)w(t)\overline{M(t) z(t)} \mathrm{d}t = \int_{\mathbb R}\wh{Mw}(f) \overline{\wh{Mz}(f)} \mathrm{d}f =0.
\]
Consequently,
\begin{align}
|\langle w,z\rangle_{[-1,1]}|
&\le \left|\int_{[-1,1]} M(t)^2 w(t)\overline{z(t)} \mathrm{d}t\right|
+ \left|\int_{-1}^{1} (1-M(t)^2)w(t)\overline{z(t)} \mathrm{d}t\right| \\
&\le \left|\int_{\mathbb R\setminus[-1,1]} M(t)^2 w(t)\overline{z(t)} \mathrm{d}t\right|
+ \left|\int_{-1}^{1} (1-M(t)^2)w(t)\overline{z(t)} \mathrm{d}t\right|.
\label{eq:strengthened-localizer-decomposition}
\end{align}

By applying Properties~\ref{item:square_root_bound_outside_interval} and
\ref{item:polynomial_bound_outside_interval} of
Lemma~\ref{lemma:bounds_Fourier_sparse_signals} to $w$ and
$z$, we have
\[
|w(t)z(t)|
\le r^{O(1)} \|w\|_{[-1,1]}\|z\|_{[-1,1]} \cdot
\begin{cases}
    \exp(O(r\sqrt{|t|-1})),&1\le|t|\le1+1/C,\\
    \bigl(e(|t|+1)\bigr)^{O(r)},&|t|\ge1+1/C.
\end{cases}
\]
Plus Properties~\ref{item:sqrt-localizer-near-tail} and~\ref{item:sqrt-localizer-far-tail} of Claim~\ref{claim:sqrt-localizer-bounds}, we can bound the first term in \eqref{eq:strengthened-localizer-decomposition} with $(\delta/r)^{C/4} \cdot \|w\|_{[-1,1]}\|z\|_{[-1,1]}$.

It remains to bound the second term in \eqref{eq:strengthened-localizer-decomposition}. 

Let $I=\left[-1+\frac{2}{Ch},1-\frac{2}{Ch}\right]$, where $h =\frac{\ell}{\delta} \cdot \min \left\{r, \frac{\ell}{\delta}\right\}$.
On $I$, $|1-M(t)^2|\le3(\delta/r)^{C/2}$. The Cauchy--Schwarz inequality therefore bounds the contribution by $3(\delta/r)^{C/2}\|w\|_{[-1,1]}\|z\|_{[-1,1]}$.

The total length of $[-1, 1] \setminus I$ is $4/(Ch)$ and $1-M(t)^2 \le 1$ in \eqref{eq:strengthened-localizer-decomposition} in this interval. We discuss the second integral in \eqref{eq:strengthened-localizer-decomposition} by two cases:
\begin{enumerate}
\item Case 1: $r\le\ell/\delta$. Then $h = \ell r / \delta$. Moreover, applying Property~\ref{item:uniform_bound_on_interval} of Lemma~\ref{lemma:bounds_Fourier_sparse_signals} to $w$ and $z$ implies
\begin{align*}
\left|\int_{[-1,1]\setminus I} (1-M(t)^2)w(t)\overline{z(t)} \mathrm{d}t\right|
\le & \Big| [-1,1]\setminus I \Big| \cdot \sup_{t \in [-1, 1]} |w(t)| \cdot \sup_{t \in [-1, 1]} |z(t)| \\
\le & \frac{\pi^2 \ell r}{Ch} \cdot \|w\|_{[-1,1]}\|z\|_{[-1,1]}
\le \frac{\pi^2 \delta}{C} \cdot \|w\|_{[-1,1]}\|z\|_{[-1,1]}.
\end{align*}
\item Case 2: $r>\ell/\delta$. We use Cauchy--Schwarz and apply Property~\ref{item:uniform_bound_on_interval} of Lemma~\ref{lemma:bounds_Fourier_sparse_signals} to only $w$, which yields
\begin{align*}
\left|\int_{[-1,1]\setminus I} (1-M(t)^2)w(t)\overline{z(t)} \mathrm{d}t\right|
\le & \| z \|_{[-1,1]\setminus I} \cdot \sqrt{\int_{[-1,1]\setminus I} (1-M(t)^2)^2 |w(t)|^2 \mathrm{d}t} \\
\le & \| z \|_{[-1,1]} \cdot \sup_{t \in [-1, 1]} |w(t)| \cdot \sqrt{\Big| [-1,1]\setminus I \Big|} \\
\le & \frac{\pi \ell}{\sqrt{Ch}} \cdot \|w\|_{[-1,1]}\|z\|_{[-1,1]}
= \frac{\pi \delta}{\sqrt{C}} \cdot \|w\|_{[-1,1]}\|z\|_{[-1,1]}.
\end{align*}
\end{enumerate}

Combining these bounds in \eqref{eq:strengthened-localizer-decomposition} proves the lemma.

\subsection{Proof of Corollary~\ref{cor:filter_functions}}\label{sec:proofs_cor_filters}

In this section, let $C$ be the constant defined in \eqref{eq:sqrt-localizer-kernel}. By the definition of $M_{\ell, \delta} = M_{\ell, \ell, \delta}$, 
\[
\Delta_{\ell,\delta}
= \frac{C^2}{2} \left( 
    \frac{\ell^2}{\delta} \left\lceil\log \frac{\ell}{\delta}\right\rceil
    + \sum_{i=0}^{\lceil\log \ell\rceil} \frac{2^i}{\delta/\ell^2+4^i/\ell^2}
\right).
\]
Hence, we have
\begin{equation} \label{eq:direct-radius-small}
\frac{C^2\ell^2}{2\delta} \log \frac\ell\delta
\le \Delta_{\ell,\delta}
\le\frac{C^2\ell^2}{\delta} \log \frac\ell\delta.             
\end{equation}
Similarly,
\begin{equation} \label{eq:direct-radius-large}
\frac{C^2k^2}{2\epsilon} \log \frac k\epsilon
\le \Delta_{k,\epsilon}
\le\frac{C^2k^2}{\epsilon} \log \frac k\epsilon.             
\end{equation}

Repeating the calculation for \eqref{eq:strengthened-localizer-decomposition} shows
\begin{align}
\|(M_{k,\epsilon}-1)w\|_{[-1,1]}^2 + \|M_{k,\epsilon}w\|_{\mathbb R\setminus[-1,1]}^2    
& \le O \left(\frac{\ell^2\epsilon}{k^2} +\left(\frac\epsilon k\right)^{C/4}\right) \|w\|_{[-1,1]}^2. \label{eq:direct-localization-large} \\
\text{ and } \quad \|(M_{\ell,\delta}-1)w\|_{[-1,1]}^2 +\|M_{\ell,\delta}w\|_{\mathbb R\setminus[-1,1]}^2
& \le O(\delta)\|w\|_{[-1,1]}^2  \label{eq:direct-localization-small}
\end{align}

\begin{proof}[Proof of the first property]
We first prove that, under the condition of Corollary~\ref{cor:filter_functions}, one has
\begin{equation}
\frac{\ell^2\epsilon}{k^2}\le 2\delta. \label{eq:filter-comparison-parameter}
\end{equation}
Suppose ${\ell^2\epsilon}/{k^2} \ge \delta$, otherwise \eqref{eq:filter-comparison-parameter} already holds. With $\ell \le k$, the assumption shows $\ell / \delta \ge k / \epsilon$. In this case, the condition $\Delta_{\ell,\delta} \le \Delta_{k,\epsilon}$, with \eqref{eq:direct-radius-small} and $\eqref{eq:direct-radius-large}$, implies \eqref{eq:filter-comparison-parameter}.

Therefore,
\begin{align*}
    & \| M_{\ell,\delta} \cdot w - M_{k,\epsilon} \cdot w \|_2^2 \\
    = & \| M_{\ell,\delta} \cdot w - M_{k,\epsilon} \cdot w \|_{[-1, 1]}^2 + \| M_{\ell,\delta} \cdot w - M_{k,\epsilon} \cdot w \|_{\R \setminus [-1, 1]}^2 \\
    \le & 2 \|(M_{\ell,\delta}-1)w\|_{[-1,1]}^2 + 2 \|M_{\ell,\delta}w\|_{\mathbb R\setminus[-1,1]}^2
        + 2 \|(M_{k,\epsilon}-1)w\|_{[-1,1]}^2 + 2 \|M_{k,\epsilon}w\|_{\mathbb R\setminus[-1,1]}^2 
        \tag{by $|a-b|^2 \le 2 |a|^2 + 2 |b|^2$} \\
    \le & O(\delta) \|w\|_{[-1,1]}^2 
        + O \left(\frac{\ell^2\epsilon}{k^2} +\left(\frac\epsilon k\right)^{C/4}\right) \|w\|_{[-1,1]}^2 
        \tag{by \eqref{eq:direct-localization-small} and \eqref{eq:direct-localization-large}} \\
    \le & O(\delta) \|w\|_{[-1,1]}^2 
    \tag{by \eqref{eq:filter-comparison-parameter}} \\
    \le & O(\delta) \| M_{k,\epsilon} w\|_{[-1,1]}^2.
    \tag{by Lemma~\ref{lem:filter_sparse}}
\end{align*}

\end{proof}

\begin{proof}[Proof of the second property]
Let $C_H:=8 \cdot C_M$ for the constant $C_M$ defined in Lemma~\ref{lem:almost_orthogonal_clusters} and $d$ denote $\min_{j,j'}|f_j-f'_{j'}|$.
If $d > 2 \Delta_{k,\epsilon}$, the Fourier supports of $M_{k,\epsilon}w$ and $M_{k,\epsilon}z$ are disjoint, so Parseval's identity shows
\begin{align*}
    \langle M_{k,\epsilon}w,M_{k,\epsilon}z\rangle =\langle\wh{M_{k,\epsilon}w}, \wh{M_{k,\epsilon}z}\rangle=0
\end{align*}

Thus we suppose $2 \Delta_{k,\epsilon} \ge d \ge C_H \cdot \frac{\ell r}{ \delta } \big(\log \frac{\ell r}{\delta} + \log \frac{k}{\epsilon} \big)$ in the remaining proof. In this case, \eqref{eq:direct-radius-large} implies
\begin{equation}
\frac{\epsilon\ell r}{k^2} \le\frac{2C^2}{C_H}\delta.
\label{eq:filtered-boundary-parameter}
\end{equation}

We now compare the filtered inner product with the unfiltered one:
\begin{align}
\label{eq:filtered-comparison-decomposition}
\langle M_{k,\epsilon}w,M_{k,\epsilon}z\rangle -\langle w,z\rangle_{[-1,1]}
= \int_{-1}^{1}(M_{k,\epsilon}(t)^2-1) w(t)\overline{z(t)} dt
    + \int_{\mathbb R\setminus[-1,1]} M_{k,\epsilon}(t)^2w(t)\overline{z(t)} dt.
\end{align}

Applying Lemma~\ref{lem:almost_orthogonal_clusters} to $w$ and $z$ yields (with $C_H=8C_M$)
\begin{equation}
|\langle w,z\rangle_{[-1,1]}| \le\frac{\delta}{4} \|w\|_{[-1,1]}\|z\|_{[-1,1]}.
\label{eq:filtered-use-unfiltered}
\end{equation}

On the interval $I := [-1+2\epsilon/(Ck^2),1-2\epsilon/(Ck^2)]$,
Claim~\ref{claim:sqrt-localizer-bounds} implies $|M_{k,\epsilon}(t)^2-1| \le 3(\epsilon/k)^{C/2}$.
Cauchy--Schwarz therefore shows
\begin{equation}
\int_{I}(M_{k,\epsilon}(t)^2-1) w(t)\overline{z(t)} dt
\le  3(\epsilon/k)^{C/2} \|w\|_{[-1,1]}\|z\|_{[-1,1]}.
\label{eq:filtered-core-integral}
\end{equation}

Meanwhile, Lemma~\ref{lemma:bounds_Fourier_sparse_signals} and Claim~\ref{claim:sqrt-localizer-bounds} yield
\begin{align}
    & \left| \int_{[-1, 1] \setminus I}(M_{k,\epsilon}(t)^2-1) w(t)\overline{z(t)} dt \right| \notag \\
    \le & \Big| [-1, 1] \setminus I \Big| \cdot \left( 1 + (\delta/r)^{C/2} \right) \cdot \frac{\pi\ell}{2} \|w\|_{[-1,1]} \cdot \frac{\pi r}{2} \|z\|_{[-1,1]} \notag \\
    \le & \frac{\pi^2\epsilon\ell r}{Ck^2} \|w\|_{[-1,1]}\|z\|_{[-1,1]}.
    \label{eq:filtered-boundary-integral}
\end{align}

For $\int_{\mathbb R\setminus[-1,1]} M_{k,\epsilon}(t)^2w(t)\overline{z(t)} dt$, similar to \eqref{eq:addendum-single-near-integral} and \eqref{eq:addendum-single-far-integral} in the proof of Lemma~\ref{lem:filter_sparse}, Lemma~\ref{lemma:bounds_Fourier_sparse_signals} and Claim~\ref{claim:sqrt-localizer-bounds} show it is bounded:
\begin{equation}
    \frac{\pi^2 \ell r}{C^2 k} \cdot \left(\frac{\epsilon}{k}\right)^C  \|w\|_{[-1,1]}\|z\|_{[-1,1]} 
    + \pi^2 \ell r \left(\frac{\epsilon}{k}\right)^{2C} (2e)^{\ell+r}\left(\frac{\pi}{2}\right)^{-Ck/2} \cdot \|w\|_{[-1,1]}\|z\|_{[-1,1]}
    \label{eq:filtered-far-integral}
\end{equation}

Combining \eqref{eq:filtered-use-unfiltered}, \eqref{eq:filtered-core-integral}, \eqref{eq:filtered-boundary-integral} and \eqref{eq:filtered-far-integral}, for sufficiently large $C$, we have
\[
\left| \langle M_{k,\epsilon} \cdot w, M_{k,\epsilon} \cdot z\rangle \right| \le \frac{\delta}{2} \|w\|_{[-1,1]}\|z\|_{[-1,1]} \le \delta \cdot \|M_{k,\epsilon} \cdot w\|_2 \cdot \|M_{k,\epsilon} \cdot z\|_2,
\]
where the last step uses Lemma~\ref{lem:filter_sparse}.

\end{proof}

\section{Learning Multiband Signals}\label{sec:proof_main}
We finish the proof of Theorem~\ref{inform:learn_multiband} in this section. Recall that $mid(J)$ denotes the middle point of an interval $J$. Here is a formal restatement.

\begin{theorem}\label{thm:learn_multiband}
    Given parameters $F$, $R \ge 1$, $n$, and a small constant $\epsilon$, let $x$ be a multiband signal whose Fourier spectrum $\supp(\wh{x})=J_1 \cup J_2 \cup \cdots \cup J_n$ in $[-F,F]$ has $n$ \emph{unknown} intervals of length at most $R$. Suppose $x$ satisfies
    \begin{equation}\label{cond:multiband_recovery} 
        \int_{-1}^1 |x(t)|^2 \mathrm{d} t \ge (1-\epsilon) \int |x(t)|^2 \mathrm{d} t.
    \end{equation} 
    
     For any noisy observation $w=x+\eta$ in the time window $[-1,1]$ with $\|\eta\|_{[-1,1]}^2 \le \epsilon \cdot \|x\|_{[-1,1]}^2$,     there exists an efficient algorithm that takes $\tilde{O}(n \cdot R)$ samples and     $\tilde{O}\big( (n R )^{\omega} + n^2 R^3 \big)$ time to output a list of $\ell=\tilde{O}(n)$ frequencies $\tilde{f}_1,\ldots,\tilde{f}_\ell$ and $\tilde{x}$. The output satisfies the following properties with prob.~0.98:
    \begin{enumerate}
        \item for any $J_i$ with $\int_{J_i} |\wh{w}(f)|^2 \mathrm{d}f \ge \frac{\epsilon}{4n} \cdot \|w\|_{[-1,1]}^2$, there exists $\tilde{f}_j$ with $|mid(J_i) - \tilde{f}_j|=O(R)$.
        \item $\|x-\tilde{x}\|_{[-1,1]}^2 = O(\epsilon) \cdot \|x\|_{[-1,1]}^2$. 
    \end{enumerate}
    \end{theorem}
\begin{remark}
    As discussed earlier, for $\supp(\wh{x})=I_1 \cup \cdots \cup I_{n}$ with a very long interval, instead of setting $R=\max |I_i|$, one could pick any $R$ and reset $n':=\sum_i \lceil {|I_i|}/{R} \rceil$ to apply Theorem~\ref{thm:learn_multiband}. 
    
    In particular, if one chooses $R=1$, then $n' \le n+ \sum_{i=1}^n |I_i|$. This set of parameters recovers the statement in Theorem~\ref{inform:learn_multiband}.
\end{remark}

Besides Theorem~\ref{thm:locations_multiband}, this result relies on classical reconstructing algorithms in \cite{slepian1961prolate,landau1961prolate,landau1962prolate,AKMMVZ18}. Specifically, given the multiband $I_1 \cup \ldots \cup I_n$, classical results by Slepian, Landau, and Pollak have shown that space of multiband signals of duration $2T$ and Fourier support $I_1 \cup I_2 \cdots \cup I_n$ is essentially $O(T \cdot \sum_{j=1}^n |I_j|)$-dimensional. For an interval $J:=[a,b]$ in $\mathbb{R}$, let $mid(J)=\frac{a+b}{2}$ denote its middle point for convenience.

\begin{theorem}\label{thm:multiband_interpolate}{\cite{slepian1961prolate,landau1961prolate,landau1962prolate,AKMMVZ18}}
    Given any multiband $I_1,\ldots,I_n$, let $\mathcal{F}$ denote the linear operator from $I_1 \cup I_2 \cdots \cup I_n$ to $[-T,T]$ as $\mathcal{F}[\wh{x}](t)=\int_{I_1 \cup \cdots \cup I_n} \wh{x}(f) e^{2 \pi i f t} \mathrm{d} f$. Then for any $\epsilon>0$, this linear operator $\mathcal{F}$ has $s_\epsilon:=O(T \sum_{j=1}^n |I_j| + n \log \frac{1}{\epsilon})$ eigenfunctions with an eigenvalue $\ge \epsilon$. Moreover, there are efficient procedures to compute these eigenfunctions.

    In particular, given the multiband $I_1,\ldots,I_n$ of $x$ and a noisy observation $y=x+\eta$ in the time window $[-T,T]$, for any $\epsilon$, there exists an efficient interpolation algorithm that takes $O(s_\epsilon \log^2 s_\epsilon)$ samples and $\tilde{O}(s_\epsilon^{\omega}+s_\epsilon^2 \cdot n)$ time to output $\tilde{x}$ with $\|x-\tilde{x}\|_{[-T,T]}^2 \le \epsilon \cdot \|x\|_{[-T,T]}^2 + O(\|\eta\|_{[-T,T]}^2)$.
\end{theorem}
We remark the time complexity in Theorem~\ref{thm:multiband_interpolate} comes from Theorem 3 of \cite{AKMMVZ18}.

\begin{proofof}{Theorem~\ref{thm:learn_multiband}}
    Our algorithm assumes $w(t)=0$ for $t \notin [-1,1]$ whenever it has to query $w(t)$ outside the time window.     Although $w$ and the noise $\eta$ are defined in the time window originally, one could define $\eta(t)=-x(t)$ for any $t \notin [-1,1]$. This would increase $\|\eta\|_2^2$ to at most $2 \epsilon \|x\|_2^2$ from Condition~\eqref{cond:multiband_recovery}.
    
    The algorithm has two steps: in the first step, we apply Theorem~\ref{thm:locations_multiband} to obtain $\tilde{f}_1,\ldots,\tilde{f}_{\ell}$ for $\ell=\tilde{O}(n)$. The sample complexity and time complexity of the first step is $\tilde{O}(n R)$. Then we assume the support of $x$ is $\cup_{i=1}^{\ell} [\tilde{f}_i-O(R),\tilde{f}_i+O(R)]$ in order to apply reconstruction algorithms for multi-band signals from  Theorem~\ref{thm:multiband_interpolate} \cite{slepian1961prolate,landau1961prolate,landau1962prolate,AKMMVZ18}. The second step has sample complexity $\ell \cdot O(R + \log \frac{1}{\epsilon})=\tilde{O}(n R + n \log \frac{1}{\epsilon})$ and time complexity $\tilde{O}\big( (n R + n \log \frac{1}{\epsilon})^{\omega} + n (n R + n \log \frac{1}{\epsilon})^2 \big)$ from Theorem~\ref{thm:multiband_interpolate}.

    Then we analyze its correctness. Similar to Theorem~\ref{thm:locations_multiband}, we call an interval $J_i$ heavy only if $\int_{J_i} |\wh{w}(f)|^2 \mathrm{d}f \ge \frac{\epsilon}{4n} \cdot \|w\|_{2}^2$. Then    Theorem~\ref{thm:locations_multiband} implies that $\tilde{f}_1,\ldots,\tilde{f}_{\ell}$ in the first step satisfy that for any heavy $J_i$, there exists $\tilde{f}_j$ such that $|mid(J_i) - \tilde{f}_j|=O(R)$.

    Finally we bound $\|x-\tilde{x}\|_{[-1,1]}^2$. The key is that the contribution of non-heavy intervals is small. There are two types --- either 
    $\int_{J_i} |\wh{x}(f)|^2 \mathrm{d}f \le \frac{\epsilon}{n} \cdot \|w\|_{2}^2$ or $\int_{J_i} |\wh{\eta}(f)|^2 \mathrm{d}f \ge \frac{1}{4} \int_{J_i} |\wh{x}(f)|^2 \mathrm{d}f$.     The total contribution $\sum_{i} \int_{J_i} |\wh{x}(f)|^2 \mathrm{d}f$ of the first type is at most $\epsilon \cdot \|w\|_{2}^2 \le 2\epsilon \cdot \|x\|_2^2$ since there are at most $n$ intervals. The total contribution $\sum_{i} \int_{J_i} |\wh{x}(f)|^2 \mathrm{d}f$ of the second type is at most $8\epsilon \cdot \|x\|_{2}^2$ since $\int |\wh{\eta}(f)|^2 \mathrm{d}f \le 2\epsilon \cdot \|x\|_2^2$. So let $B:=\cup_i [\tilde{f}_i-O(R),\tilde{f}_i+O(R)]$ be the recovered support of $\wh{x}$. Then we have $\int_B |\wh{x}(f)|^2 \mathrm{d} f \ge (1-10 \epsilon) \|x\|_2^2$. Hence classical algorithms \cite{slepian1961prolate,landau1961prolate,landau1962prolate,AKMMVZ18} in Theorem~\ref{thm:multiband_interpolate} for reconstructing multiband signals provide an output $\tilde{x}$ with $\|x-\tilde{x}\|_{[-1,1]}^2 = O(\epsilon) \cdot \|x\|_{[-1,1]}^2$.
\end{proofof}

Now we prove a corollary that replaces \eqref{cond:multiband_recovery} by natural quantities of $x$. Chen and Price \cite{CP19_ICALP} proposed two quantities to characterize the recovery of a single-band signal $x$. The first quantity is the ratio of the maximum value of $x$ to its average: $\kappa:=\frac{\underset{s \in [-1,1]}{\max} |x(s)|^2}{\underset{s \in [-1,1]}{\E} |x(s)|^2}$. This ratio determines the concentration of $|x(t)|^2$ in the time window, which affects the ``effective" length of the time window under noise, as the above example of $g$ and $h$. The second quantity is the growth rate of $x$ outside the time window --- let $\tau$ be a parameter satisfying $|x(t)| \le \kappa^{O(1)} \cdot |t|^\tau \cdot \max_{s \in [-1,1]} |x(s)|$ for any $t \notin [-1,1]$. Basically, $\tau$ characterizes the trouble that a degree $(k-1)$ polynomial in the time window could either be the Taylor expansion of $e^{2 \pi \bi f t}$ for $k \ge C \cdot f$ or a $(k-1)$-Fourier-sparse signal whose frequencies are around 0. Combining the method of Chen and Price \cite{CP19_ICALP} and our algorithm provides an interpolation algorithm for multiband signals whose $\kappa$ and $\tau$ are bounded.

\begin{corollary}\label{cor:interpolating_multiband}
        Let $x$ be a multiband signal whose Fourier spectrum $\supp(\wh{x})=J_1 \cup J_2 \cup \cdots \cup J_n$ in $[-F,F]$ has $n$ unknown intervals of length at most $R$. For some parameters $\tau$ and $\kappa$, suppose $x$ satisfies
        \begin{enumerate}
            \item $\frac{\max_{s \in [-1,1]} |x(s)|^2}{\E_{s \sim [-1,1]} |x(s)|^2} \le \kappa.$ 
            \item $|x(t)| \le \kappa^{O(1)} \cdot |t|^{\tau} \cdot (\max_{s \in [-1,1]} 
            |x(s)|)$ for any $t \notin [-1,1]$.
        \end{enumerate}

                For any noisy observation $y=x+\eta$ in the time window $[-1,1]$ with $\|\eta\|_{[-1,1]}^2 \le \epsilon \cdot \|x\|_{[-1,1]}^2$, given those parameters $F$, $\tau$, $\kappa$, $R$, and $n$, for $R':=R+O(\tau \log \tau + \frac{\kappa}{\epsilon} \log \frac{\kappa}{\epsilon})$, there exists an efficient algorithm that takes $\tilde{O}(n \cdot R')$ samples and $(n \cdot R')^{O(1)}$ time to output $\tilde{x}$ such that $\|x-\tilde{x}\|_{[-1,1]}^2 = O(\epsilon) \cdot \|x\|_{[-1,1]}^2$.         
\end{corollary}

In Corollary~\ref{cor:interpolating_multiband}, we state $R$ as an upper bound of $\max_i |I_i|$ such that the sample complexity and the time complexity depend on $n \cdot R$ instead of $\sum_i |I_i|$. We remark that in many applications, one could pick a proper $R$ and set $n$ to be an upper bound of $\sum_i \lceil \frac{|I_i|}{R} \rceil$.

\begin{proof}
    Given $\tau$ and $\kappa$, Theorem 4.2 in \cite{CP19_ICALP} provides a pair of filter function $(H,\wh{H})$ such that 
    \begin{enumerate}
        \item $\supp(\wh{H})=[-\Delta,\Delta]$ for $\Delta = C \frac{\kappa }{\epsilon} \log \frac{\kappa}{\epsilon} + C \tau \log \tau$ such that $\wh{H \cdot x}=\wh{H}*\wh{x}$ can be covered by $n$ intervals of length $R'$.

        \item $H(t) \le 1 + \epsilon$ for all $t \in [-1,1]$. In fact, $H(t) = 1 \pm \epsilon$ for any $t$ with $|t| \le 1 - \Omega(\frac{1}{\kappa})$.
        
        \item $\|H x - x\|^2_{[-1,1]} \le \epsilon \|x\|^2_{[-1,1]}$ and $\|H x\|_{[-1,1]}^2 \ge (1-\epsilon) \|x\|_{[-1,1]}^2$.
        
        \item $\|H  x\|^2_{[-1,1]} \ge (1-\epsilon) \cdot \|H x\|_2^2$.
    \end{enumerate}

    Recall that $y=x+\eta$ is the observation. We define $w(t)=H(t) \cdot y(t)$ for $t \in [-1,1]$ and 0 otherwise. Property 4 implies Condition~\ref{cond:multiband_recovery} for the multiband signal $H \cdot x$. 
    
    Then we bound the noise. So $w(t)=H(t) \cdot x(t) + H(t) \cdot \eta(t)$ in the time window. By the second property and the third property of $(H,\wh{H})$, 
    \[
    \|H(t) \cdot \eta(t)\|_{[-1,1]}^2 \le (1+\epsilon)^2 \|\eta\|_{[-1,1]}^2 \le (1+\epsilon)^2 \epsilon \cdot \|x\|_{[-1,1]}^2 \le \frac{\epsilon (1+\epsilon)^2}{(1-\epsilon)} \cdot \|H x\|_{[-1,1]}^2.
    \]    
    Moreover, setting $w(t)=0$ for $t \notin [-1,1]$ only increases the noise by $\|H x\|_2^2 - \|H \cdot x\|^2_{[-1,1]} \le \frac{\epsilon}{1-\epsilon} \|H \cdot x\|^2_{[-1,1]}$ by the fourth property. 

    So $w(t)=H(t) \cdot x(t) +\eta'(t)$ for a noise function $\|\eta'\|_2^2 \le 3 \epsilon \cdot \|H \cdot x\|^2_{[-1,1]}$. Theorem~\ref{thm:learn_multiband} provides an interpolant $\tilde{x}$ with $\|\tilde{x} - H x\|_{[-1,1]}^2=O(\epsilon) \|H \cdot x\|_{[-1,1]}^2$. By the third property of $H$, this implies $\|\tilde{x} - x\|_{[-1,1]}^2=O(\epsilon) \cdot \|x\|_{[-1,1]}^2$.
\end{proof}

\bibliographystyle{alpha} 
\bibliography{bibFFT}

\section{AI Usage Statement}
In the early stage of this work, the authors proved the multiband approximation in Theorem~\ref{thm:learn_Fourier_sparse} and developed the learning algorithms for multiband signals in Theorem~\ref{inform:learn_multiband} by human reasoning. This led to a learning algorithm of Fourier-sparse signals within $\tilde{O}(k^2)$ samples and $\tilde{O}(k^2)$ time for Theorem~\ref{inform:improper_learning}. During this stage, the authors used Gemini 3.1 to assist with proving the error-shifting Lemma~\ref{lem:shift_one_freq}.

Then the authors used ChatGPT to assist with proving the statistical upper bound $\tilde{O}(k)$ in Theorem~\ref{inform:stat_bound}. After discovering the Haar-Fourier basis, the authors used ChatGPT to assist with improving results in the early stage and obtain the current version. The authors take full
responsibility for all content and arguments in this paper and for their
correctness.
\end{document}